%% file: jiang_540.tex
\documentclass[accepted,dvipsnames]{uai2023} 

\usepackage[british]{babel}

\usepackage{natbib} 
\usepackage{mathtools} 
\usepackage{amssymb}

\usepackage{floatrow}
\newfloatcommand{capbtabbox}{table}[][\FBwidth]

\usepackage{booktabs} 
\usepackage{tikz} 
\usepackage{subcaption}
\usepackage{tikz-network}
\usepackage{xspace}
\usepackage{color,soul}
\usepackage{amsmath}
\usepackage{algorithm}
\usepackage{algpseudocode}
\usepackage{xurl}
\usepackage{rotating}
\usetikzlibrary{automata,positioning,arrows}

\usepackage{thmtools}
\usepackage{thm-restate}
\usepackage{hyperref}
\usepackage{cleveref}

\pgfdeclarelayer{bg}    
\pgfsetlayers{bg,main}

\newcommand{\mathbbm}[1]{\text{\usefont{U}{bbm}{m}{n}#1}}
\let\oldReturn\Return
\renewcommand{\Return}{\State\oldReturn}
\newtheorem{definition}{Definition}
\newtheorem{proof}{Proof}

\newtheorem{theorem}{Theorem}
\newtheorem{lemma}{Lemma}
\newtheorem{corollary}{Corollary}

\def\H{\mathcal{H}}
\def\P{\mathcal{P}}
\def\L{\mathcal{L}}
\def\A{\mathcal{A}}
\def\T{\mathcal{T}}
\def\F{\mathcal{F}}
\def\U{\mathcal{U}}
\def\V{\mathcal{V}}
\def\Os{\mathcal{O}}
\def\M{\mathcal{M}}

\def\Fig{{Fig.}}
\def\sec{{Sec.}}
\def\LE{{LE}}

\makeatletter
\DeclareRobustCommand{\cev}[1]{%
  {\mathpalette\do@cev{#1}}%
}
\newcommand{\do@cev}[2]{%
  \vbox{\offinterlineskip
    \sbox\z@{$\m@th#1 x$}%
    \ialign{##\cr
      \hidewidth\reflectbox{$\m@th#1\vec{}\mkern4mu$}\hidewidth\cr
      \noalign{\kern-\ht\z@}
      $\m@th#1#2$\cr
    }%
  }%
}
\makeatother

\title{Bayesian Inference for Vertex-Series-Parallel Partial Orders}

\author[1]{Chuxuan (Jessie) Jiang}
\author[1]{Geoff K. Nicholls}
\author[2]{Jeong Eun Lee}
\affil[1]{%
    Department of Statistics\\
    University of Oxford\\
    United Kingdom
}
\affil[2]{%
    Department of Statistics\\
    University of Auckland\\
    New Zealand
}
  
\begin{document}
\maketitle

\begin{abstract}

Partial orders are a natural model for the social hierarchies that may constrain ``queue-like'' rank-order data. However, the computational cost of counting the linear extensions of a general partial order on a ground set with more than a few tens of elements is prohibitive. Vertex-series-parallel partial orders (VSPs) are a subclass of partial orders which admit rapid counting and represent the sorts of relations we expect to see in a social hierarchy. However, no Bayesian analysis of VSPs has been given to date. We construct a marginally consistent family of priors over VSPs with a parameter controlling the prior distribution over  VSP depth. The prior for VSPs is given in closed form. We extend an existing observation model for queue-like rank-order data to represent noise in our data and carry out Bayesian inference on ``Royal Acta'' data and Formula 1 race data. Model comparison shows our model is a better fit to the data than Plackett-Luce mixtures, Mallows mixtures, and ``bucket order'' models and competitive with more complex models fitting general partial orders.

\end{abstract}

\section{Introduction}\label{sec:intro}

Rank-order data are lists in which a set of elements are ranked. They are analysed in a wide range of areas, including decision support \citep{carlosequential}, medical research \citep{beerenwinkel2007conjunctive} and chemistry \citep{pavan2008scientific}. We classify ranking methods into two categories - total order ranking and partial order ranking.

\textit{Total order} models seek a ranking of the elements of the ground set (in our setting, the labels of a group of actors we want to rank) that is ``central'' to the rank-lists in the data. These models are suitable when we believe that an order relation exists between every pair of actors. The Mallows model \citep{mallows1957non}, the Plackett-Luce model \citep{plackett1975analysis,luce1959possible} and related mixture models are models for total orders. However, the real-world relations we are looking to recover may be weaker than a total order: perhaps relations between pairs of actors are not simply weak or uncertain, they don't actually exist. We expect this for some precedence relations that define some social hierarchies. 

If we want to learn social-order relations between actors by observing their behavior, then the elements of the model we fit should correspond to elements of reality: if relations are incomplete then we should fit a {\it partial order}. A partial order $h=\{[n],\prec_h\}$ is a (possibly incomplete) set of binary order relations $\prec_h$ over a ``ground set'' of actors with labels $[n]=\{1,\dots,n\}$. Our data are records of queues of actors constrained by a social hierarchy $h$, which is unknown. If we see enough queue realisations we can identify the hierarchy. In this setting the queue is just a \textit{linear extension} (\LE) of $h$, that is, a permutation of actors in $[n]$ that doesn't put an actor ahead of someone of higher precedence. 

Partial orders are widely used as a ranking summary tool, or to support efficient computation. For example, partial orders and LEs support efficient computation of marginals in Bayesian networks \citep{cano2011approximate,smail2018junction}. By contrast, in our work the partial order $h$ is the object of inference, so it is a parameter in the likelihood: the data are \LE s and the likelihood depends on the number of \LE s of $h$. Counting \LE s is an \#P-complete task \citep{brightwell1991counting}, so work to date in this setting either restricts the class of partial orders \citep{mannila2000global,gionis2006algorithms,mannila2008finding} to orders which admit fast counting or works with orders of manageable size \citep{beerenwinkel2007conjunctive,beerenwinkel12,nicholls122011partial,nicholls2022ts}. This approach does not scale well with $n$. We follow \cite{mannila2000global} and work with \textit{vertex-series-parallel} partial orders (VSPs). These orders are a sub-class of partial orders which can be formed by repeated series and parallel operations on smaller VSPs. They include \textit{bucket orders} \footnote{Actors are grouped in buckets - every actor is ordered with respect to actors in other groups, and any pair of actors in the same group are incomparible.} as a special case. \cite{valdes1979recognition} represent VSPs using binary decomposition trees (BDTs). These support counting in a time linear in $n$ \citep{wells1971elements} and scale to VSPs with hundreds of actors. 

VSPs are a well characterised combinatorial class \citep{wells1971elements,valdes1979recognition}. However, work on fitting VSPs to data is limited. \citet{mannila2000global} learn VSPs from \LE s by adapting a greedy search over VSPs. However, there is to date no Bayesian inference and hence no one has given a prior probability distributions over VSPs with good properties for inference. \cite{mannila2008finding} gave Bayesian inference for bucket orders, a subclass of VSP, and \cite{beerenwinkel12} for partial orders, a super-set that doesn't scale. 

\textbf{Contributions.} This is the first Bayesian inference for VSPs from \LE s and presents some useful new priors and likelihoods. VSPs are equivalent to ``transitively closed'' Directed Acyclic Graphs (DAGs); when we specify priors over objects of this sort we have to be careful to ensure the prior doesn't impose unwanted weighting and inconsistency. 

We specify a prior and give its probability mass function in a simple closed form. Our prior (\sec~\ref{sec:vsp-prior}) is marginally consistent. This property (defined in Definition~\ref{defn:mc_vsp_defn} below) is needed for the model to make sense in our setting. 
Our prior also represents the information available well: it is non-informative with respect to VSP depth, one of the most interesting summary statistics for a social hierarchy. 

Our new observation model (\sec~\ref{sec:QJ-B-defn-main}) generalises earlier models for observation noise in records for queue-like data and has a natural physical interpretation in terms of ``queue jumping'' and ``arriving late''. 

We give MCMC algorithms in Appendix~C 
which target the VSP posterior. We carry out model comparison with the Plackett-Luce and Mallows mixture models in Appendix~E.1. We further compare our model with a simple restriction to Bucket-Order models in Appendix~E.2 and we compare it with a more general partial order model \citep{nicholls2022ts} in Appendix~E.3. 

Finally, our reconstruction of relations between witnesses appearing in Royal Acta (\sec~\ref{sec:QJ-up}) is new. Historians are interested in these relations, but it wasn't possible to reconstruct them all till now as the partial orders were too big to count their LEs (\cite{nicholls2022ts} analyse a subset, working in a time-series setting; we give timing comparisons in Appendix~F). Our models are relevant for any ranking problem where relations may be partial: in Appendix~D.2 we  fit Formula~1 race results for the 2021 season. These data show the same preference for our model over other models.


\subsection{Background}\label{sec:background}


A partial order $h=\{V,\prec_h\}$ is a binary relation\footnote{The binary relation $\prec_h$ is both irreflexive (the relation $i \prec_h i$ does not exist) and transitive (if $i\prec_h j$ and $j\prec_h k$, then $i \prec_h k$), where $i,j,k\in [n]$ and $i\neq j\neq k$.} $\prec_h$ over a ``ground set'' of actors $V$. 
In our setting the actor labels are $V=[n]$ where $[n]=\{1,2,...,n\}$ or some subset. Two actors $i,j \in [n]$ are \textit{incomparable} $i\|_hj$, if neither $i\prec_h j$ nor $i \succ_h j$. Partial orders on $[n]$ are in one-to-one correspondence with transitively-closed DAGs $([n],E)$ with edges $E=\{\langle i,j\rangle\in [n]\times [n]: i\succ_h j\}$. Denote by $\H_{[n]}$ the set of all partial orders on actor labels $[n]$. 
Let $\P_{[n]}$ be the set of all permutations of $[n]$. A linear extension $l_h\in\P_{[n]}$ is a permutation of actors in $[n]$ that does not violate partial order $h$. See \Fig~\ref{fig:POex} for an example partial order\footnote{In this article, we visualise a partial order via its transitive reduction - this omits all edges implied by transitivity and is unique.} and its \LE s. We denote the set of all \LE s for partial order $h$ as $\L[h]$. A \textit{sub-order} $h[o]=(o,\prec_h)$ of a partial order $h\in \H_{[n]}$ restricts $h$ to a subset $o \subseteq [n],\ o=\{o_1,...,o_m\}$: all order relations in $h$ are inherited by $h[o]$ so its DAG representation $(o,E[o])$ has edges $E[o]=\{e\in E: e\in o\times o\}$; directed edges incident vertices in $[n]\setminus o$ are removed and all others remain.
A \textit{chain} of $h\in \H_{[n]}$ is a sub-order $h[o]$ that is also a total order. The \textit{length} of a chain is the number of nodes $|o|$ in the sub-order. The \textit{depth} $D(h)$ of a partial order is the length of its longest chain, with $1\le D(h)\le n$. 

    \begin{minipage}{\linewidth}
    \centering
    \begin{tikzpicture}[thick,scale=.8, every node/.style={scale=0.6}]
        \node[draw, circle, minimum width=.1cm] (1) at (0, 1) {$1$};
        \node[draw, circle, minimum width=.1cm] (2) at (-1, -0.5) {$2$};
        \node[draw, circle, minimum width=.1cm] (3) at (1, 0) {$3$};
        \node[draw, circle, minimum width=.1cm] (4) at (1, -1) {$4$};
        \node[draw, circle, minimum width=.1cm] (5) at (0, -2) {$5$};
        \node[draw, circle, minimum width=.1cm] (6) at (2.5, 1) {$1$};
        \node[draw, circle, minimum width=.1cm] (7) at (2.5, 0.25) {$2$};
        \node[draw, circle, minimum width=.1cm] (8) at (2.5, -0.5) {$3$};
        \node[draw, circle, minimum width=.1cm] (9) at (2.5, -1.25) {$4$};
        \node[draw, circle, minimum width=.1cm] (10) at (2.5, -2) {$5$};
        \node[draw, circle, minimum width=.1cm] (11) at (3.1, 1) {$1$};
        \node[draw, circle, minimum width=.1cm] (12) at (3.1, 0.25) {$3$};
        \node[draw, circle, minimum width=.1cm] (13) at (3.1, -0.5) {$2$};
        \node[draw, circle, minimum width=.1cm] (14) at (3.1, -1.25) {$4$};
        \node[draw, circle, minimum width=.1cm] (15) at (3.1, -2) {$5$};
        \node[draw, circle, minimum width=.1cm] (16) at (3.7, 1) {$1$};
        \node[draw, circle, minimum width=.1cm] (17) at (3.7, 0.25) {$3$};
        \node[draw, circle, minimum width=.1cm] (18) at (3.7, -0.5) {$4$};
        \node[draw, circle, minimum width=.1cm] (19) at (3.7, -1.25) {$2$};
        \node[draw, circle, minimum width=.1cm] (20) at (3.7, -2) {$5$};
        \draw[-latex] (1) -- (2);
        \draw[-latex] (2) -- (5);
        \draw[-latex] (1) -- (3);
        \draw[-latex] (3) -- (4);
        \draw[-latex] (4) -- (5);
        \draw[-latex] (6) -- (7);
        \draw[-latex] (7) -- (8);
        \draw[-latex] (8) -- (9);
        \draw[-latex] (9) -- (10);
        \draw[-latex] (11) -- (12);
        \draw[-latex] (12) -- (13);
        \draw[-latex] (13) -- (14);
        \draw[-latex] (14) -- (15);
        \draw[-latex] (16) -- (17);
        \draw[-latex] (17) -- (18);
        \draw[-latex] (18) -- (19);
        \draw[-latex] (19) -- (20);
    \end{tikzpicture}
    
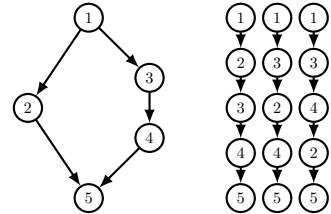
\captionof{figure}{(left) A partial order with 5 actors and depth 4 which is also a VSP, $v_0$ say, and (right) its three \LE s.}\label{fig:POex}
    \end{minipage}


The \textit{vertex-series-parallel partial orders (VSP)} on $[n]$ are a class of partial orders $\V_{[n]}\subset\H_{[n]}$ formed by repeated \textit{series} $\otimes$ and \textit{parallel} $\oplus$ operations. For partial orders $h_1$ and $h_2$, let $V(h_1)$ and $V(h_2)$ represent the ground sets of actors for $h_1$ and $h_2$ respectively (which we assume are disjoint).
\begin{itemize}[topsep=0pt,itemsep=0pt,partopsep=0pt, parsep=0pt]
    \item A \textit{series partial order}, $h = h_1 \otimes h_2$, is the union of all relations in $h_1$ and $h_2$, with additional relations $i \succ_h j$ if $i \in V(h_1)$ and $j\in V(h_2)$.
    \item A \textit{parallel partial order}, $h = h_1 \oplus h_2$, is the union of all relations in $h_1$ and $h_2$ with incomparability $i \|_h j$ if $i \in V(h_1)$ and $j\in V(h_2)$.
\end{itemize}
The set of VSPs $\V_{[n]}$ is defined recursively: if $|V(h)|=1$ then $h$ is a VSP; if $h_1$ and $h_2$ are VSPs then 
$h_1 \otimes h_2$ and $h_1 \oplus h_2$ are VSPs. \cite{valdes1979recognition} show that a partial order is a VSP if it does not contain the ``forbidden sub-graph'' (Appendix~G, \Fig~G.1) as a subgraph isomorphism. 

The partial order $v_0$ in \Fig~\ref{fig:POex} is a VSP. It can be constructed using the series and parallel operations in \Fig~\ref{fig:SP}.

\begin{minipage}{\linewidth}
    \centering
    \begin{tikzpicture}[thick,scale=.9, every node/.style={scale=0.6}]
        \node[draw, circle, minimum width=.1cm,fill=lightgray,label={[xshift=.2cm,yshift=-.1cm]\small +}] (1) at (-5, 2) {$3$};
        \node[minimum width=1cm] at (-4.5,2) {$\bigotimes$};
        \node[draw, circle, minimum width=.1cm,fill=lightgray,label={[xshift=.2cm,yshift=-.1cm]\small -}] (2) at (-4, 2) {$4$};
        \draw[-implies,double equal sign distance] (-3.5,2) -- (-3,2) node[midway,above] {$S$};
        \node[draw, circle, minimum width=.1cm] (3) at (-2.5, 2.5) {$3$};
        \node[draw, circle, minimum width=.1cm] (4) at (-2.5, 1.5) {$4$};
        \draw[-latex] (3) -- (4);
        \draw[dashed] (-2.1,3) -- (-2.1,1);
        \node[minimum width=1cm] at (-1.7,2) {$\bigoplus$};
        \node[draw, circle, minimum width=.1cm,fill=lightgray] (5) at (-1.2, 2) {$2$};
        \draw[-implies,double equal sign distance] (-0.7,2) -- (-.2,2) node[midway,above] {$P$};
        \node[draw, circle, minimum width=.1cm] (5) at (.3, 2) {$2$};
        \node[draw, circle, minimum width=.1cm] (6) at (0.8, 2.5) {$3$};
        \node[draw, circle, minimum width=.1cm] (7) at (0.8, 1.5) {$4$};
        \draw[-latex] (6) -- (7);
        \draw[dashed] (1.2,3) -- (1.2,1);
        \node[minimum width=1cm] at (1.6,2) {$\bigotimes$};
        \node[draw, circle, minimum width=.1cm,fill=lightgray,label={[xshift=.2cm,yshift=-.1cm]\small -}] (8) at (2.1, 2) {$5$};
        \draw[-implies,double equal sign distance] (1.6,1) -- (1.6,.5) node[midway,right] {$S$};
        \node[draw, circle, minimum width=.1cm] (8) at (0.9, -.5) {$2$};
        \node[draw, circle, minimum width=.1cm] (9) at (2.1, -.15) {$3$};
        \node[draw, circle, minimum width=.1cm] (10) at (2.1, -.85) {$4$};
        \node[draw, circle, minimum width=.1cm] (11) at (1.5, -1.3) {$5$};
        \draw[-latex] (9) -- (10);
        \draw[-latex] (10) -- (11);
        \draw[-latex] (8) -- (11);
        \draw[dashed] (0.5,0.5) -- (0.5,-1.5);
        \node[minimum width=1cm] at (0.1,-.5) {$\bigotimes$};
        \node[draw, circle, minimum width=.1cm,fill=lightgray,label={[xshift=.2cm,yshift=-.1cm]\small +}] (12) at (-.4, -.5) {$1$};
        \draw[-implies,double equal sign distance] (-.9,-.5) -- (-1.4,-.5) node[midway,above] {$S$};
        \node[draw, circle, minimum width=.1cm] (13) at (-3.1, -.5) {$2$};
        \node[draw, circle, minimum width=.1cm] (14) at (-1.9, -.15) {$3$};
        \node[draw, circle, minimum width=.1cm] (15) at (-1.9, -.85) {$4$};
        \node[draw, circle, minimum width=.1cm] (16) at (-2.5, -1.3) {$5$};
        \node[draw, circle, minimum width=.1cm] (17) at (-2.5, .3) {$1$};
        \draw[-latex] (13) -- (16);
        \draw[-latex] (15) -- (16);
        \draw[-latex] (14) -- (15);
        \draw[-latex] (17) -- (13);
        \draw[-latex] (17) -- (14);
    \end{tikzpicture}
    
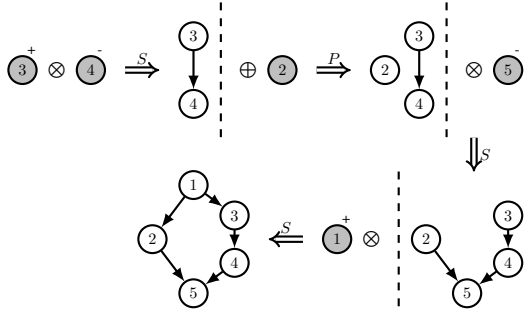
\captionof{figure}{One possible construction procedure for the VSP $v_0$ shown in \Fig~\ref{fig:POex}.}\label{fig:SP}
\end{minipage}


A VSP on $n$ actors can be parameterised as a Binary Decomposition Tree (BDT) \cite{valdes1979recognition} - a binary tree $t\in \T_{[n]}$ with $n$ leaves in which nodes have additional attributes (listed below) and edges are directed from the root to the leaves. Let $\F$ and $\A$ be the index sets for the $n$ leaves and $n-1$ internal nodes respectively, with $\F\cup\A=[2n-1]$. Each leaf node index corresponds to a unique actor in the VSP. It is convenient to distinguish between leaf nodes indices and the actor labels to which they correspond.
For each leaf node $i\in \F$, let $F_i(t)\in [n]$ give the actor label for the actor corresponding to that leaf node. Internal nodes $i\in \A$ are $S$ nodes if the subtrees rooted by their child nodes are merged in series, otherwise they are $P$ nodes and the subtrees are merged in parallel. Internal nodes with an $S$ label have an additional attribute indicating which of its child nodes is the ``upper child'': the subtree of this child node (indicated by a `+' and a red edge in \Fig~\ref{fig:tree_ex}) is stacked above the subtree rooted by the other child node (indicated by a `-'). As an example, the VSP $v_0$ in \Fig~\ref{fig:POex} can be represented by the BDT $t_0$ in \Fig~\ref{fig:tree_ex}. Let $S(t)\in [n-1]$ be the number of  $S$-nodes in tree $t$. 

A tree $t$ with edge set $E(t)$ is written $t=(F(t),E(t),L(t))$. Here $L(t)=\{L_i\}_{i\in \A}$ with $L_i(t)=(j,j')$ indicating that internal node $i$ is an $S$-node with child nodes $j,j'$ and the subtree rooted by $j$ is stacked above that rooted by $j'$, and $L_i(t)=\emptyset$ if $i$ is a $P$-node. The map from a BDT to the VSP $v: \T_{[n]} \to \V_{[n]}$ is not bijective: for a VSP $v\in\V_{[n]}$, there may exist many BDTs $t\in\T_{[n]}$ which represent it. Let $t(v)=\{t\in \T_{[n]}: v(t)=v\}$ give the set of BDTs representing VSP $v\in \V_{[n]}$. 

\begin{minipage}{\linewidth}
    \centering
    \begin{tikzpicture}[thick,scale=.58, every node/.style={scale=0.45}]
        \node[draw, circle, minimum width=1cm,fill=pink] (1) at (0, 2) {$S$};
        \node[draw, circle, minimum width=1cm] (2) at (-1, 1) {$1+$};
        \node[draw, circle, minimum width=1cm,fill=pink] (3) at (1, 1) {$S-$};
        \node[draw, circle, minimum width=1cm,fill=cyan] (4) at (0, 0) {$P+$};
        \node[draw, circle, minimum width=1cm] (5) at (2, 0) {$5-$};
        \node[draw, circle, minimum width=1cm] (6) at (-1, -1) {$2$};
        \node[draw, circle, minimum width=1cm,fill=pink] (7) at (1, -1) {$S$};
        \node[draw, circle, minimum width=1cm] (8) at (0, -2) {$3+$};
        \node[draw, circle, minimum width=1cm] (9) at (2, -2) {$4-$};
        \draw[-latex] (1) -- (3);
        \draw[-latex][red] (1) -- (2);
        \draw[-latex][red] (3) -- (4);
        \draw[-latex] (3) -- (5);
        \draw[-latex] (4) -- (6);
        \draw[-latex] (4) -- (7);
        \draw[-latex][red] (7) -- (8);
        \draw[-latex] (7) -- (9);
    \end{tikzpicture}
    
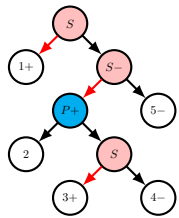
\captionof{figure}{A BDT $t_0$ representing $v_0$ in \Fig~\ref{fig:POex}, so that $v(t_0) = v_0$. Red edges and `$+$' signs indicate the upper child.}\label{fig:tree_ex}
\end{minipage}


\cite{brightwell1991counting} show that counting the number of \LE s of a partial order is a \#P-complete problem. However, the subclass of VSP partial orders admits fast counting. \cite{wells1971elements} gives
\begin{align}
    |\L(h_1 \otimes h_2)| = & |\L(h_1)\| \L(h_2)| \label{eq:le_s}\\
    |\L(h_1 \oplus h_2)| = & |\L(h_1)\|\L(h_2)|{{|V(h_1)|+|V(h_2)|\choose |V(h_1)|}}\label{eq:le_p}
\end{align}
where $|V(h_1)|$ and $|V(h_2)|$ give the number of actors in $h_1$ and $h_2$. This may be evaluated recursively in $O(n)$ steps.

In the following we make use of one more representation of a VSP: the \textit{Multi-Decomposition Tree} (MDT). These trees are obtained by collapsing edges which connect internal nodes of the same $S/P$-type in the BDT, as in \Fig~\ref{fig:multitree-ex}. Let $\M_{[n]}$ be the set of all MDTs with $n$ distinguisable leaves.
A formal definition is given in Appendix~A.3.

    \begin{minipage}{\linewidth}
    \centering
    \begin{tikzpicture}[thick,scale=.55, every node/.style={scale=0.45}]
        \node[draw, circle, minimum width=1cm,fill=pink] (1) at (-2, 2) {$S$};
        \node[draw, circle, minimum width=1cm,fill=pink] (2) at (-3, 1) {$S+$};
        \node[draw, circle, minimum width=1cm,fill=pink] (3) at (-1, 1) {$S-$};
        \node[draw, circle, minimum width=1cm] (4) at (0, 0) {$6-$};
        \node[draw, circle, minimum width=1cm] (5) at (-1.5, 0) {$5+$};
        \node[draw, circle, minimum width=1cm] (6) at (-4, 0) {$1+$};
        \node[draw, circle, minimum width=1cm,fill=cyan] (7) at (-2.5, 0) {$P-$};
        \node[draw, circle, minimum width=1cm,fill=cyan] (8) at (-3.2, -1.1) {$P$};
        \node[draw, circle, minimum width=1cm] (9) at (-1.8, -1.1) {$4$};
        \node[draw, circle, minimum width=1cm] (10) at (-3.9, -2.2) {$2$};
        \node[draw, circle, minimum width=1cm] (11) at (-2.5, -2.2) {$3$};
        \node[draw, circle,fill=pink, minimum width=1cm] (12) at (3.5, 1.8) {$S$};
        \node[draw, circle, minimum width=1cm,label={\small 1}] (13) at (2, 0) {$1$};
        \node[draw, circle, minimum width=1cm,fill=cyan,label={\small 2}] (14) at (3, 0) {$P$};
        \node[draw, circle, minimum width=1cm,label={\small 3}] (15) at (4, 0) {$5$};
        \node[draw, circle, minimum width=1cm,label={\small 4}] (16) at (5, 0) {$6$};
        \node[draw, circle, minimum width=1cm] (17) at (2, -2) {$2$};
        \node[draw, circle, minimum width=1cm] (18) at (3, -2) {$3$};
        \node[draw, circle, minimum width=1cm] (19) at (4, -2) {$4$};
        \draw[-latex] (1) -- (3);
        \draw[-latex][red] (1) -- (2);
        \draw[-latex][red] (2) -- (6);
        \draw[-latex] (2) -- (7);
        \draw[-latex] (3) -- (4);
        \draw[-latex][red] (3) -- (5);
        \draw[-latex] (8) -- (10);
        \draw[-latex] (8) -- (11);
        \draw[-latex] (7) -- (8);
        \draw[-latex] (7) -- (9);
        \draw[-latex] (12) -- (13);
        \draw[-latex] (12) -- (14);
        \draw[-latex] (12) -- (15);
        \draw[-latex] (12) -- (16);
        \draw[-latex] (14) -- (17);
        \draw[-latex] (14) -- (18);
        \draw[-latex] (14) -- (19);
        \draw[-implies,double equal sign distance] (.5,0) -- (1.5,0) node[midway,above] {$m_\T(\cdot)$};;
    \end{tikzpicture}
    
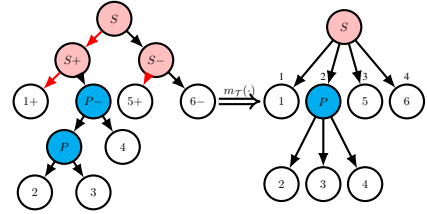
\captionof{figure}{An example BDT $t_1$ (left) and its corresponding MDT $m_1$ (right). The child-nodes of any $S$-node in the MDT are numbered to give the order in which their subtrees are stacked by the BDT. 
    }\label{fig:multitree-ex}
\end{minipage}

\cite{valdes78} has shown that MDTs are one-to-one with VSPs, so all the BDTs in $t(v)$ representing the VSP $v$ must ``collapse down'' to give the same MDT.
    For $m\in\M_{[n]}$ we write $v=v(m)$ for the map to VSPs (relations between any pair of actors in the VSP are simply given by the type of their Most Recent Common Ancestor (MRCA) in $m$).  
    Let $m_\V(v)=\{m\in\M_{[n]}: v(m)=v\}$ be the set of MDTs representing $v\in\V_{[n]}$. 
    
    \begin{lemma}\label{lemma:mdt-vsp}
    The map $m_\V: \V_{[n]}\to \M_{[n]}$ is bijective (so that $|m_\V(v)|=1$). See \cite{valdes78} for proof and \cite{valdes1979recognition} for further discussion.
    \end{lemma}
    
\section{VSP Prior}\label{sec:vsp-prior}

In this section we give a marginally consistent prior $\pi_{\V_{[n]}}(v|q)$ over VSPs on actors in $[n]$, controling the distribution over VSP-depth. We begin by defining a prior probability distribution $\pi_{\T_{[n]}}(t|q)$ over BDTs $t\in \T_{[n]}$. 

Our prior on $\T_{[n]}$ has a uniform distribution over trees $([2n-1],E(t))$ with distinguishable leaves. Internal nodes are labelled $S$ with probability $q$ and otherwise $P$. We choose an ``upper child'' for each $S$ node at random from its two child nodes, so we have
\begin{equation}\label{eq:tree-prior}
        \pi_{\T_{[n]}}(t|q) = \frac{1}{|\T_{[n]}|} \left(\frac{q}{2}\right)^{S(t)} (1-q)^{n-S(t)-1},
    \end{equation}
    where $S(t)$ is the number of $S$-nodes, $|\T_{[n]}|=(2n-3)!!\equiv (2n-3)\cdot (2n-5)...3\cdot 1$ is the number of binary tree topologies with $n$ distinguishable leaves, and the types of the $n-1$ internal nodes are independent with a factor $2^{-S(t)}$ for the stacking order of the children of $S$-nodes.

We get the prior on VSPs $v\in \V_{[n]}$ by summing over all BDTs that represent $v$,
\begin{equation}\label{eq:vsp_prior}
    \pi_{\V_{[n]}}(v|q)=\sum_{t\in t(v)} \pi_{\T_{[n]}}(t|q)
\end{equation}
This simple choice, based on a uniform distribution over tree topologies, determines a prior for VSPs that represents the prior knowledge we want to impose in our setting. If a social hierarchy is built up by making comparisons between groups of people, based for example on their profession, then it will be a VSP. Secondly, the unknown true depth of the social hierarchy we are trying to reconstruct (which is the length of the longest chain in the VSP) is a feature of particular interest, so we don't want the prior to strongly inform depth. 
We choose a prior distribution over $q$ so that the marginal distribution $\pi_{\V_{[n]}}(v)$ gives a reasonably flat prior distribution for depth $D(v)$ (see Appendix~H and Fig.H.1).

We assume that relations between two actors are determined by (unknown) properties intrinsic to those actors (for example, their professions, or ancestry). If that is true then the presence or absence of a third actor should not affect the relations between the first two. It is not straightforward to get this property {\it and} transitivity. If two actors $1\| 2$ are unordered and we add actor $3$ with relations $1\succ 3$ and $3\succ 2$ then $1\succ 2$ by transitivity: the presence of actor $3$ changes the relation between actors $1$ and $2$. Random VSPs can be built up in many different ways (that is, they are represented by many different BDTs), so we want the prior probability that $1\succ_w 2$ in a random VSP $w\sim \pi_{\V_{[2]}}$ to be the same as the prior probability that $1\succ_{v} 2$ in a random VSP $v\sim \pi_{\V_{[3]}}$. This adds a consistency restriction on any family of prior distributions $\pi_{\V_[n]},\ {n\ge 1}$ we write down.

A family of priors like $\pi_{\T_{[n]}}(t|q)$ or $\pi_{\V_{[n]}}(v|q),\ n\ge 1$ is \emph{marginally consistent} (also known as \emph{projective}) if every marginal of every distribution in the family is also in the family. Marginal consistency is not a property we get for free from the axioms of probability: the uniform distribution on partial orders $h\sim \mathcal{U}(\H_{[n]})$ is not consistent: there are 3 partial orders on the labels $\{1,2\}$ and 19 on $\{1,2,3\}$; since 19 is not divisible by 3, the probability for $1\succ_h 2$ in $h\sim \mathcal{U}(\H_{[2]})$ cannot equal the marginal probability for $1\succ_g 2$ in $g\sim \mathcal{U}(\H_{[3]})$.

\begin{definition}[Marginal consistency]\label{defn:mc_vsp_defn}
Let $\Os_{[n]}=\{o\subseteq [n]: |o|>0\}$ be the set of all subsets of $[n]$ with at least one element. The family of VSP priors $\pi_{\V_o}(v|q),\ o\in \Os_{[n]},\ n\ge 1$ is marginally consistent if, for all $n\ge 1$ and all $o,\tilde o\in \Os_{[n]}$ with $o\subseteq \tilde o$, distributions in the family satisfy
\begin{equation}\label{MarC}
\pi_{\V_o}(w|q)=\sum_{\substack{v\in\V_{\tilde o}\\ v[o]=w}} \pi_{\V_{\tilde o}}(v|q)\quad \mbox{for all $w\in \V_o$}.
\end{equation}
\end{definition}
If marginal consistency holds for all $q$ then it holds for marginals $\pi_o(w)$ by taking expectations over $q$ in (\ref{MarC}).

The following Theorem is our first main result: we give a closed form expression for the prior for a VSP (we calculate the sum in (\ref{eq:vsp_prior})) and show that the family of priors is marginally consistent. 
For $v\in\V_{[n]}$, let $t\in t(v)$ be some tree representing $v$. Partition the internal nodes $\A$ of $t$ into 
    \textit{$S$-clusters} $C_k^{(S)},\ k=1,...,K_S$ and \textit{$P$-clusters} $C_k^{(P)}, k=1,...,K_P$. An $S$-cluster is a maximal set of internal nodes of type $S$ which are connected by edges in $E(t)$ and corresponds to a node in the MDT-representation. The $P$-clusters are defined similarly. 
    We will see (in Appendix~A.2, proof of Proposition~5) that two BDTs representing the same VSP have the same numbers of $S$ and $P$ clusters, with the same sizes.
\begin{theorem}\label{thm:vsp-prior}
    The family, $\pi_{\V_o}(v|q),\ o\in \Os_{[n]}\ n\ge 1$, of VSP priors is marginally consistent. The probability distribution over VSPs $v\in\V_{[n]}$ in (\ref{eq:vsp_prior}) is
    \begin{align}\label{eq:vsp-prior}
    \pi_{\V_{[n]}}(v|q)&\!=\!\pi_{\T_{[n]}}(t|q)\prod_{k=1}^{K_P} (2|C_k^{(P)}|\!-\! 1)!!\!\prod_{k'=1}^{K^S} \mathcal{C}_{|C_{k'}^{(S)}|}
\end{align}
where $t$ may be taken to be any tree $t\in t(v)$ with $P$- and $S$-clusters defined above, $\pi_{\T_{[n]}}(t|q)$ is given in (\ref{eq:tree-prior}) and
    \begin{equation}\label{eq:catalan}
        \mathcal{C}_{s} = \frac{1}{s+1}{2s \choose s},\quad s\ge 0
    \end{equation} 
    is the $s$'th Catalan number \citep{Weisstein02}.
\end{theorem}

\begin{proof}[Theorem~\ref{thm:vsp-prior}]
    The proof of Theorem \ref{thm:vsp-prior} is given in two parts in Appendix~A. In Proposition~3 
    in Appendix~A.1 we show that the family of tree-priors $\pi_{\T_{[n]}}(t|q),\ o\in\Os_{[n]},\ n\ge 1$ is marginally consistent. This result is used in Proposition~4 
    in A.1 to show that VSPs are marginally consistent - the first part of Theorem~\ref{thm:vsp-prior}. 
    
    The proof of the second part is given in Appendix~A.2. 
    We show in Proposition~5
    that all trees $t\in t(v)$ have equal values of $\pi_{\T_{[n]}}(t|q)$, so that $\pi_{\V_{[n]}}(v|q)=|t(v)|\pi_{\T_{[n]}}(t|q)$ for any $t\in t(v)$. This is straightforward, as they must all collapse down to the same MDT. Finally, in Proposition~6,
    we give a formula for $|t(v)|$. We count the number of BDTs that collapse down to a given MDT.
    Any $P$-cluster $C_k^P$ of a BDT corresponds to a $P$-node in its MDT and covers a small sub-tree of the BDT representing an empty partial order on its $|C_k^P|+1$ labeled leaves. It can be replaced in the BDT by any sub-tree representing the empty partial order without changing the MDT, and there are $(2|C_k^{(P)}|\!-\! 1)!!$ such trees. Similarly, any $S$-cluster $C_k^S$ corresponds to a $S$-node in the MDT and covers a sub-tree of the BDT representing a total order on its leaves. It can be replaced in the BDT by any sub-tree representing the same total order.
    The Catalan numbers enter because $\mathcal{C}_{s-1}$ gives the number of BDTs representing a total order on $s$ elements (see proof Proposition~6).
    This last result is new, gives (\ref{eq:vsp-prior}) and completes the proof of Theorem~\ref{thm:vsp-prior}. 
\end{proof}

Theorem~\ref{thm:vsp-prior} gives the prior for a VSP in terms of the prior for one of the BDTs that represent that VSP. We can also parameterise VSPs using MDTs and this leads to the second MCMC scheme given in Appendix~C.2. 

\begin{corollary}\label{cor:mdt-prior}
For $m\in \M_{[n]}$ with internal nodes $\A$, let $c_i$ give the number of children of node $i\in A$ and let $P(m)=\{i\in\A: L_i(m)=\emptyset\}$ and $S(m)=\A\setminus P(m)$ give the sets of $P-$ and $S-$node labels. The prior for VSPs given in (\ref{eq:vsp-prior}) is equivalently a prior for MDTs,
    \begin{align}\label{eq:mdt-prior}
    \pi_{\V_{[n]}}(v(m)|q)&=\pi_{\M_{[n]}}(m|q)\\&=\frac{1}{(2n-3)!!}\prod_{i\in P(m)}(1-q)^{c_i-1}(2c_i-3)!!\nonumber \\
    &\quad \times\quad\prod_{j\in S(m)} \left(\frac{q}{2}\right)^{c_j-1}\mathcal{C}_{c_j-1}.\nonumber
    \end{align}
\end{corollary}

\begin{proof}[Corollary~\ref{cor:mdt-prior}]
Substitute (\ref{eq:tree-prior}) into (\ref{eq:vsp-prior}) and note a tree with $c_i$ leaves has $c_i-1$ internal nodes. This result gives a convenient representation for prior evaluation. 
\end{proof}

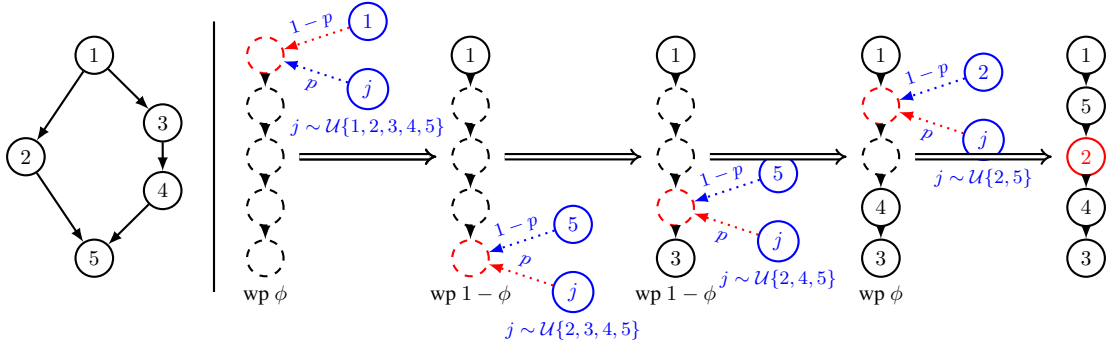
\begin{figure*}
\begin{minipage}{\linewidth}
    \centering
    \begin{tikzpicture}[thick,scale=.9, every node/.style={scale=0.8}]
        \node[draw, circle, minimum width=.6cm] (1) at (0, 1) {$1$};
        \node[draw, circle, minimum width=.6cm] (2) at (-1, -0.5) {$2$};
        \node[draw, circle, minimum width=.6cm] (3) at (1, 0) {$3$};
        \node[draw, circle, minimum width=.6cm] (4) at (1, -1) {$4$};
        \node[draw, circle, minimum width=.6cm] (5) at (0, -2) {$5$};
        \draw [solid] (1.75,-2.5) -- (1.75,1.5);
        \node[draw, circle, dashed, color = red, minimum width=.6cm] (6) at (2.5, 1) {};
        \node[draw, circle, dashed, minimum width=.6cm] (7) at (2.5, 0.25) {};
        \node[draw, circle, dashed, minimum width=.6cm] (8) at (2.5, -0.5) {};
        \node[draw, circle, dashed, minimum width=.6cm] (9) at (2.5, -1.25) {};
        \node[draw, circle, dashed, minimum width=.6cm, label=below:{wp $\phi$}] (10) at (2.5, -2) {};
        \node[draw, circle, color=blue, minimum width=.6cm] (21) at (4, 1.5) {$1$};
        \node[draw, circle, color=blue, minimum width=.6cm,label=below:\textcolor{blue}{\small $j\sim \mathcal{U}\{1,2,3,4,5\}$}] (22) at (4, .5) {$j$};
        \node[draw, circle, minimum width=.6cm] (11) at (5.5, 1) {$1$};
        \node[draw, circle, dashed, minimum width=.6cm] (12) at (5.5, 0.25) {};
        \node[draw, circle, color=blue, minimum width=.6cm] (23) at (7, -1.5) {$5$};
        \node[draw, circle, color=blue, minimum width=.6cm,label=below:\textcolor{blue}{\small $j\sim \mathcal{U}\{2,3,4,5\}$}] (24) at (7, -2.5) {$j$};
        \node[draw, circle, dashed, minimum width=.6cm] (13) at (5.5, -0.5) {};
        \node[draw, circle, dashed, minimum width=.6cm] (14) at (5.5, -1.25) {};
        \node[draw, circle, dashed, minimum width=.6cm,color=red, label=below:{wp $1-\phi$}] (15) at (5.5, -2) {};
        \node[draw, circle, minimum width=.6cm] (25) at (8.5, 1) {$1$};
        \node[draw, circle, dashed, minimum width=.6cm] (26) at (8.5, 0.25) {};
        \node[draw, circle, color=blue, minimum width=.6cm] (30) at (10, -.75) {$5$};
        \node[draw, circle, color=blue, minimum width=.6cm,label=below:\textcolor{blue}{\small $j\sim \mathcal{U}\{2,4,5\}$}] (31) at (10, -1.75) {$j$};
        \node[draw, circle, dashed, minimum width=.6cm] (27) at (8.5, -0.5) {};
        \node[draw, circle, dashed, color=red, minimum width=.6cm] (28) at (8.5, -1.25) {};
        \node[draw, circle, minimum width=.6cm, label=below:{wp $1-\phi$}] (29) at (8.5, -2) {$3$};
        \node[draw, circle, minimum width=.6cm] (32) at (11.5, 1) {$1$};
        \node[draw, circle, dashed, color=red, minimum width=.6cm] (33) at (11.5, 0.25) {};
        \node[draw, circle, color=blue, minimum width=.6cm] (37) at (13, 0.75) {$2$};
        \node[draw, circle, color=blue, minimum width=.6cm,label=below:\textcolor{blue}{\small $j\sim \mathcal{U}\{2,5\}$}] (38) at (13, -.25) {$j$};
        \node[draw, circle, dashed, minimum width=.6cm] (34) at (11.5, -0.5) {};
        \node[draw, circle,minimum width=.6cm] (35) at (11.5, -1.25) {$4$};
        \node[draw, circle, minimum width=.6cm, label=below:{wp $\phi$}] (36) at (11.5, -2) {$3$};
        \node[draw, circle, minimum width=.6cm] (16) at (14.5, 1) {$1$};
        \node[draw, circle, minimum width=.6cm] (17) at (14.5, 0.25) {$5$};
        \node[draw, circle, color=red, minimum width=.6cm] (18) at (14.5, -0.5) {$2$};
        \node[draw, circle, minimum width=.6cm] (19) at (14.5, -1.25) {$4$};
        \node[draw, circle, minimum width=.6cm] (20) at (14.5, -2) {$3$};
        \node at ($(27)!.7!(18)$) {\ldots};
        \draw[-latex] (1) -- (2);
        \draw[-latex, color=red, dotted] (21) -- (6) node[midway,above,sloped] {\textcolor{blue}{\small{$1-p$}}};
        \draw[-latex, color=blue, dotted] (22) -- (6) node[midway,below,sloped] {\textcolor{blue}{\small{$p$}}};
        \draw[-latex, color=blue, dotted] (23) -- (15) node[midway,above,sloped] {\textcolor{blue}{\small{$1-p$}}};
        \draw[-latex, color=red, dotted] (24) -- (15) node[midway,above,sloped] {\textcolor{blue}{\small{$p$}}};
        \draw[-implies,double equal sign distance] (3,-.5) -- (5,-.5);
        \draw[-latex, color=blue, dotted] (30) -- (28) node[midway,above,sloped] {\textcolor{blue}{\small{$1-p$}}};
        \draw[-latex, color=red, dotted] (31) -- (28) node[midway,below,sloped] {\textcolor{blue}{\small{$p$}}};
        \draw[-latex, color=blue, dotted] (37) -- (33) node[midway,above,sloped] {\textcolor{blue}{\small{$1-p$}}};
        \draw[-latex, color=red, dotted] (38) -- (33) node[midway,below,sloped] {\textcolor{blue}{\small{$p$}}};
        \draw[-implies,double equal sign distance] (3,-.5) -- (5,-.5);
        \draw[-implies,double equal sign distance] (6,-.5) -- (8,-.5);
        \draw[-implies,double equal sign distance] (9,-.5) -- (11,-.5);
        \draw[-implies,double equal sign distance] (12,-.5) -- (14,-.5);
        \draw[-latex] (2) -- (5);
        \draw[-latex] (1) -- (3);
        \draw[-latex] (3) -- (4);
        \draw[-latex] (4) -- (5);
        \draw[-latex] (6) -- (7);
        \draw[-latex] (7) -- (8);
        \draw[-latex] (8) -- (9);
        \draw[-latex] (9) -- (10);
        \draw[-latex] (11) -- (12);
        \draw[-latex] (12) -- (13);
        \draw[-latex] (13) -- (14);
        \draw[-latex] (14) -- (15);
        \draw[-latex] (16) -- (17);
        \draw[-latex] (17) -- (18);
        \draw[-latex] (18) -- (19);
        \draw[-latex] (19) -- (20);
        \draw[-latex] (25) -- (26); 
        \draw[-latex] (26) -- (27); 
        \draw[-latex] (27) -- (28); 
        \draw[-latex] (28) -- (29); 
        \draw[-latex] (32) -- (33); 
        \draw[-latex] (33) -- (34); 
        \draw[-latex] (34) -- (35); 
        \draw[-latex] (35) -- (36); 
    \end{tikzpicture}
    \captionof{figure}{One example list simulation process from the VSP $v_0$ (left) via the QJ-B observation model. The simulated list is displayed on the right. }\label{fig:QJB}
\end{minipage}
\end{figure*}

\section{BI-DIRECTIONAL QUEUE-JUMPING OBSERVATION MODEL}\label{sec:QJ-B-defn-main}

Our data is a collection of $N$ lists. For $j\in [N]$ let $o_j\subseteq [n],\ o_j=\{o_1,...,o_{n_j}\}$ be the actors present when the $j$'th ranking list was observed and let $y_j\in \P_{o_j},\ y_j=(y_{j,1},...,y_{j,n_j})$ be the observed list, just an ordered version of $o_j$. Let $y=(y_1,...,y_N)$ be the list of lists. The `queue-based' observation model given in \cite{nicholls122011partial} models list data as a realisation of a random queue constrained to put higher status individuals before those of lower status. In this model the queue is dynamic. It forms and then unconstrained pairs of actors swap places at random. If this process reaches equilibrium before the queue is read off then the resulting list is a uniform draw from the linear extensions of the constraining social hierarchy \citep{Karzanov91}. 
In this noise-free model $y_j\sim \U(\L[v[o_j]])$ independently for $j\in [N]$.

It is unlikely the observations are ``error free''. In a ``queue-jumping'' model (QJ-U, see Appendix B.1
and \cite{nicholls122011partial} for details) the queue is read from the top: with probability $p\in [0,1]$ the ``next'' person in the queue is drawn at random from those remaining, ignoring the social hierarchy; otherwise they are the first person in the remaining \LE. 
The queue can also be read from the bottom up. In this model (QJ-D) actors fall down the queue. We think of these events as actors arriving while the queue is being read.

We would like to have a queue-based model in which displacement in both directions is possible. The resulting ``bi-directional queue-jumping'' model (QJ-B) is not simply a mixture of QJ-U and QJ-D, as it allows displacement in both directions within a single realisation. The cost of evaluating a QJ-B likelihood is exponential in $n$. However, for the application in Section~\ref{sec:acta-data-analysis-outline} there is a subset of actors (bishops) known a priori to appear as a group. Separate modelling of this manageable subset ($n\simeq 20$) is well-motivated. Although QJ-B cannot be evaluated for a general partial order (counting \LE s is prohibitive) it is fine for a VSP. 

Like QJ-U, QJ-B ranks by repeated selection. \Fig~\ref{fig:QJB} provides an example QJ-B list-realisation from VSP $v_0$. A generic list $x\in \P_{[n]}$ is built up from both ends (see Appendix B.2). Let $z\in\{0,1\}^{n-1}$ with $z_k\sim Bern(\phi)$. Here $z_k=0$ indicates the $k$'th actor to be added to the list was placed bottom-up in the QJ-D model and $z_k=1$ indicates they were placed top-down in the QJ-U model. In \Fig~\ref{fig:QJB}, $z=(1,0,0,1)$. If we let $U_0=0$ then $U_k=U_{k-1}+z_k$ gives the number of places filled from the top after the $k$'th actor has arrived, so if $z_k=1$ then the $k$'th actor was placed into position $i_k=U_k$ in $x$. Similarly, if $D_0=n+1$ then $D_k=D_{k-1}-(1-z_k)$ tracks places filled from the bottom and gives the placement index $i_k=D_k$ in $x$ when $z_k=0$, so 
    $i_k=z_k U_k+(1-z_k)D_k$
gives the position in $x$ into which the $k$'th actor was added. If $z=(1,0,0,1)$, then $(i_1,...,i_4)=(1,5,4,2)$ (and $i_5=3$, the only remaining place).

\begin{definition}[Bi-Directional Queue-Jumping Model]
    Let $L_T(v)=|\L[v]|$ be the number of \LE s of VSP $v\in\V_{ [n]}$ and for $i\in [n]$ let $T_i(v)=|\{l\in\L[v]:l_1=i\}|$ give the number of \LE s with actor $i$ at the top. Let $B_i(v)=|\{l\in\L[v]:l_n=i\}|$ give the number of \LE s with actor $i$ at the bottom. If $z\in \{0,1\}^{n-1}$ is given then $i_k=i_k(z), k=1,...,n$ is given above. The observation model for QJ-B for a list $x\in\P_{[n]}$ given $z$ is
    \begin{align*}
       &Q_{bi}(x|z,v,p,\phi)\!=\!\prod_{k=1}^{n-1}\!
       [\phi\mathbbm{1}_{\{z_k=0\}}
       Q_{bi}(x_{i_k}|x_{i_{1:k-1}},z_k,v,p)\\
       &\quad +\quad (1-\phi)\mathbbm{1}_{\{z_k=1\}}
       Q_{bi}(x_{i_k}|x_{i_{1:k-1}},z_k,v,p)],
       \end{align*}
       where
       \begin{align*}
       &Q_{bi}(x_{i_k}|x_{i_{1:k-1}},z_k=0,v,p)=\\
       &\frac{p}{n-k+1}+(1-p)\frac{T_{x_{i_k}}(v[x_{[n]\backslash\{i_1,\dots,i_{k-1}\}}])}{L_T(v[x_{[n]\backslash\{i_1,\dots,i_{k-1}\}}])},
       \end{align*}
       \begin{align*}
       & Q_{bi}(x_{i_k}|x_{i_1:k-1},z_k=1,v,p)=\\
       &\frac{p}{n-k+1}+(1-p)\frac{B_{x_{i_k}}(v[x_{[n]\backslash\{i_1,\dots,i_{k-1}\}}])}{L_T(v[x_{[n]\backslash\{i_1,\dots,i_{k-1}\}}])},
       \end{align*}
       and marginally,
    \begin{align}\label{eq:QJ-bidirectn}
        Q_{bi}(x|v,p,\phi)=\!\!\!\!\!\!\sum_{z\in\{0,1\}^n}\!\!\!\! Q_{bi}(x|z,v,p,\phi) p(z|\phi)
    \end{align}
    where $p(z|\phi)=\phi^{\sum_i z_i}(1-\phi)^{n-\sum_i z_i}$.
\end{definition}
We give a generative model realising $x\sim Q_{bi}$ in Appendix B.2. 
This distribution reduces to $Q_{up}$/QJ-U in Appendix B.1 
when $\phi=1$ (and $Q_{down}$/QJ-D when $\phi=0$). We use this nesting to investigate whether QJ-U or QJ-D or QJ-B fits the data better. This is of interest in our application as different error types correspond to obvious physical mechanisms (downwards displacement may be ``arrived late'' and upwards displacement may be ``my friend the King is present''). When $p=0$ this is the noise free model for every $\phi\in [0,1]$, so $\phi$ is not identifiable in the noise free setting.  

Counting \LE s of a VSP (evaluating $L_T(v)$ etc) is $O(n)$ so the computational complexity for naive evaluation of $Q_{bi}$ using (\ref{eq:QJ-bidirectn}) is $O(n^2 2^n)$. We used a recursion (Algorithm~B.3
) of computational complexity 
$O(n2^n)$. This avoids repeated evaluation of \LE-counts for the same suborders and (by Proposition~7 
in Appendix~B.3
) evaluates $Q_{bi}$.







\section{SUMMARISING THE VSP POSTERIOR}\label{sec:posterior}

Bayesian inference is straightforward in principle given an explicit prior distribution over VSPs and an observation model $Q=Q_{up}$ or $Q=Q_{bi}$ for our $N$ ranking-lists. We can represent a VSP as a MDT (since the mapping is one-to-one) or carry out Bayesian inference on the latent space of BDTs $t\in \T_{[n]}$ and use the fact that they marginalise to MDTs. We present the posteriors for BDT and VSP. Let $\psi=(p,\phi)$ for QJ-B and $\psi=p$ for QJ-U.

The posterior for the BDT $t\in \T_{[n]}$ is 
    \begin{equation}\label{eq:post-BDT}
        \pi_{\T_{[n]}}(t,q,\psi|y)\propto \pi_{\T_{[n]}}(t|q)\pi(q,\psi) Q(y|v(t),\psi)
    \end{equation}
    The posterior distribution for the VSP $v\in \V_{[n]}$ is
    \begin{equation}\label{eq:post-VSP}
        \pi_{\V_{[n]}}(v,q,\psi|y) \propto \pi_{\V_{[n]}}(v|q)\pi(q,\psi) Q(y|v,\psi),
    \end{equation}
    where we use the equivalent MDT posterior with prior given in Corollary~\ref{cor:mdt-prior} 
    for VSPs in $(\ref{eq:post-VSP})$.

\begin{restatable}[Posterior Marginals]{prop}{posteriors}
\label{prop:posteriors}
Sampling the BDT posterior $(t,q,\psi)\sim \pi_{\T_{[n]}}(\cdot|y)$ gives samples $(v(t),q,\psi)\sim \pi_{\V_{[n]}}(\cdot|y)$ from the VSP posterior (see Appendix A.4 
for proof).
\end{restatable}

We implemented separate MCMC samplers targeting both (\ref{eq:post-BDT}) and (\ref{eq:post-VSP}). Our MCMC algorithms are given in Appendix~C. We checked that the VSP-posterior marginals for the two implementations were equal (up to Monte-Carlo error). We implemented MCMC targeting the BDT posterior (\ref{eq:post-BDT}) first, as BDT data structures are slightly more straightforward to handle than the MDT data structures needed to target the VSP posterior in (\ref{eq:post-VSP}). All results in the next section were computed using the BDT-MCMC.

\section{APPLICATIONS}\label{sec:results}

\subsection{DATA AND ANALYSES}\label{sec:acta-data-analysis-outline}

We analyse a dataset accessed through a database made for ``The Charters of William II and Henry I'' project by Professor Richard Sharpe and Dr Nicholas Karn \citep{sharpe14}. These data
collect witness lists from legal documents from England and Wales in the eleventh and twelfth century. 
Witness lists respect a rigid social hierarchy: higher status individuals come ahead of lower status individuals in the lists. \Fig~D.1 
is an example list.

We represent the hierarchy on actors $[n]$ appearing in the lists as a partial order which is a VSP $v\in \V_{[n]}$ and model a list as the outcome of one of the queuing processes described in Section~\ref{sec:QJ-B-defn-main}.
We imagine the actors lining up to witness the document in a virtual queue.

Lists are witnessed by people from all walks of life and we have their titles. These include ``others'' (actors who lack titles). Historians are interested in social hierarchies and how they change over time.  For illustration we reconstruct hierarchies in three snapshots: the years 1080-84, 1126-30 and 1134-38. The last two cover periods shortly before and after Stephen became King, a time of great change. The 5-year intervals are short enough for any changes in the hierarchy to be slight \citep{nicholls2022ts}. For ease of visualisation we present results for individuals appearing in at least 5 lists (5LPA data) here and results on all actors (1LPA data) in Appendix~D.1.1.
We fit VSP/QJ-U to all data and fit VSP/QJ-B to 2 of the 3 5LPA data sets (not 1134-38, as QJ-B has runtime growing exponentially with the length of the longest list). However, relations between bishops in 1134-38 are of particular interest so we present VSP/QJ-B results for this subgroup. Table D.1 
summarises the data in the different experiments on the Royal Acta data.

In a separate analysis illustrating how our methods apply more generally to any rank-order data, we give an analysis of Formula 1 race outcomes for the 2021 season. Data and results are given in Appendix~D.2.

The prior for error probability $p$ and for $q$ (probability for an $S$-node) is given in \Fig~\ref{fig:pq}. All fitting is done using MCMC in the BDT representation, Algorithm~C.1. 
For any given model we draw MCMC samples $t^{(k)},p^{(k)},q^{(k)},\phi^{(k)}\sim \pi_{\T_{[n]}}(\cdot|y)$ for $k=1,...,K$ and set $v^{(k)}=v(t^{(k)})$ per Proposition~\ref{prop:posteriors}. Example MCMC traces are given in the supplement with Effective Sample Size (ESS) values (Appendix~D.1).
Sampled VSPs are summarised using consensus VSPs:
$V^{con}(\epsilon)$ includes order relation/edge $\langle i,j\rangle$ if the relation appears more than $\epsilon K$ times in the MCMC output. We color edges black if they are in $V^{con}(\epsilon)$ at $\epsilon=0.5$ but not $\epsilon=0.9$ and red if they are supported at $\epsilon=0.9$. We plot transitive reductions. These omit strongly supported edges from the top to the bottom of the DAG for clarity.

In \sec~\ref{sec:QJ-up}, we fit the QJ-U and QJ-B models to the 5LPA data and make a model comparison using Bayes factors. Consensus orders for the 1LPA data are given in Appendix~D.1.1. 
We additionally compare these models with bucket order models, a Plackett-Luce mixture, Mallows mixture and latent partial order model in Appendix~E. We carry out these tests on both the Royal Acta data and the F1 race result data. We report computing time measurements for counting \LE s for the latent partial order model and the VSP. They are compared empirically in Appendix~F.

\subsection{RESULTS}\label{sec:QJ-up}


{\it We begin by making reconstruction-accuracy tests on synthetic data.} Our list data are incomplete, in the sense that the membership in list $i=1,...,N$ is $o_i$ not $[n]$ and the $N$-values in Table~D.1 
are not much larger than the number of actors $n$. In order to measure the reliability of the reconstructions which follow we take representative parameters (parameters sampled from the corresponding posterior, the last sampled state $v^{(K)},p^{(K)},q^{(K)},\phi^{(K)}$) and generate synthetic data with the same list-membership and length structures as the real data. 
The ROC curves in Fig.~D.12 (5LPA data and QJ-U) 
and D.15 (5LPA data and QJ-B) 
for consensus orders $V^{con}(\epsilon)$ show the proportion of inferred false-positive and true-positive relations increasing with decreasing $\epsilon$ from $(0,0)$ at $\epsilon=1$ (the consensus order is empty) to $(1,1)$ at $\epsilon=0$ (complete graph). For each simulated data set there is $\epsilon$ giving high true-positive and low false-positive reconstructed relation fractions: if our model is accurate then we reconstruct relations well.


{\it We next report consensus partial orders.} Consensus orders for actors color-coded by their professions are shown in \Fig~\ref{fig:34-38fullcon-order} and \ref{fig:8084con-order}. For both QJ-U and QJ-B models, we observe three clear social hierarchies: King $\succ$ Queen $\succ$ Duke appear at the top, in that order (when they are in the 5LPA data, in 1180-84 and 1134-38); then archbishop/prince $\succ$ bishops; the remaining professions (earl, count, chancellor, other) are ranked lower than bishops in a relatively complex hierarchy. 

\begin{minipage}{\linewidth}
    \centering
    \begin{tikzpicture}[thick,scale=.3, every node/.style={scale=0.3}]
        \node[draw, circle, minimum width=.7cm, fill=Goldenrod] (1) at (-7, -1.5) {};
        \node[draw, circle, minimum width=.7cm,fill=Fuchsia] (2) at (-7, -2.5) {};
        \node[draw, circle, minimum width=.7cm,fill=RoyalBlue] (3) at (-7, -3.5) {};
        \node[draw, circle, minimum width=.7cm,fill=Maroon] (4) at (-7, -4.5) {};
        \node[draw, circle, minimum width=.7cm,fill=Maroon] (5) at (-7, -5.5) {};
        \node[draw, circle, minimum width=.7cm,fill=Maroon] (6) at (-7, -6.5) {};
        \node[draw, circle, minimum width=.7cm,fill=Magenta] (7) at (-7, -7.5) {};
        \node[draw, circle, minimum width=.7cm,fill=Magenta] (8) at (-7, -8.5) {};
        \node[draw, circle, minimum width=.7cm,fill=Magenta] (9) at (-12, -9) {};
        \node[draw, circle, minimum width=.7cm,fill=Magenta] (10) at (-6.5, -9.5) {};
        \node[draw, circle, minimum width=.7cm,fill=Magenta] (11) at (-11, -10.5) {};
        \node[draw, circle, minimum width=.7cm,fill=Magenta] (12) at (-10, -10.5) {};
        \node[draw, circle, minimum width=.7cm,fill=Magenta] (13) at (-9, -10.5) {};
        \node[draw, circle, minimum width=.7cm,fill=Magenta] (14) at (-8, -10.5) {};
        \node[draw, circle, minimum width=.7cm,fill=Magenta] (15) at (-7, -10.5) {};
        \node[draw, circle, minimum width=.7cm,fill=Magenta] (16) at (-6,-10.5) {};
        \node[draw, circle, minimum width=.7cm,fill=Magenta] (17) at (-5, -10.5) {};
        \node[draw, circle, minimum width=.7cm,fill=Magenta] (18) at (-4, -10.5) {};
        \node[draw, circle, minimum width=.7cm,fill=Magenta] (19) at (-3, -10.5) {};
        \node[draw, circle, minimum width=.7cm,fill=Magenta] (20) at (-2, -10.5) {};
        \node[draw, circle, minimum width=.7cm,fill=Magenta] (21) at (-6.5, -11.5) {};
        \node[draw, circle, minimum width=.7cm,fill=SpringGreen] (22) at (-6.5, -12.5) {};
        \node[draw, circle, minimum width=.7cm,fill=SkyBlue] (23) at (-5.5, -13) {};
        \node[draw, circle, minimum width=.7cm,fill=SkyBlue] (24) at (-6.2, -13.5) {};
        \node[draw, circle, minimum width=.7cm,fill=Apricot] (25) at (-4.8, -13.5) {};
        \node[draw, circle, minimum width=.7cm,fill=SkyBlue] (26) at (-7, -14.5) {};
        \node[draw, circle, minimum width=.7cm,fill=SkyBlue] (27) at (-5.5, -14.5) {};
        \node[draw, circle, minimum width=.7cm,fill=SkyBlue] (28) at (-4, -14.5) {};
        \node[draw, circle, minimum width=.7cm,fill=SkyBlue] (29) at (-2.5, -14.5) {};
        \node[draw, circle, minimum width=.7cm,fill=Gray] (30) at (-8.5, -14.5) {};
        \node[draw, circle, minimum width=.7cm,fill=Gray] (31) at (-10, -14.5) {};
        \node[draw, circle, minimum width=.7cm,fill=Magenta] (32) at (-12, -14.5) {};
        \node[draw, circle, minimum width=.7cm,fill=Gray] (33) at (-11.5, -15.5) {};
        \node[draw, circle, minimum width=.7cm,fill=Gray] (34) at (-10, -15.5) {};
        \node[draw, circle, minimum width=.7cm,fill=Gray] (35) at (-8.5, -15.5) {};
        \node[draw, circle, minimum width=.7cm,fill=Gray] (36) at (-7, -15.5) {};
        \node[draw, circle, minimum width=.7cm,fill=Gray] (37) at (-5.5, -15.5) {};
        \node[draw, circle, minimum width=.7cm,fill=SkyBlue] (38) at (-7, -16.5) {};
        \node[draw, circle, minimum width=.7cm,fill=SkyBlue] (39) at (-9, -16.5) {};
        \node[draw, circle, minimum width=.7cm,fill=Gray] (40) at (-11, -16.5) {};
        \node[draw, circle, minimum width=.7cm,fill=SkyBlue] (41) at (-13, -16.5) {};
        \node[draw, circle, minimum width=.7cm,fill=SkyBlue] (42) at (-5, -16.5) {};
        \node[draw, circle, minimum width=.7cm,fill=Gray] (43) at (-3, -16.5) {};
        \node[draw, circle, minimum width=.7cm,fill=Magenta] (44) at (-1, -16.5) {};
        \node[draw, circle, minimum width=.7cm,fill=SkyBlue] (45) at (-7.8, -17.5) {};
        \node[draw, circle, minimum width=.7cm,fill=Gray] (46) at (-9.4, -17.5) {};
        \node[draw, circle, minimum width=.7cm,fill=Gray] (47) at (-6.2, -17.5) {};
        \node[draw, circle, minimum width=.7cm,fill=Gray] (48) at (-4.6, -17.5) {};
        \node[draw, circle, minimum width=.7cm,fill=Gray] (49) at (-7, -18.5) {};
        \begin{pgfonlayer}{bg}
            \draw[-latex,red] (1) -- (2);
            \draw[-latex,red] (2) -- (3);
            \draw[-latex] (3) -- (4);
            \draw[-latex,red] (4) -- (5);
            \draw[-latex] (5) -- (6);
            \draw[-latex] (6) -- (7);
            \draw[-latex,red] (7) -- (8);
            \draw[-latex,red] (8) -- (9);
            \draw[-latex,red] (9) -- (10);
            \draw[-latex] (10) -- (11);
            \draw[-latex] (10) -- (12);
            \draw[-latex,red] (10) -- (13);
            \draw[-latex,red] (10) -- (14);
            \draw[-latex] (10) -- (15);
            \draw[-latex,red] (10) -- (16);
            \draw[-latex] (10) -- (17);
            \draw[-latex] (10) -- (18);
            \draw[-latex] (10) -- (19);
            \draw[-latex,red] (10) -- (20);
            \draw[-latex] (11) -- (21);
            \draw[-latex] (12) -- (21);
            \draw[-latex,red] (13) -- (21);
            \draw[-latex,red] (14) -- (21);
            \draw[-latex] (15) -- (21);
            \draw[-latex,red] (16) -- (21);
            \draw[-latex] (17) -- (21);
            \draw[-latex] (18) -- (21);
            \draw[-latex] (19) -- (21);
            \draw[-latex] (20) -- (21);
            \draw[-latex,red] (21) -- (22);
            \draw[-latex] (22) -- (23);
            \draw[-latex] (23) -- (24);
            \draw[-latex] (23) -- (25);
            \draw[-latex] (9) -- (32);
            \draw[-latex] (22) -- (31);
            \draw[-latex] (30) -- (40);
            \draw[-latex,red] (27) -- (40);
            \draw[-latex,red] (28) -- (40);
            \draw[-latex,red] (29) -- (40);
            \draw[-latex] (26) -- (40);
            \draw[-latex] (31) -- (40);
            \draw[-latex] (32) -- (40);
            \draw[-latex] (30) -- (42);
            \draw[-latex,red] (27) -- (42);
            \draw[-latex,red] (28) -- (42);
            \draw[-latex,red] (29) -- (42);
            \draw[-latex] (26) -- (42);
            \draw[-latex] (31) -- (42);
            \draw[-latex] (32) -- (42);
            \draw[-latex] (30) -- (38);
            \draw[-latex,red] (27) -- (38);
            \draw[-latex,red] (28) -- (38);
            \draw[-latex,red] (29) -- (38);
            \draw[-latex] (26) -- (38);
            \draw[-latex] (31) -- (38);
            \draw[-latex] (32) -- (38);
            \draw[-latex] (30) -- (43);
            \draw[-latex,red] (27) -- (43);
            \draw[-latex,red] (28) -- (43);
            \draw[-latex,red] (29) -- (43);
            \draw[-latex] (26) -- (43);
            \draw[-latex] (31) -- (43);
            \draw[-latex] (32) -- (43);
            \draw[-latex] (24) -- (30);
            \draw[-latex] (24) -- (29);
            \draw[-latex] (24) -- (28);
            \draw[-latex] (24) -- (27);
            \draw[-latex] (24) -- (26);
            \draw[-latex] (25) -- (30);
            \draw[-latex] (25) -- (28);
            \draw[-latex] (25) -- (27);
            \draw[-latex,red] (27) -- (33);
            \draw[-latex,red] (27) -- (34);
            \draw[-latex,red] (27) -- (35);
            \draw[-latex,red] (27) -- (36);
            \draw[-latex,red] (27) -- (37);
            \draw[-latex,red] (28) -- (33);
            \draw[-latex,red] (28) -- (34);
            \draw[-latex,red] (28) -- (35);
            \draw[-latex,red] (28) -- (36);
            \draw[-latex,red] (28) -- (37);
            \draw[-latex,red] (29) -- (33);
            \draw[-latex,red] (29) -- (34);
            \draw[-latex,red] (29) -- (35);
            \draw[-latex,red] (29) -- (36);
            \draw[-latex,red] (29) -- (37);
            \draw[-latex] (26) -- (33);
            \draw[-latex] (26) -- (34);
            \draw[-latex] (26) -- (35);
            \draw[-latex] (26) -- (36);
            \draw[-latex] (26) -- (37);
            \draw[-latex] (31) -- (33);
            \draw[-latex] (31) -- (34);
            \draw[-latex] (31) -- (35);
            \draw[-latex] (31) -- (36);
            \draw[-latex] (31) -- (37);
            \draw[-latex] (32) -- (33);
            \draw[-latex] (32) -- (34);
            \draw[-latex] (32) -- (35);
            \draw[-latex] (32) -- (36);
            \draw[-latex] (32) -- (37);
            \draw[-latex] (30) -- (33);
            \draw[-latex] (30) -- (34);
            \draw[-latex] (30) -- (35);
            \draw[-latex] (30) -- (36);
            \draw[-latex] (30) -- (37);
            \draw[-latex] (9) -- (44);
            \draw[-latex] (33) -- (39);
            \draw[-latex] (34) -- (39);
            \draw[-latex] (35) -- (39);
            \draw[-latex] (36) -- (39);
            \draw[-latex] (37) -- (39);
            \draw[-latex] (37) -- (39);
            \draw[-latex] (33) -- (41);
            \draw[-latex,red] (38) -- (45);
            \draw[-latex] (38) -- (46);
            \draw[-latex,red] (38) -- (47);
            \draw[-latex] (38) -- (48);
            \draw[-latex,red] (39) -- (45);
            \draw[-latex] (39) -- (46);
            \draw[-latex,red] (39) -- (47);
            \draw[-latex] (39) -- (48);
            \draw[-latex,red] (40) -- (45);
            \draw[-latex] (40) -- (46);
            \draw[-latex,red] (40) -- (47);
            \draw[-latex] (40) -- (48);
            \draw[-latex] (41) -- (45);
            \draw[-latex] (41) -- (46);
            \draw[-latex] (41) -- (47);
            \draw[-latex] (41) -- (48);
            \draw[-latex] (42) -- (45);
            \draw[-latex] (42) -- (46);
            \draw[-latex] (42) -- (47);
            \draw[-latex] (42) -- (48);
            \draw[-latex] (43) -- (45);
            \draw[-latex] (43) -- (47);
            \draw[-latex] (43) -- (48);
            \draw[-latex] (44) -- (45);
            \draw[-latex] (44) -- (47);
            \draw[-latex] (44) -- (48);
            \draw[-latex] (45) -- (49);
            \draw[-latex] (46) -- (49);
            \draw[-latex] (47) -- (49);
            \draw[-latex] (48) -- (49);
        \end{pgfonlayer}
    \end{tikzpicture}
    \qquad
    \begin{tikzpicture}[scale=0.4,every node/.style={scale=0.4}]
        \matrix [draw, left]{
          \node [circle,fill=Goldenrod,label=right:King] {}; \\
          \node [circle,fill=Fuchsia,label=right:Queen] {}; \\
          \node [circle,fill=RoyalBlue,label=right:Son of King] {}; \\
          \node [circle,fill=Maroon,label=right:Archbishop] {}; \\
          \node [circle,fill=Magenta,label=right:Bishop] {}; \\
          \node [circle,fill=Violet,label=right:Duke] {}; \\
          \node [circle,fill=SpringGreen,label=right:Chancellor] {};\\
          \node [circle,fill=SkyBlue,label=right:Earl] {}; \\
          \node [circle,fill=Apricot,label=right:Count] {}; \\
          \node [circle,fill=Gray,label=right:Other] {}; \\
        };
    \end{tikzpicture}
    
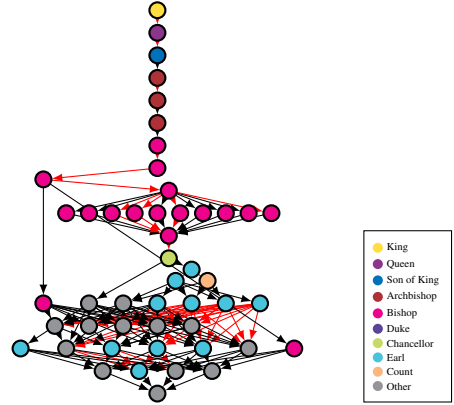
\captionof{figure}{VSP/QJ-U model. Consensus order for 1134-38 5LPA data. Significant/strong order relations are indicated by black/red edges respectively.}
    \label{fig:34-38fullcon-order}
\end{minipage}

\begin{figure}[h]
\begin{minipage}{\linewidth}
    \centering
    \begin{tikzpicture}[thick,scale=.17, every node/.style={scale=0.2}]
        \node[draw, circle, minimum width=.7cm,fill=Magenta] (50) at (19, 12) {};
        \node[draw, circle, minimum width=.7cm,fill=Magenta] (51) at (19, 10.5) {};
        \node[draw, circle, minimum width=.7cm,fill=Magenta] (52) at (19, 9) {};
        \node[draw, circle, minimum width=.7cm,fill=Magenta] (53) at (19, 7.5) {};
        \node[draw, circle, minimum width=.7cm,fill=Magenta] (54) at (16, 6.5) {};
        \node[draw, circle, minimum width=.7cm,fill=Magenta] (55) at (22, 6.5) {};
        \node[draw, circle, minimum width=.7cm,fill=Magenta] (56) at (17, 5.5) {};
        \node[draw, circle, minimum width=.7cm,fill=Magenta] (57) at (15, 5.5) {};
        \node[draw, circle, minimum width=.7cm,fill=Magenta] (58) at (16, 4.5) {};
        \node[draw, circle, minimum width=.7cm,fill=Magenta] (59) at (16, 3) {};
        \node[draw, circle, minimum width=.7cm,fill=Magenta] (60) at (19, 3) {};
        \node[draw, circle, minimum width=.7cm,fill=Magenta] (61) at (21, 3) {};
        \node[draw, circle, minimum width=.7cm,fill=Magenta] (62) at (23, 3) {};
        \node[draw, circle, minimum width=.7cm,fill=Magenta] (63) at (19, 1) {};
        \begin{pgfonlayer}{bg}
            \draw[-latex,red] (50) -- (51);
            \draw[-latex,red] (51) -- (52);
            \draw[-latex,red] (52) -- (53);
            \draw[-latex,red] (53) -- (54);
            \draw[-latex] (53) -- (55);
            \draw[-latex] (54) -- (56);
            \draw[-latex] (54) -- (57);
            \draw[-latex] (56) -- (58);
            \draw[-latex] (57) -- (58);
            \draw[-latex] (58) -- (59);
            \draw[-latex] (55) -- (56);
            \draw[-latex] (53) -- (60);
            \draw[-latex] (53) -- (61);
            \draw[-latex] (53) -- (62);
            \draw[-latex] (60) -- (63);
            \draw[-latex] (61) -- (63);
            \draw[-latex] (62) -- (63);
            \draw[-latex] (59) -- (63);
        \end{pgfonlayer}
        \node[draw, circle, minimum width=.7cm,fill=Gray] (64) at (-6, 1) {};
        \node[draw, circle, minimum width=.7cm,fill=Gray] (65) at (-6, 1.8) {};
        \node[draw, circle, minimum width=.7cm,fill=Apricot] (66) at (-9, 2.6) {};
        \node[draw, circle, minimum width=.7cm,fill=SkyBlue] (67) at (-9, 3.4) {};
        \node[draw, circle, minimum width=.7cm,fill=SkyBlue] (68) at (-9, 4.2) {};
        \node[draw, circle, minimum width=.7cm,fill=Magenta] (69) at (-9, 5) {};
        \node[draw, circle, minimum width=.7cm,fill=Magenta] (70) at (-6, 5.8) {};
        \node[draw, circle, minimum width=.7cm,fill=Magenta] (71) at (-6, 6.6) {};
        \node[draw, circle, minimum width=.7cm,fill=Magenta] (72) at (-6, 7.4) {};
        \node[draw, circle, minimum width=.7cm,fill=Magenta] (73) at (-6, 8.2) {};
        \node[draw, circle, minimum width=.7cm,fill=Apricot] (74) at (-3, 3.8) {};
        \node[draw, circle, minimum width=.7cm,fill=Maroon] (75) at (-9, 8.9) {};
        \node[draw, circle, minimum width=.7cm,fill=RoyalBlue] (76) at (-3, 8.9) {};
        \node[draw, circle, minimum width=.7cm,fill=Maroon] (77) at (-6, 9.6) {};
        \node[draw, circle, minimum width=.7cm,fill=Violet] (78) at (-6, 10.4) {};
        \node[draw, circle, minimum width=.7cm,fill=Fuchsia] (79) at (-6, 11.2) {};
        \node[draw, circle, minimum width=.7cm,fill=Goldenrod] (80) at (-6, 12) {};
        \begin{pgfonlayer}{bg}
            \draw[-latex] (65) -- (64);
            \draw[-latex,red] (66) -- (65);
            \draw[-latex] (67) -- (66);
            \draw[-latex] (68) -- (67);
            \draw[-latex,red] (69) -- (68);
            \draw[-latex,red] (70) -- (69);
            \draw[-latex,red] (71) -- (70);
            \draw[-latex] (72) -- (71);
            \draw[-latex] (73) -- (72);
            \draw[-latex] (70) -- (74);
            \draw[-latex] (74) -- (65);
            \draw[-latex,red] (75) -- (73);
            \draw[-latex] (76) -- (73);
            \draw[-latex,red] (77) -- (75);
            \draw[-latex] (77) -- (76);
            \draw[-latex] (78) -- (77);
            \draw[-latex] (79) -- (78);
            \draw[-latex,red] (80) -- (79);
        \end{pgfonlayer}
        \node[draw, circle, minimum width=.7cm,fill=Maroon] (81) at (7, 12) {};
        \node[draw, circle, minimum width=.7cm,fill=Magenta] (82) at (7, 10.8) {};
        \node[draw, circle, minimum width=.7cm,fill=Magenta] (83) at (7, 9.6) {};
        \node[draw, circle, minimum width=.7cm,fill=Magenta] (84) at (7, 8.4) {};
        \node[draw, circle, minimum width=.7cm,fill=Magenta] (85) at (7, 7.2) {};
        \node[draw, circle, minimum width=.7cm,fill=SkyBlue] (86) at (2, 6) {};
        \node[draw, circle, minimum width=.7cm,fill=Gray] (87) at (2, 4.5) {};
        \node[draw, circle, minimum width=.7cm,fill=Gray] (88) at (2, 3) {};
        \node[draw, circle, minimum width=.7cm,fill=Gray] (89) at (3.5, 2) {};
        \node[draw, circle, minimum width=.7cm,fill=Apricot] (90) at (5, 3) {};
        \node[draw, circle, minimum width=.7cm,fill=Magenta] (91) at (9, 4) {};
        \node[draw, circle, minimum width=.7cm,fill=Gray] (92) at (12, 4) {};
        \node[draw, circle, minimum width=.7cm,fill=Gray] (93) at (9, 1) {};
        \begin{pgfonlayer}{bg}
            \draw[-latex,red] (81) -- (82);
            \draw[-latex,red] (82) -- (83);
            \draw[-latex,red] (83) -- (84);
            \draw[-latex] (84) -- (85);
            \draw[-latex,red] (85) -- (86);
            \draw[-latex] (86) -- (87);
            \draw[-latex] (87) -- (88);
            \draw[-latex] (88) -- (89);
            \draw[-latex] (90) -- (89);
            \draw[-latex,red] (85) -- (90);
            \draw[-latex] (90) -- (89);
            \draw[-latex,red] (85) -- (91);
            \draw[-latex,red] (92) -- (93);
            \draw[-latex,red] (85) -- (92);
            \draw[-latex] (91) -- (93);
            \draw[-latex] (89) -- (93);
        \end{pgfonlayer}
        \node[draw, circle, minimum width=.7cm,fill=Maroon] (1) at (7, -1) {};
        \node[draw, circle, minimum width=.7cm,fill=Magenta] (2) at (7, -2) {};
        \node[draw, circle, minimum width=.7cm,fill=Magenta] (3) at (7, -3) {};
        \node[draw, circle, minimum width=.7cm,fill=Magenta] (4) at (7, -4) {};
        \node[draw, circle, minimum width=.7cm,fill=Magenta] (5) at (7, -5) {};
        \node[draw, circle, minimum width=.7cm,fill=Magenta] (6) at (7, -6) {};
        \node[draw, circle, minimum width=.7cm,fill=SkyBlue] (7) at (7, -7) {};
        \node[draw, circle, minimum width=.7cm,fill=Gray] (8) at (7, -8) {};
        \node[draw, circle, minimum width=.7cm,fill=Gray] (9) at (10, -9) {};
        \node[draw, circle, minimum width=.7cm,fill=Apricot] (10) at (4, -9) {};
        \node[draw, circle, minimum width=.7cm,fill=Gray] (11) at (7, -10) {};
        \node[draw, circle, minimum width=.7cm,fill=Gray] (12) at (7, -11) {};
        \node[draw, circle, minimum width=.7cm,fill=Gray] (13) at (7, -12) {};
        \begin{pgfonlayer}{bg}
            \draw[-latex,red] (1) -- (2);
            \draw[-latex,red] (2) -- (3);
            \draw[-latex,red] (3) -- (4);
            \draw[-latex] (4) -- (5);
            \draw[-latex,red] (5) -- (6);
            \draw[-latex,red] (6) -- (7);
            \draw[-latex,red] (7) -- (8);
            \draw[-latex,red] (8) -- (9);
            \draw[-latex] (8) -- (10);
            \draw[-latex] (10) -- (11);
            \draw[-latex] (9) -- (11);
            \draw[-latex] (11) -- (12);
            \draw[-latex] (12) -- (13);
        \end{pgfonlayer}
        \node[draw, circle, minimum width=.7cm,fill=Goldenrod] (14) at (-6, -.8) {};
        \node[draw, circle, minimum width=.7cm,fill=Fuchsia] (15) at (-6, -1.6) {};
        \node[draw, circle, minimum width=.7cm,fill=Maroon] (16) at (-6, -2.4) {};
        \node[draw, circle, minimum width=.7cm,fill=Maroon] (17) at (-9, -3.2) {};
        \node[draw, circle, minimum width=.7cm,fill=Violet] (18) at (-3, -3.2) {};
        \node[draw, circle, minimum width=.7cm,fill=Magenta] (19) at (-6, -4) {};
        \node[draw, circle, minimum width=.7cm,fill=RoyalBlue] (20) at (-6, -4.8) {};
        \node[draw, circle, minimum width=.7cm,fill=Magenta] (21) at (-6, -5.6) {};
        \node[draw, circle, minimum width=.7cm,fill=Magenta] (22) at (-6, -6.4) {};
        \node[draw, circle, minimum width=.7cm,fill=Magenta] (23) at (-6, -7.2) {};
        \node[draw, circle, minimum width=.7cm,fill=Magenta] (24) at (-6, -8) {};
        \node[draw, circle, minimum width=.7cm,fill=SkyBlue] (25) at (-6, -8.8) {};
        \node[draw, circle, minimum width=.7cm,fill=Apricot] (26) at (-9, -9.6) {};
        \node[draw, circle, minimum width=.7cm,fill=SkyBlue] (27) at (-3, -9.6) {};
        \node[draw, circle, minimum width=.7cm,fill=Apricot] (28) at (-6, -10.4) {};
        \node[draw, circle, minimum width=.7cm,fill=Gray] (29) at (-6, -11.2) {};
        \node[draw, circle, minimum width=.7cm,fill=Gray] (30) at (-6, -12) {};
        \begin{pgfonlayer}{bg}
            \draw[-latex,red] (14) -- (15);
            \draw[-latex,red] (15) -- (16);
            \draw[-latex,red] (16) -- (17);
            \draw[-latex] (16) -- (18);
            \draw[-latex] (17) -- (19);
            \draw[-latex] (18) -- (19);
            \draw[-latex] (19) -- (20);
            \draw[-latex] (20) -- (21);
            \draw[-latex] (21) -- (22);
            \draw[-latex,red] (22) -- (23);
            \draw[-latex,red] (23) -- (24);
            \draw[-latex] (24) -- (25);
            \draw[-latex] (25) -- (26);
            \draw[-latex] (25) -- (27);
            \draw[-latex] (27) -- (28);
            \draw[-latex] (26) -- (28);
            \draw[-latex] (28) -- (29);
            \draw[-latex] (29) -- (30);
        \end{pgfonlayer}

        \node[draw, circle, minimum width=.7cm,fill=Magenta] (31) at (19, -1) {};
        \node[draw, circle, minimum width=.7cm,fill=Magenta] (32) at (19, -2) {};
        \node[draw, circle, minimum width=.7cm,fill=Magenta] (33) at (19, -2) {};
        \node[draw, circle, minimum width=.7cm,fill=Magenta] (34) at (19, -3) {};
        \node[draw, circle, minimum width=.7cm,fill=Magenta] (35) at (19, -4) {};
        \node[draw, circle, minimum width=.7cm,fill=Magenta] (36) at (17.5, -6) {};
        \node[draw, circle, minimum width=.7cm,fill=Magenta] (37) at (17.5, -7) {};
        \node[draw, circle, minimum width=.7cm,fill=Magenta] (38) at (16, -8) {};
        \node[draw, circle, minimum width=.7cm,fill=Magenta] (39) at (16, -9) {};
        \node[draw, circle, minimum width=.7cm,fill=Magenta] (40) at (17.5, -10) {};
        \node[draw, circle, minimum width=.7cm,fill=Magenta] (41) at (17.5, -11) {};
        \node[draw, circle, minimum width=.7cm,fill=Magenta] (42) at (19, -8.5) {};
        \node[draw, circle, minimum width=.7cm,fill=Magenta] (43) at (19, -12) {};
        \node[draw, circle, minimum width=.7cm,fill=Magenta] (44) at (22, -8.5) {};
        \begin{pgfonlayer}{bg}
            \draw[-latex,red] (31) -- (32);
            \draw[-latex,red] (32) -- (33);
            \draw[-latex,red] (33) -- (34);
            \draw[-latex,red] (34) -- (35);
            \draw[-latex] (35) -- (36);
            \draw[-latex] (35) -- (44);
            \draw[-latex] (36) -- (37);
            \draw[-latex] (37) -- (38);
            \draw[-latex] (38) -- (39);
            \draw[-latex,red] (39) -- (40);
            \draw[-latex,red] (40) -- (41);
            \draw[-latex,red] (41) -- (43);
            \draw[-latex] (37) -- (42);
            \draw[-latex] (42) -- (40);
            \draw[-latex] (44) -- (43);
        \end{pgfonlayer}
    \end{tikzpicture}
    \qquad
    \begin{tikzpicture}[scale=0.4,every node/.style={scale=0.4}]
        \matrix [draw, left]{
          \node [circle,fill=Goldenrod,label=right:King] {}; \\
          \node [circle,fill=Fuchsia,label=right:Queen] {}; \\
          \node [circle,fill=RoyalBlue,label=right:Son of King] {}; \\
          \node [circle,fill=Maroon,label=right:Archbishop] {}; \\
          \node [circle,fill=Magenta,label=right:Bishop] {}; \\
          \node [circle,fill=Violet,label=right:Duke] {}; \\
          \node [circle,fill=SkyBlue,label=right:Earl] {}; \\
          \node [circle,fill=Apricot,label=right:Count] {}; \\
          \node [circle,fill=Gray,label=right:Other] {}; \\
        };
    \end{tikzpicture}
    \captionof{figure}{VSP/QJ-U (top row) and VSP/QJ-B (bottom row). Consensus orders for 1080-84, 1126-30 and 1134-38 (bishops) (left to right columns) 5LPA data. } 
    \label{fig:8084con-order}
\end{minipage}
\end{figure}
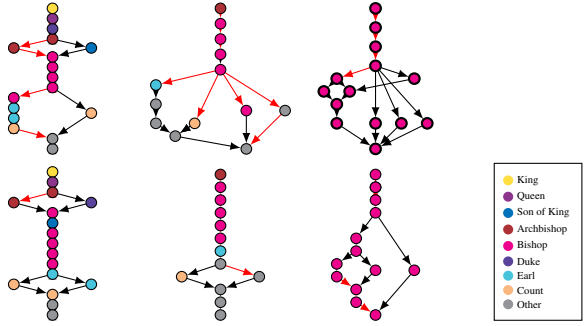

Some of this is common sense. However, the web of strongly attested relations between earls and others in 1134-38 is new. There is clear evidence for hierarchies within professions. The bishop-only QJ-U analysis in 1134-38 (top-right graph in Fig.~\ref{fig:8084con-order}) is similar to the bishop subgraph of the full QJ-U analysis for the same period (pink nodes in Fig.~\ref{fig:34-38fullcon-order}). The prior is marginally consistent, but information is shared across lists so removing actors changes the data and changes estimated order relations between those that remain. However, the bishops appear as a group in the lists and in Fig.~\ref{fig:34-38fullcon-order} and there are few non-bishops ``between'' bishops in lists, so this effect is slight. We can attach names to nodes: for example, the top three bishops in 34-38 (in Fig.~\ref{fig:34-38fullcon-order} and in both QJ-U and QJ-B analyses in the rightmost column of Fig.~\ref{fig:8084con-order}) are Henry, de Blois, Bishop of Wincester $\succ$ Roger, Bishop of Salisbury $\succ$ Alexander, Bishop of Lincoln. 

The status hierarchies fitted using by QJ-B (bottom row Fig.~\ref{fig:8084con-order}) are simpler and deeper than QJ-U (top row Fig.~\ref{fig:8084con-order}). The data must contain a small number of errors in both directions. A uni-directional model must fit a shallower VSP as it accommodates errors in the ``wrong'' direction by removing order relations in the reconstructed VSP. 




We summarise the status of ``professions'' within VSPs by averaging ranks. Given a partial order $v\in\V_{[n]}$, the rank of actor $i\in[n]$ is the number of actors above them, $\text{rank}_i(v) = 1+|\{\langle e_1,e_2\rangle\in E(v): e_2=i\}|$, and take as our summary the average rank of actors in the profession. The posterior mean ranks given in Table D.5 
and D.7 
match our remarks on consensus orders.

{\it We next report parameter distributions.} Prior and posterior distributions for the probability $q$ for a serial node, error probability $p$ and QJ-B parameter $\phi$ (equal one for QJ-U and zero for QJ-D) for the three periods are shown in \Fig~\ref{fig:pq}. The $p$-posteriors are weighted toward smaller values and overlap, though errors are low in 1126-30 and higher in 1180-84 indicating greater respect for the rules of precedence in 1126-30 than in 1180-84. Prior and posterior depth distributions are shown in \Fig~D.11 
and D.14.
The prior depth distributions are fairly flat so any depth-structure in the posterior comes from the data. The probability for a series node in the BDT ($q$) controls the depth of the fitted order relation. For example, in 1180-84 a relatively high $q$ for QJ-U is associated with relatively high depth VSPs with a mean depth of 14 relative to maximum depth 17 (the number of actors). In contrast, the posterior probabilities for S and P nodes are almost equal in 1134-38 and so we get a relatively shallower hierarchy: the posterior mean depth is about 23 relative to a maximum depth 49 in \Fig~D.11.

\begin{figure}[htb]
  \centering
\includegraphics[width=\linewidth]{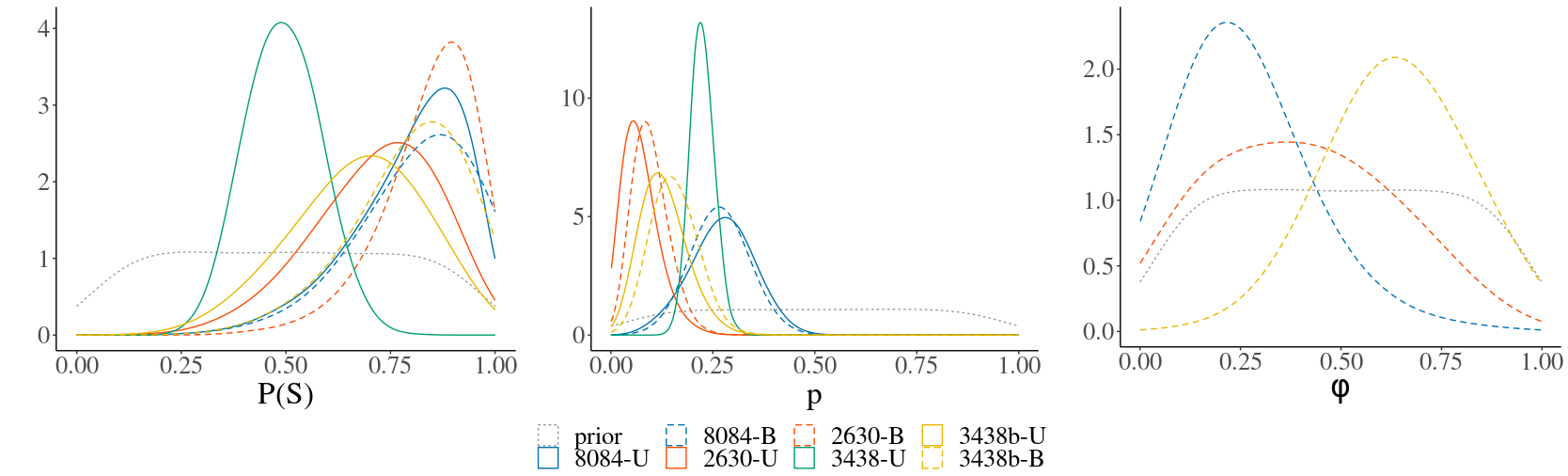}
  \caption{Posterior distributions for $q=P(S)$ (left), error probability $p$ (middle) and QJ-B probability $\phi$ (right) for the time periods 1180-1184 (blue), 1126-1130 (red), 1134-1138 (green) and 1134-1138(b) (yellow) from both the VSP/QJ-U (solid) and VSP/QJ-B (dashed) models. The prior is represented in grey in all figures. }\label{fig:pq}
\end{figure}





The QJ-B model for noise in the list data allows actors to jump up or down from a queue-position appropriate for their status. QJ-U is favored if $\phi>1/2$ and otherwise QJ-D so we see from Fig.~\ref{fig:pq} that QJ-U is favored in 1134-38(b), while the 1080-84 data supports QJ-D. 
However, the $p$-posteriors 
both favor small $p$. The displacement direction controlled by $\phi$ is hard to measure and not identifiable at $p=0$ so the $\phi$-distributions are correspondingly broad.


{\it We next report results of model selection between different queue jumping error models.} Preference shifts from downwards to bidirectional to upwards displacement error models over the period 1080-1140. We justify this reading of the results using Bayes factors below. In summary, QJ-D is slightly favored over QJ-B (so we write ``$D>B$'') in 1080-84 while in 1126-30 models QJ-D and QJ-B are equally good ($D\approx B$). Both are clearly favored over QJ-U in these periods ($D,B\gg U$). In 34-38(b) we have $U\approx B$ and $U,B\gg D$. 

We can read the Bayes factors we need off Fig.~\ref{fig:pq} because the models QJ-U and QJ-D are nested in the model QJ-B. The Bayes factor $B_{U,B}$ for QJ-U over QJ-B is 
\begin{align*}
    B_{U,B}&= \lim_{\delta\to 0}\frac{p(y|\phi>1-\delta)}{p(y|\phi\in(0,1))}\\
    &=\lim_{\delta\to 0}\frac{\pi(\phi>1-\delta|y)}{\pi(\phi\in(0,1)|y)}\frac{\pi(\phi\in(0,1))}{\pi(\phi>1-\delta)}\\
    &= \lim_{\delta\to 0}\frac{\pi(\phi>1-\delta|y)}{\pi(\phi>1-\delta)},
\end{align*}
since $\phi\in (0,1)$ with probability one. Similarly, 
\[
B_{D,B}=\lim_{\delta\to 0}\frac{\pi(\phi<\delta|y)}{\pi(\phi<\delta)},
\]
and then $B_{U,D}=B_{U,B}/B_{D,B}$.
From \Fig~\ref{fig:pq}, $B_{U,B}$ is close to 0 in periods 1180-84 and 1126-30 as the posterior density is well below the prior density at $\phi\to 1$, providing strong support for QJ-B over QJ-U. In 1134-38(b), we see $B_{U,B} \approx 1$, as the curves meet as $\phi\to 1$ so there is no clear signal from the data. The other comparisons may be justified similarly.

{\it Finally we make model comparisons with other models.} Comparisons with a Plackett-Luce mixture model and a Mallows mixture model are given in Appendix~E.1, the latent partial order model from \cite{nicholls122011partial} in Appendix~E.3 and a simple Bucket Order model in Appendix~E.2. When models are nested (Bucket Order) we estimate a Bayes factor. When they are not, we use the Expected Log Pointwise Predictive Density (ELPD, \cite{vehtari2017practical}) as our criterion. This is a predictive loss which can be estimated using LOOCV or the WAIC \citep{Watanabe2012}. On this basis VSP/QJ (-U and -B) is clearly favoured over Placket-Luce mixture models and Mallows mixture model in Table E.1 (``Royal Acta'') and E.3 (Formula 1 race data). With Bayes factors around 2 or 3, Bucket orders are equal or slightly preferred over VSPs in the QJ-B model (Table E.3). Our VSP-based model QJ-U is clearly preferred over Bucket orders in the QJ-U fit (some very large Bayes factors in favor of VSP).

The support of our VSP model is a subset of the PO support, as POs containing the forbidden sub-graph (Appendix G) are not VSPs. The PO/QJ-U has a slightly larger ELPD ($-36.7$, see Table~E.4) than VSP/QJ-U ($-37.8$) on the 1126-1130 data with 5LPA. However, the difference is not significant at the precision ($\pm 10$) of these estimates so we conclude that VSP/QJ-U models these data as well as PO/QJ-U. It gives similar consensus orders (\Fig~E.2) and profession rankings (Table E.5). 

A VSP-based analysis is far more computationally efficient than a PO-based model when the number of actors is large. The computing time for counting the LEs of a VSP rises linearly with the number of actors (\Fig~F.1) while it increases exponentially for PO (using the best code we could find, \textit{LEcount}, \cite{kangas2016counting}, but inevitable given \cite{brightwell1991counting}). 
We have to count LEs of random POs. In our experience counting LEs on random POs with up to about 25-30 actors is feasible. However, at larger numbers we encounter occasional random POs which are especially ``hard'' to count and VSP-based analysis is the only way forward at present.


\section{Discussion and Conclusion}

Our work was motivated by the need to fit relatively large partial orders (up to 200 nodes) to noisy linear-extension data. We saw that, for data on this scale, counting linear extensions in the VSP-tree representation is much faster than current state-of-art counting for general partial orders, enabling our methods to scale. We gave a new consistent and closed form prior distribution over VSPs with a parameter $q$ controlling VSP depth, and a new observation model QJ-B for noisy \LE s which generalises QJ-U \citep{nicholls122011partial}. We fit the new model to some of the smaller data sets and the old model to all data sets. Neither of these analyses would be possible without the VSP-setup. The data support the new observation model in our application. Our $elpd_{waic}$-based model comparisons also clearly favor VSP/QJ-U and VSP/QJ-B over a Plackett-Luce mixture or a Mallows Mixture. Although we could fit the large data sets, visualising consensus partial orders proved challenging (compare \Fig~\ref{fig:8084con-order} (top left corner) and \Fig~D.4).

We gave MCMC algorithms targeting the posterior for VSPs in both the latent-space (BDT) parameterisation and the integrated MDT parameterisation. We found the BDT-MCMC adequate, though it would be good to make an efficiency comparison with MDT-MCMC, which we expect to be more efficient. These comparisons are underway. BDT updates which don't change the VSP are fast so BDT-MCMC seems to be competitive. For code see
\url{https://github.com/JessieJ315/Bayesian-Inference-for-Vertex-Series-Parallel-Partial-Orders.git}.

In future work we would like to compare our fit with the recently-proposed contextual repeated selection (CRS) model (\cite{seshadri2020learning} and \cite{ragain2018choosing}). This is a rich class of models for rank-order data. The elements of the model are not essentially physical, in the sense that a VSP represents a social hierarchy relation by relation. Also, CRS models do not encode transitivity. It is easy to show VSP models cannot be represented as CRS models with ``cliques'' of size two. CRS models may fit the data well, and a comparison would be worthwhile. However, there is currently no Bayesian CRS analysis so we leave that for future work.






\begin{acknowledgements} 
    We thank Dr. Nicholas Karn for providing the data, and Prof. David Johnson for enlightening discussions. 
\end{acknowledgements}

\bibliography{reference}

\begin{appendix} 

\input{jiang_540-supp} 
\end{appendix}

\end{document}

%% file: jiang_540-supp.tex
\onecolumn 
\appendix
\counterwithin{figure}{section}
\setcounter{prop}{1}
\counterwithin{proof}{section}
\counterwithin{table}{section}
\counterwithin{algorithm}{section}
\counterwithin{definition}{section}
\counterwithin{equation}{section}
\counterwithin{cor}{section}
\section{Proof of Theorem \ref{thm:vsp-prior}
}

This Appendix states and proves the propositions referred to in the proof of Theorem~\ref{thm:vsp-prior} 
given in Section \ref{sec:vsp-prior}. 

\subsection{Part I: MARGINAL CONSISTENCY}

We first prove marginal consistency for our VSP prior. Intuitively, relations between actors in a VSP $v\in\V_{[n]}$ are determined by the type of their ``Most Recent Common Ancestor'' (MRCA) in any BDT $t\in t(v)$ representing $v$. For example the MRCA of actors 2 and 4 in the tree $t_0$ in \Fig~\ref{fig:tree_ex}
is the blue $P$-node, so $2\|_{v_0} 4$ in the VSP $v_0$ in \Fig~\ref{fig:POex}. 
Adding or removing a leaf in the BDT doesn't change relations between other actors because it doesn't change the types of their MRCA's. This property leads to marginal consistency of trees and VSPs.

We begin by giving a stochastic process realising $t\sim  \pi_{\T_{[n]}}(t|q)$ in which leaves are added to the tree one at a time. This construction appears in \cite{valdes78} but without the random element.

\begin{definition}[Leaf Insertion and deletion]\label{defn:leaf-ins-del}
If $t'\in \T_{[n-1]}$, $t'=(F(t'),E(t'),L(t'))$, is a tree on actors $(i_1,...,i_{n-1})$ with $\F'\cup\A'=[2n-3]$ then the \emph{leaf-insertion} operation $t=t' \triangleleft (e,i_n)$ at edge $e=\langle e_1,e_2\rangle,\ e\in E(t')$, gives a 
tree $t=(F(t),E(t),L(t))$ with two new nodes $j'=2n-2$ and $j=2n-1$, leaves $\F=\F'\cup \{j'\}$, internal nodes $\A=\A'\cup\{j\}$, leaf-to-actor map $F_{\F}(t)=F_\F(t')$ and $F_{j'}(t')=i_n$, edge set
\[E(t)=E(t')\setminus \{e\}\cup \{\langle j,j'\rangle,\langle e_1,j\rangle,\langle j,e_2\rangle\}\]
and $L(t)=L(t')\cup L_j$ where $L_j=(j',e_2),(e_2,j')$ or $\emptyset$ with probabilities $q/2,q/2$ and $1-q$ respectively. The leaf deletion operation $t'=t\triangleright i_n$ reverses this operation, pruning the leaf for actor $i_n$ (and removing its parent node).
\end{definition}
%
%
\begin{definition}[Generative model for BDTs]\label{defn:tree-gen-proc}
Let $(i_1,\dots,i_n)\in \P_{[n]}$ be the actor list taken in any order. Simulate $t\sim\pi_{\T_{[n]}}(t|q)$ as follows: 
    \begin{enumerate}
        \item Set $\F=\{0,1\}$, $\A=\emptyset$, $F_1=1,E=\{\langle 0,1\rangle\}, L=\emptyset$ and $t_{(1)}=(F,E,L)$ (a single-edge tree);
        \item For $k=2:n$, (add the actors one at a time)
        \begin{enumerate}
            \item choose an edge $e\sim\U\{E(t_{(k-1)})\}$ at random;
            \item set $t_{(k)}=t_{(k-1)}\triangleleft (e,i_k)$;
        \end{enumerate}
        \item $E(t_{(n)})$ contains an edge $e=\langle 0,e_2\rangle$. Return the BDT $t=(F(t_{(n)}),E(t_{(n)})\setminus \{e\},L(t_{(n)}))$
        with leaf labels $\F\gets \F\setminus \{0\}$.
    \end{enumerate}    
\end{definition}
If we run this generative model we get a random tree distributed according to $\pi_{\T_{[n]}}$.
\begin{restatable}[Prior Probability Distribution over $\T_{[n]}$]{prop}{priorT}
\label{prop:priorT}
The probability distribution over BDTs determined by the process in Definition~\ref{defn:tree-gen-proc} is 
    given by (\ref{eq:tree-prior}). 
    \begin{proof}[Proposition~\ref{prop:priorT}]
    Each distinct topology is determined by a unique sequence of edge choices at step 2a in Definition~\ref{defn:tree-gen-proc}, and at step $k$ an edge is chosen uniformly over the $2k-3$ edges of a tree with $k$ leaves (recall there is a temporary leaf $0\in \F$ which is removed at step 3). The types of internal nodes are independent so it makes no difference if we set them as we build the tree or at the end.
\end{proof}
\end{restatable}
 

  We now define sub-trees of BDTs. At the end of step $k$ in the tree-generation process in Definition~\ref{defn:tree-gen-proc} the ``current tree'' is $t_{(k)}\in \T_o$ with $o=(i_1,...,i_k)$ and at the end of 
  step $k'>k$ it is $t_{(k')}\in \T_{\tilde o}$ with $\tilde o=(i_1,...,i_k,i_{k+1},...,i_{k'})$. 
  If, for $o,\tilde o\in \Os_{[n]}$ with $o\subseteq \tilde o$, we fix $\tau\in \T_o$ and $t\in \T_{\tilde o}$ then the conditional probability
  \[
  \pi_{\T_{\tilde o}|\T_{o}}(t|\tau,q)= \Pr(t_{(k')}=t|t_{(k)}=\tau,q)
  \]
  is the probability to realise $t_{(k')}=t$ when $t_{(k)}=\tau$.
  \begin{definition}[Sub-trees and containing trees]\label{defn:sub-contain-tree}
       Tree $\tau$ is a sub-tree of $t$ (and $t$ contains $\tau$) if $\pi_{\T_{\tilde o}|\T_{o}}(t|\tau,q)>0$. Let \[\T_{\tilde o}(\tau)=\{t\in \T_{\tilde o}: \pi_{\T_{\tilde o}|\T_{o}}(t|\tau,q)>0\}\]
       give the set of trees in $\T_{\tilde o}$ containing a given tree $\tau\in \T_{o}$.
  \end{definition}
  If $t$ contains $\tau$ then $t$ can be realised from $\tau$ by a sequence of edge insertions $\triangleleft$ and $\tau$ can be recovered from $t$ removing the actors in $\tilde o\setminus o$ using the pruning operator $\triangleright$.
  
  The family of prior distributions over trees $\pi_{\T_{o}}(\tau|q),\ o\in \Os_{[n]},\ n\ge 1$ is marginally consistent  if, for all $n\ge 1$ and all $o,\tilde o\in \Os_{[n]}$ with $o\subseteq \tilde o$, distributions in the family satisfy
\begin{equation}\label{eq:mc_trees_defn}
\pi_{\T_{o}}(\tau|q)=\sum_{t\in\T_{\tilde o}(\tau)}\pi_{\T_{\tilde o}}(t|q)\quad\mbox{for all $\tau \in \T_o$ }.
\end{equation}

\begin{restatable}{prop}{treemc}
\label{prop:tree-mc}
The probability distribution over BDTs given in (3) 
is marginally consistent.
\begin{proof}[Proposition \ref{prop:tree-mc}]
    It is sufficient show marginal consistency holds for $\tilde o=[n]$ and $o=[n]\setminus \{i\}$ for any single actor $i\in [n]$ as Eqn.~\ref{eq:mc_trees_defn} follows for any pair of  subsets of $[n]$ by pruning leaves one at a time using the $\triangleright$ operator. 
    
    Since $\pi_{\T_{[n]}}(t|q)$ in Eqn.~3 
    does not depend on the order $i_1,\dots,i_n$ in which we add actors, we can make node $i_n=i$ the last arrival. If $t_{-i}$ is the tree at the end of the penultimate loop then
    \begin{equation} \label{eq:mc_tree_sum_e}
    \pi_{\T_{o}}(t_{-i}|q) = \sum_{e\in E(t_{-i})}\pi_{\T_{\tilde o}}( t_{-i}\triangleleft (e,i) |q).
    \end{equation}
    Now take $\tau = t_{-i}$. Since leaf deletion reverses edge insertion, the set of trees $\T_{[n]}(\tau)$ that contain $\tau$ is the set of trees that are obtained from $\tau$ by some edge addition,
    \[
      \T_{[n]}(\tau)=\bigcup_{e\in E(\tau)} \{\tau \triangleleft (e,i)\}
    \]
    and so 
    \[
    \pi_{\T_{o}}(\tau|q) = \sum_{t\in \T_{[n]}(\tau)}\pi_{\T_{\tilde o}}(t|q).
    \]
    which is marginal consistency for addition of one actor.
\end{proof}
\end{restatable}

\begin{restatable}{prop}{vspmc}
\label{prop:vsp-mc}
The probability distribution over VSPs given in (4) 
is marginally consistent.
\begin{proof}[Proposition~\ref{prop:vsp-mc}]
It is sufficient to show that Eqn.~5 
holds for $\tilde o=[n]$ and $o=[n]\setminus \{i\}$ and any $i\in [n]$ in Definition~1 
since Eqn.~5 
follows for any pair of subsets of $[n]$ by removing actors one at a time. 

In this case $v[o]$ is the suborder obtained from $v\in \V_{[n]}$ by removing actor $i$ and we want to verfiy
    \begin{equation}\label{eq:vsp-mc}
        \pi_{\V_{o}}(w|q) = \sum_{v\in\V_{[n]} \atop v[o]=w} 
        \pi_{\V_{[n]}}(v|q)\quad \mbox{for all $w\in \V_o$}.
    \end{equation}
    Picking up the RHS of Eqn.~\ref{eq:vsp-mc} we have from Eqn.~4
    \begin{align*}
        \sum_{v\in\V_{[n]} \atop v[o]=w} 
        \pi_{\V_{[n]}}(v|q)&=\sum_{v\in\V_{[n]} \atop v[o]=w}\sum_{t\in t(v)}\pi_{\T_{[n]}}(t|q).
    \end{align*}
    Referring to Definition~\ref{defn:sub-contain-tree}, the sum on the right is a sum over all trees ``containing'' a tree in $t(w)$, that is, the set of all trees which can be constructed by taking a tree $\tau\in t(w)$ and adding actor $i$ to the tree by edge insertion at any edge in $\tau$:
    \begin{align*}
        \bigcup_{v\in\V_{[n]}\atop v[o]=w} \bigcup_{t\in t(v)}\{t\} 
        & = \bigcup_{\tau\in t(w)}\bigcup_{e\in E(\tau)} \{\tau \triangleleft (e,i)\}.
    \end{align*}
    It follows that
    \begin{align*}
        \sum_{v\in\V_{[n]} \atop v[o]=w} 
        \pi_{\V_{[n]}}(v|q) & = \sum_{\tau\in t(w)}\sum_{e\in E(\tau)} \pi_{\T_{[n]}}(\tau\triangleleft (e,i)|q)\\
        & = \sum_{\tau\in t(w)}\pi_{\T_o}(\tau|q),\ (\mbox{by Eqn.~\ref{eq:mc_tree_sum_e}})\\
        & = \pi_{\V_o}(w|q) \hspace*{0.55in} (\mbox{by Eqn.~4)},
    \end{align*}
    which is the LHS of Eqn.~\ref{eq:vsp-mc}.
\end{proof}
\end{restatable}

This concludes the first part of Theorem~1. 
We now prove the second part.

\subsection{PART II: CLOSED FORM PRIOR}

The following proof makes use of the MDT representation of a VSP introduced in Section~1.1 
and detailed in \ref{sec:mdt-defn} below. 

We next observe that all BDTs representing the same VSP have equal prior probabilities (they collapse to the same MDT and that fixes $S(t)$). This makes it easy to do the sum in (4) 
as the summand is constant.
\begin{restatable}[Probability Distribution over VSPs]{prop}{vspprior}
\label{prop:vsp-prior}
The prior probability for a VSP with $n$ nodes is
\begin{align*}\label{eq:vsp-prior}
    \pi_{\V_{[n]}}(v|q) =|t(v)|&\pi_{\T_{[n]}}(t|q),
\end{align*}
for any tree $t\in t(v)$. 
\begin{proof}[Proposition \ref{prop:vsp-prior}]\label{proof:vsp-prior}
    For $v\in \V_{[n]}$, any two trees $t,t'\in t(v)$ are both in $\T_{[n]}$. They also satisfy $S(t)=S(t')$. This follows from Lemma~1: 
    if these numbers differ then the $S$-clusters of $t$ and $t'$ cannot all have equal sizes; the $S$-cluster sizes of a BDT determine of the numbers of children of the $S$-nodes in its MDT; it follows that $m=t_\M(t)$ and $m'=t_\M(t')$ cannot be isomorphic (identifying leaves by actor labels); but $m$ and $m'$ are then distinct MDT's for $v$ which contradicts Lemma~1. 
    Referring to Eqn.~3 
    we see that $\pi_{\T_{[n]}}(t|q)$ is constant over $t\in t(v)$ so the sum in Eqn.~4 
    just counts trees in $t(v)$.
\end{proof}
\end{restatable}

Finally, we count trees in $t(v)$ and this gives us the closed form we seek. This seems to be new.
\begin{restatable}{prop}{treetotal}
\label{prop:tree-total}
Let $t\in t(v)$ be an arbitrary BDT of a VSP $v\in \V_{[n]}$ with $P$- and $S$-clusters defined as in Theorem~1. 
The number of BDTs of $v$ is
    \begin{equation}\label{eq:tree-total}
         |t(v)|= \prod_{k=1}^{K_P} (|2C_k^{(P)}|-1)!!\prod_{k'=1}^{K^S} \mathcal{C}_{|C_{k'}^{(S)}|}
    \end{equation} 
    with $\mathcal{C}_s,\ s\ge 0$ given in (7). 
    \begin{proof}[Proposition~\ref{prop:tree-total}]
    By Lemma~1 
    the set of BDT trees $t(v)$ for any $v\in \V_{[n]}$ is identical to the set $t_\M(m)=\{t\in \T_{[n]}: m_\T(t)=m\}$ when $m=m_\V(v)$ so we need to count the number of BDT's that collapse down to the same MDT. Let $m=(F,E,L)$ be an MDT with leaves $\F$ and internal nodes $\A$. 
    
    A $P$-node $i\in \A$ in $m$ having $c$ child nodes is generated by collapsing some $P$-cluster $C^{(P)}_k$ of a BDT $t\in t_\M(m)$ with $|C^{(P)}_k|=c-1$ nodes ``internal'' to the $P$-cluster. This $P$-cluster corresponds to a sub-tree $t_k=(V(t_k),E(t_k))$ with vertices $V(t_k)=C^{(P)}_k$ and edges \[E(t_k)=E(t)\cap (C^{(P)}_k\times C^{(P)}_k).\] 
    The sub-tree $t_k$ has $c=|C^{(P)}_k|+1$ leaves. If we replace $t_k$ with any tree with $|C^{(P)}_k|+1$ labelled leaves then it collapses to a MDT node with in- and out-edges isomorphic to those of node $i$ in $m$. The number of such trees is $(2|C^{(P)}_k|-1)!!$.  
    
    An $S$-node $i\in \A$ of the MDT with $s$ child nodes and stacking data $L_i(m)=(i_1,...,i_s)$ is generated by collapsing some $S$-cluster $S^{(S)}_k$ of a BDT. Again, that cluster covers $|S^{(P)}_k|=s-1$ internal nodes in the BDT. This $S$-cluster corresponds to a sub-tree 
    of $t$ with $s=|S^{(P)}_k|+1$ leaves. Since all the internal nodes of the sub-tree are of type $S$ and its leaf nodes are labelled, this sub-tree is a BDT representing the fixed total order $i_1\succ i_2...\succ i_s$ on its leaf nodes. If we replace this subtree with any tree with $s$ labelled leaves representing the same total order then it collapses to a MDT node with in- and out-edges isomorphic to $i$ and the same stacking data. The number of such trees is given by the Catalan number $\mathcal{C}_{s-1}=\mathcal{C}_{|S^{(P)}_k|}$. This can be shown by the following induction. 
    
    The number of BDT's representing a total order on $1$ or $2$ elements is one and indeed $\mathcal{C}_0=\mathcal{C}_1=1$. Suppose the number of BDT's representing a total order $1\succ 2\succ ... \succ s$ is $\mathcal{C}_{s-1}$ and consider a BDT representing $1\succ 2\succ ... \succ s+1$. The root of such a BDT must partition the leaves into $1,...,k$ and $k+1,...,s+1$ for some $1\le k\le s$ so that the root stacks $1,...,k$ above $k+1,...,s+1$. By the induction hypothesis the number of subtrees representing $1\succ 2\succ ... \succ k$ is $\mathcal{C}_{k-1}$ and the number representing $k+1\succ 2\succ ... \succ s+1$ is $\mathcal{C}_{s-k-1}$, so the number of BDT's splitting the leaves into $1,...,k$ and $k+1,...,s+1$ is $\mathcal{C}_{k-1}\mathcal{C}_{s-k-1}$. The total number of
    BDT's representing $1\succ 2\succ ... \succ s+1$ is then
    \begin{align*}
        \sum_{k=1}^s\mathcal{C}_{k-1}\mathcal{C}_{s-k-1}&=\sum_{k=0}^s\mathcal{C}_{k}\mathcal{C}_{s-k}\\
        &=\mathcal{C}_s,
    \end{align*}
    where the last step is given in \citet{Weisstein02}.
    
    The total number of BDT's is given by the product over the internal nodes of the MDT of the numbers of BDT sub-trees which collapse to give those nodes.
    This gives Eqn.~\ref{eq:tree-total} and completes the proof of Theorem~1. 
    \end{proof}
\end{restatable}


\subsection{Multi-Decomposition Trees}\label{sec:mdt-defn}

A MDT $m\in \M_{[n]}$ is a tree $m=(F(m),E(m),L(m))$ with $n$ leaves and edges $E(m)$ directed from the root to the leaves. Let $\F$ and $\A$ be the index sets for the leaves and internal nodes, such that $|\F|=n$ and $1\leq|\A|\leq n-1$. An internal node $i\in\A$ of a MDT may have any number of child nodes between two and $n-1$. For $i \in \F$ and $m\in \M_{[n]}$, the array $F_i(m)\in [n]$ records the actor represented by leaf node $i$. The internal nodes $i\in\A$ are either of type $S$ or type $P$. The key defining feature of an MDT is that the internal nodes of an MDT which are adjacent must have unequal types. 
    
    Let $S(m)$ be the number of $S$-nodes in multi-tree $m\in\M_{[n]}$. For $m\in\M_{[n]}$ let $v(m):\in \V_{[n]}$ map an MDT to its corresponding VSP and for $i\in\F\cup\A$ let $m_i(m)$ denote the sub-tree rooted by node $i$. If $i\in\A$ is of type $P$ with $k$ children $j_1,\dots,j_k$, then $$v(m_i(m))=v(m_{j_1}(m))\oplus\dots\oplus v(m_{j_k}(m)).$$
    If $i\in\A$ is of type $S$ with $k$ child nodes $\{j_1,\dots,j_k\}=\{j\in \F\cup\A: \langle i,j\rangle\in E(m)\}$, an ordered set $L_i=(j_1,\dots,j_k)$ gives the stacking order (with $j_1$ at the top) for the sub-trees rooted by the children of $i$. It follows that
    $$v(m_i(m))=v(m_{j_1}(m))\otimes\dots\otimes v(m_{j_k}(m)).$$
    Let $L(m)=\{L_i\}_{i\in \A}$ with $L_i=\emptyset$ if $i$ is a $P$-node. Adjacent internal nodes have unequal type so if $\langle i,j\rangle\in E(m)$ then exactly one of $L_i$ and $L_j$ is empty. In this notation a MDT tree is a BDT if all its internal nodes have two child nodes and a BDT is an MDT if all adjacent internal nodes have different $S/P$-types.

    An MDT can be formed from a BDT by collapsing edges between internal nodes in the BDT which have the same type while preserving information about stacking order at $S$-nodes. This collapses $P$- and $S$-clusters to a single node. A set of BDT's can be recovered from an MDT by ``unpacking'' internal nodes of the MDT with more than two child nodes in different ways. For $t\in \T_{[n]}$ let $m_\T(t)\in \M_{[n]}$ map the BDT $t$ to its corresponding MDT. See Figure 4 
    for an example.

Counting linear extensions in the MDT formulation is similar to the BDT case (Eqns.~1 \& 2). 

\begin{align}
    |\L(h_1 \otimes \dots \otimes h_n)| = & |\L(h_1)| \times \dots \times |\L(h_n)| \label{eq:le_s_mdt}\\
    |\L(h_1 \oplus \dots \oplus h_n)| = & |\L(h_1)| \times \dots \times |\L(h_n)|{{|V(h_1)|+\dots +|V(h_n)|\choose |V(h_1)|,\dots,|V(h_n)|}}\label{eq:le_p_mdt}
\end{align}
where $|V(h_1)|$ and $|V(h_2)|$ give the number of actors in $h_1$ and $h_2$. This may be evaluated recursively in $O(n)$ steps.

\subsection{Proof of Proposition~\ref{prop:posteriors}
} \label{proof:prop:posteriors}

\posteriors*

\begin{proof}[Proposition \ref{prop:posteriors}]
 Eqn.~11 
 is the marginal over $t\in t(v)$ 
    of Eqn.~10: 
    if $(t,q,\psi)\sim \pi_{\T_{[n]}}(\cdot|y)$ then the new joint distribution at $v(t)=v$ is
    \begin{align*}
        p(v,q,\psi)&\propto \sum_{t'\in t(v)} \pi_{\T_{[n]}}(t'|q)\pi(q,\psi) Q(y|v(t'),\psi)\\
        &=\pi(q,\psi) Q(y|v,\psi)\sum_{t'\in t(v)} \pi_{\T_{[n]}}(t'|q)\\
        &=\pi_{\V_{[n]}}(v,q,\psi|y)
    \end{align*}
    as $Q(y|v(t),\psi)=Q(y|v,\psi)$ is a constant for $t\in t(v)$ and the prior marginalises to $\pi_{\V_{[n]}}(v|q)$ by Eqn.~4. 
\end{proof}

\newpage
\section{Queue-Jumping Models}

\subsection{Queue-Jumping Up/Down Observation Model}\label{sec:QJ-U-appendix}

Let $L_T(v)=|\L[v]|$ be the number of linear extensions of VSP $v\in\V_{\tilde [n]}$ and for $i\in [n]$ let $T_i(v)=|\{l\in\L[v]:l_1=i\}|$ give the number of linear extensions with actor $i$ at the top. The observation model for QJ-U for a generic list $x\in\P_{[n]}$ is
\begin{equation}\label{eq:QJ-up}
   Q_{up}(x|v,p)=\prod_{i=1}^{n-1}\left(\frac{p}{n-i+1}+(1-p)\frac{T_{x_i}(v[x_{i:n}])}{L_T(v[y_{i:n}])}\right). 
\end{equation}
We can interpret this as the distribution over lists determined by a process in which the list is formed by building it up one element at a time from the top, choosing the next actor at random from those that remain with probability $p$ and otherwise choosing the next actor as the first actor in a list drawn from the noise free model (beginning of Section~3) 
applied to the remaining actors. \Fig~\ref{fig:QJU} gives an example list realisation for VSP $v_0$. We give the generative model alg.\ref{alg:jump-up}.

\begin{minipage}{\linewidth}
    \centering
    \begin{tikzpicture}[thick,scale=.9, every node/.style={scale=0.8}]
        \node[draw, circle, minimum width=.6cm] (1) at (0, 1) {$1$};
        \node[draw, circle, minimum width=.6cm] (2) at (-1, -0.5) {$2$};
        \node[draw, circle, minimum width=.6cm] (3) at (1, 0) {$3$};
        \node[draw, circle, minimum width=.6cm] (4) at (1, -1) {$4$};
        \node[draw, circle, minimum width=.6cm] (5) at (0, -2) {$5$};
        \draw [solid] (1.75,-2.5) -- (1.75,1.5);
        \node[draw, circle, dashed, color = red, minimum width=.6cm] (6) at (2.5, 1) {};
        \node[draw, circle, dashed, minimum width=.6cm] (7) at (2.5, 0.25) {};
        \node[draw, circle, dashed, minimum width=.6cm] (8) at (2.5, -0.5) {};
        \node[draw, circle, dashed, minimum width=.6cm] (9) at (2.5, -1.25) {};
        \node[draw, circle, dashed, minimum width=.6cm] (10) at (2.5, -2) {};
        \node[draw, circle, color=blue, minimum width=.6cm] (21) at (4, 1.5) {$1$};
        \node[draw, circle, color=blue, minimum width=.6cm,label=below:\textcolor{blue}{\small $j\sim \mathcal{U}\{1,2,3,4,5\}$}] (22) at (4, .5) {$j$};
        \node[draw, circle, minimum width=.6cm] (11) at (5.5, 1) {$1$};
        \node[draw, circle, dashed, color = red, minimum width=.6cm] (12) at (5.5, 0.25) {};
        \node[draw, circle, color=blue, minimum width=.6cm,label=above:\textcolor{blue}{\small $i\in \{2,3\}\sim (\frac{1}{3},\frac{2}{3})$}] (23) at (7, .75) {$i$};
        \node[draw, circle, color=blue, minimum width=.6cm,label=below:\textcolor{blue}{\small $j\sim \mathcal{U}\{2,3,4,5\}$}] (24) at (7, -.25) {$j$};
        \node[draw, circle, dashed, minimum width=.6cm] (13) at (5.5, -0.5) {};
        \node[draw, circle, dashed, minimum width=.6cm] (14) at (5.5, -1.25) {};
        \node[draw, circle, dashed, minimum width=.6cm] (15) at (5.5, -2) {};
        \node[draw, circle, minimum width=.6cm] (25) at (8.5, 1) {$1$};
        \node[draw, circle, minimum width=.6cm] (26) at (8.5, 0.25) {$5$};
        \node[draw, circle, color=blue, minimum width=.6cm,label=above:\textcolor{blue}{\small $i\in \{2,3\}\sim (\frac{1}{3},\frac{2}{3})$}] (30) at (10, 0) {$i$};
        \node[draw, circle, color=blue, minimum width=.6cm,label=below:\textcolor{blue}{\small $j\sim \mathcal{U}\{2,3,4\}$}] (31) at (10, -1) {$j$};
        \node[draw, circle, dashed, color=red, minimum width=.6cm] (27) at (8.5, -0.5) {};
        \node[draw, circle, dashed, minimum width=.6cm] (28) at (8.5, -1.25) {};
        \node[draw, circle, dashed, minimum width=.6cm] (29) at (8.5, -2) {};
        \node[draw, circle, minimum width=.6cm] (32) at (11.5, 1) {$1$};
        \node[draw, circle, minimum width=.6cm] (33) at (11.5, 0.25) {$5$};
        \node[draw, circle, color=blue, minimum width=.6cm] (37) at (13, -.75) {$3$};
        \node[draw, circle, color=blue, minimum width=.6cm,label=below:\textcolor{blue}{\small $j\sim \mathcal{U}\{3,4\}$}] (38) at (13, -1.75) {$j$};
        \node[draw, circle, minimum width=.6cm] (34) at (11.5, -0.5) {$2$};
        \node[draw, circle, dashed, color=red,minimum width=.6cm] (35) at (11.5, -1.25) {};
        \node[draw, circle, dashed, minimum width=.6cm] (36) at (11.5, -2) {};
        \node[draw, circle, minimum width=.6cm] (16) at (14.5, 1) {$1$};
        \node[draw, circle, minimum width=.6cm] (17) at (14.5, 0.25) {$5$};
        \node[draw, circle, minimum width=.6cm] (18) at (14.5, -0.5) {$2$};
        \node[draw, circle, minimum width=.6cm] (19) at (14.5, -1.25) {$4$};
        \node[draw, circle, color=red, minimum width=.6cm] (20) at (14.5, -2) {$3$};
        \node at ($(27)!.7!(18)$) {\ldots};
        \draw[-latex] (1) -- (2);
        \draw[-latex, color=red, dotted] (21) -- (6) node[midway,above,sloped] {\textcolor{blue}{\small{$1-p$}}};
        \draw[-latex, color=blue, dotted] (22) -- (6) node[midway,below,sloped] {\textcolor{blue}{\small{$p$}}};
        \draw[-latex, color=blue, dotted] (23) -- (12) node[midway,above,sloped] {\textcolor{blue}{\small{$1-p$}}};
        \draw[-latex, color=red, dotted] (24) -- (12) node[midway,below,sloped] {\textcolor{blue}{\small{$p$}}};
        \draw[-implies,double equal sign distance] (3,-.5) -- (5,-.5);
        \draw[-latex, color=red, dotted] (30) -- (27) node[midway,above,sloped] {\textcolor{blue}{\small{$1-p$}}};
        \draw[-latex, color=blue, dotted] (31) -- (27) node[midway,below,sloped] {\textcolor{blue}{\small{$p$}}};
        \draw[-latex, color=blue, dotted] (37) -- (35) node[midway,above,sloped] {\textcolor{blue}{\small{$1-p$}}};
        \draw[-latex, color=red, dotted] (38) -- (35) node[midway,below,sloped] {\textcolor{blue}{\small{$p$}}};
        \draw[-implies,double equal sign distance] (3,-.5) -- (5,-.5);
        \draw[-implies,double equal sign distance] (6,-.5) -- (8,-.5);
        \draw[-implies,double equal sign distance] (9,-.5) -- (11,-.5);
        \draw[-implies,double equal sign distance] (12,-.5) -- (14,-.5);
        \draw[-latex] (2) -- (5);
        \draw[-latex] (1) -- (3);
        \draw[-latex] (3) -- (4);
        \draw[-latex] (4) -- (5);
        \draw[-latex] (6) -- (7);
        \draw[-latex] (7) -- (8);
        \draw[-latex] (8) -- (9);
        \draw[-latex] (9) -- (10);
        \draw[-latex] (11) -- (12);
        \draw[-latex] (12) -- (13);
        \draw[-latex] (13) -- (14);
        \draw[-latex] (14) -- (15);
        \draw[-latex] (16) -- (17);
        \draw[-latex] (17) -- (18);
        \draw[-latex] (18) -- (19);
        \draw[-latex] (19) -- (20);
        \draw[-latex] (25) -- (26); 
        \draw[-latex] (26) -- (27); 
        \draw[-latex] (27) -- (28); 
        \draw[-latex] (28) -- (29); 
        \draw[-latex] (32) -- (33); 
        \draw[-latex] (33) -- (34); 
        \draw[-latex] (34) -- (35); 
        \draw[-latex] (35) -- (36); 
    \end{tikzpicture}
    \captionof{figure}{One example list simulation process from the VSP $v_0$ (left) via the QJ-U observation model. The simulated list is displayed on the right. }\label{fig:QJU}
\end{minipage}

\begin{algorithm}[H]
    \caption{Simulation algorithm for QJ-U.}\label{alg:jump-up}
    \begin{algorithmic}
        \Require $v\in \V_{[n]}, p\in [0,1]$
        \Ensure $x\sim Q_{up}(\cdot|v,p)$ \vskip 0.05in
        \State $i\gets 1, s\gets [n], v'\gets v$
        \While{$|s|>0$}
        \State $q \gets (T_j(v')/L(v'))_{j\in s}$
        \State Sample $c\sim Bern(1-p)$
        \If{$c=0$}
            \State Sample $x_i \sim \U(s)$ 
        \ElsIf{$c=1$}
            \State Sample $x_i\sim \text{multinom}(q)$
        \EndIf
        \State $s\gets s\backslash x_i$
        \State $i\gets i+1, v'\gets v[s]$
        \EndWhile
        \Return $x = (x_1,\dots,x_n)$
    \end{algorithmic}
\end{algorithm}

The output $x\sim Q_{up}(\cdot|v,p)$ is a random list of $n$ elements distributed according to $Q_{up}$. This follows because the probabilities to choose entries in $x$ at each step are just the factors in $Q_{up}$ in Eqn.~\ref{eq:QJ-up}. We can turn the model around and build the list from the bottom, allowing ``queue jumping-down''. If we set $p=0$, we get a telescoping product and $Q_{up}(l|v,p=0)=1/L_T(v)$ for $l\in \L[v]$, so we recover the error-free model. Lists are assumed to be drawn independently, and the actors present $o_j,\ j=1,...,N$ are known, so the likelihood is
\[
Q(y|v,p)=\prod_{j=1}^N Q(y_j|v[o_j],p).
\]
Here $Q=Q_{up}$ (and $Q=Q_{bi}$ in the next section).

\subsection{Bi-Directional Queue-Jumping model}\label{sec:QJ-B-app-sim}

Similar to QJ-U, QJ-B ranks by repeated selection - but from both ends. We either rank from the top with probability $\phi$ or from the bottom with probability $1-\phi$. From the top (bottom) of the list, the next actor is chosen at random from those that remain with probability $p$, and otherwise as the first (last) actor in a list drawn from the noise free model. An example simulation process from VSP $v_0$ is visualised in \Fig~5. 

Algorithm \ref{alg:bi-directn} gives the simulation algorithm for the bi-directional queue jumping model. It introduces one extra step in each loop of algorithm \ref{alg:jump-up} in which we randomly choose the top/bottom fill-direction to place the next actor in the realised list with probability $\phi$. 

\begin{algorithm}[H]
    \caption{Simulation algorithm for QJ-B.}\label{alg:bi-directn}
    \begin{algorithmic}
        \Require $v\in\V_{[n]}, p\in[0,1], \phi\in[0,1]$
        \Ensure $x\sim Q_{bi}(x|v,p,\phi)$\vskip 0.05in
        \State $s\gets [n], v'\gets v$
        \State $x \gets (\emptyset,\dots,\emptyset)\in \{\emptyset\}^n$, $k\gets 1$, $U_0\gets 0$, $D_0\gets n+1$
        \While{$|s|>0$}
            \State Sample $z_k\sim Bern(1-\phi)$
            \State $U_k\gets U_{k-1}+z_k$, $D_k\gets D_{k-1}-(1-z_k)$
            \State $i_k\gets z_k U_k + (1-z_k)D_k$
            \State Sample $c_k\sim Bern(1-p)$
            \If{$c_k=0$}
                \State Sample $a \sim \U(s)$
            \Else
                \If{$z_k=0$}
                    \State $q \gets (T_a(v')/L_T(v'))_{a\in s}$
                    \State Sample $a\sim \text{multinom}(q)$
                \ElsIf{$z_k=1$}
                    \State $q \gets (B_a(v')/L_T(v'))_{a\in s}$
                    \State Sample $a\sim \text{multinom}(q)$
                \EndIf
            \EndIf
            \State Set $x_{i_k}\gets a,\ k\gets k+1,\ s\gets s\backslash a,\ v'\gets v[s]$
        \EndWhile
        \Return $x = (x_1,\dots,x_n)$
    \end{algorithmic}
\end{algorithm} 

\newpage

\subsection{Recursive Evaluation Algorithm for QJ-B}\label{sec:QJ-B-app-eval}

This sub-section gives Algorithm~\ref{alg:recur-bi-directn}, an algorithm for recursive evaluation of the QJ-B likelihood.

\begin{algorithm}[H]
    \caption{Recursion evaluating $Q_{bi}$ in Eqn.~9
    }\label{alg:recur-bi-directn}
    \begin{algorithmic}
    \Procedure{$f$}{$v, x,p,\phi$}
        \State $n = |v|$
        \If{$n = 1$}
            \Return 1
        \EndIf
        \If{$\phi > 0$}
            \State $l_0 \gets \frac{p}{n}+(1-p)\frac{T_{x_1}(v)}{L_T(v)}$
            \State $x\gets x_{2:n}, v\gets v[x]$
            \State $P_0 = \phi\times l_0 \times$ \Call{$f$}{$v,x,p,\phi$}
            \Else{$\text{ } L_0 = 0$}
        \EndIf
        \If{$\phi < 1$}
            \State $l_1 \gets \frac{p}{n}+(1-p)\frac{B_{x_n}(v)}{L_T(v)}$
            \State $x\gets x_{1:n-1}, v\gets v[x]$
            \State $P_1 = (1-\phi)\times l_1 \times$ \Call{$f$}{$v,x,p,\phi$}
        \Else $\text{ }P_1 = 0$
        \EndIf
        \Return $P_0 + P_1$
        \EndProcedure        
    \end{algorithmic}
\end{algorithm} 

We now show this algorithm is correct.

Let $X\sim Q_{bi}$ be a random list with realisation $X=x$. For sub-list $x_{a:b},\ 1\leq a<b\leq n$ let 
    \begin{align*}
    P_{a|a:b}&=p(X_a = x_a|z_k=0,v[x_{a:b}],p),\\
    P_{b|a:b}&=p(X_b = x_b|z_k=1,v[x_{a:b}],p),\\
    P_{a:b}&=p(X_{a:b}=x_{a:b}|v[x_{a:b}],p,\phi),
    \end{align*}
    so that $P_{1:n} = Q_{bi}(x|v,p,\phi)$ and $P_a = 1$ when $a=b$. 

\setcounter{prop}{6}
\begin{restatable}{prop}{priorbd}
\label{prop:bi-directn}
\begin{equation}\label{eq:L-recursion}
        P_{a:b}=\phi P_{a|a:b}P_{a+1:b}+(1-\phi)P_{b|a:b}P_{a:b-1},
    \end{equation}
    and $f(v,x,p,\phi)$ in Algorithm~\ref{alg:recur-bi-directn} returns $Q_{bi}(x|v,p,\phi)$.  
\end{restatable}

\begin{proof}[Proposition \ref{prop:bi-directn}] 
    First of all, if Eqn.~\ref{eq:L-recursion} holds then a call to $f(v[x_{a:b}],x_{a:b},p,\phi)$
    evaluates $l_0=P_{a|a:b}$, $l_1=P_{b|a:b}$ and returns
    the sum of 
    $\phi l_0 f(v[x_{a:b}],x_{a:b},p,\phi)$
    and $(1-\phi) l_1 f(v[x_{a:b-1}],x_{a:b-1},p,\phi)$. Then since $f(v[x_a],x_{a},p,\phi)=P_a=1$ we have by induction (and Eqn.~\ref{eq:L-recursion}) that $f(v[x_{a:b}],x_{a:b},p,\phi)=P_{a:b}$ and 
    \[f(v,x,p,\phi)=Q_{bi}(x|v,p,\phi).\]
    We now show Eqn.~\ref{eq:L-recursion}) holds for the distribution of sub-lists $X_{a:b}$ of $X\sim Q_{bi}$.
    If $a:b$ remain to be realised then $a-1+n-(b-1)$ entries in $X$ have been realised and this would occur as we enter step $k=n+a-b+1$ of Algorithm~\ref{alg:bi-directn}. Partitioning on the value of $z_k$,
    \begin{align*}
        P_{a:b}&=p(X_{a:b}=x_{a:b}|v[x_{a:b}],p,\phi) \\
        & = p(z_k = 0|\phi)p(X_{a:b}=x_{a:b}|z_k = 0,v[x_{a:b}],p,\phi)\\
        & \quad +\ 
        p(z_k = 1|\phi)p(X_{a:b}=x_{a:b}|z_k = 1,v[x_{a:b}],p,\phi)\\
        & = \phi P_{a|a:b} p(x_{a+1:b}|v[x_{a+1:b}],p,\phi) 
        \\
        & \quad +\ (1-\phi)P_{b|a:b}p(x_{a:b-1}|v[y_{a:b-1}],p,\phi),\\
        &=\phi P_{a|a:b}P_{a+1:b}+(1-\phi)P_{b|a:b}P_{a:b-1}.
    \end{align*}
\end{proof}

\newpage
\section{MCMC Sampler}\label{sec:MCMC-app}

We use Metropolis-Hasting MCMC to sample posterior distributions. We can target either distribution in Proposition~\ref{prop:posteriors}. 

\subsection{MCMC sampler in the BDT representation}\label{sec:MCMC-app-BDT}

We start with MCMC targeting BDT. This was the method we implemented as the data structures seem slightly simpler. However, we would expect MCMC targeting the VSP posterior directly to be a little more efficient, as MCMC targeting the BDT posterior wastes time exploring latent subspaces $t(v)$ without changing $v$. Tree sampling requires edge operations on trees (called ``subtree prune and regraft'' (OP-PR) in the phylogenetics literature). 
For this purpose we assume the $0$-node with an edge to the root of the BDT is restored, so $0\in \F$ for a regraft above the root. Let $\F_{-0}=\F\setminus \{0\}$ and $E_{-0}(t)=E(t)\setminus\{\langle e_1,e_2\rangle\in E(t): e_1=0\}$.

\begin{definition}[Subtree Prune and Regraft on a BDT ]\label{defn:tree-op}
For $t = (F(t),E(t),L(t)),\ t\in\T_{[n]}$ a BDT with leaf node labels $\F$ and internal node labels $\A$, an edge operation $t' = t \triangleleft_e (e,e')$ moves edge $e=\langle e_1,e_2\rangle,\ e\in E_{-0}(t)$ to edge $e' = \langle e'_1, e'_2 \rangle,\ e'\in E(t')$. The leaf-to-actor map $F(t')=F(t)$ is unchanged. Let \[f_p(j|t)=\{i\in \A|\langle i,j \rangle \in E(t)\}\] give the parent of $j \in \F_{-0}\cup \A$ with $f_p(r|t)=0$ if $r$ is the root. Let \[f_c(i|t)=\{j_1,j_2\in\F\cup\A|\{\langle i,j_1 \rangle, \langle i,j_2 \rangle \subset E(t)\}\] give the children of $i\in \A$.
 Let $\cev{e}_{1}=f_p(e_1|t)$ give the parent of $e_1$ and $\vec e_2=f_c(e_1|t)\backslash \{e_2\}$ give the ``sibling'' of $e_2$ in $t$ (the child of $e_1$ which is not $e_2$). Then
\begin{align*}
    E(t')=E(t)&\backslash \{e',\langle \cev e_1,e_1\rangle,\langle e_1, \vec e_2\rangle\}\\
    & \cup \{\langle e_1',e_1\rangle, \langle e_1,e_2'\rangle, \langle \cev e_{1},\vec e_2\rangle \}.
\end{align*}
Set $L(t')=L(t)$ and make the following replacements as needed. If $L_{\cev e_1}(t)\neq \emptyset$ then $L_{\cev e_1}(t)$ is an ordered set containing two edges. Set $L_{\cev e_1}(t')=L_{\cev e_1}(t)\backslash \{e_1\}\cup \{\vec e_2\}$ where the replacement enters the vacated position in the ordered set. If $L_{e_1'}(t)\neq \emptyset$, $L_{e_1'}(t')=L_{e_1'}(t)\backslash \{e_2'\} \cup \{e_1\}$. If $L_{e_1}(t)\neq \emptyset$ then take $L_{e_1}(t')\sim\U\{(e_2,e_2'),(e_2',e_2)\}$. 
\end{definition}

The edge operation $t \triangleleft_e (e,e')$ moves the sub-tree rooted by $e_2$ into edge $e'$, breaking that edge and inserting node $e_1$. The $S/P$-type of $e_1$ travels with $e_1$, and if it is $S$ we must assign a stacking order to the subtrees rooted by $e'_2$ and $e_2$. Figure \ref{fig:edge-op-BDT} illustrates an example edge operation. 

\begin{minipage}{\linewidth}
    \centering
    \begin{tikzpicture}[thick,scale=.7, every node/.style={scale=0.6}]
        \node[draw, circle, minimum width=1cm,fill=pink] (1) at (0, 2) {$S$};
        \node[draw, circle, minimum width=1cm] (2) at (-1, 1) {$1+$};
        \node[draw, circle, minimum width=1cm,fill=pink] (3) at (1, 1) {$S-$};
        \node[draw, circle, dashed,minimum width=1cm,fill=cyan] (4) at (0, 0) {$P+$};
        \node[draw, circle, minimum width=1cm] (5) at (2, 0) {$5-$};
        \node[draw, circle, minimum width=1cm] (6) at (-1, -1) {$2$};
        \node[draw, circle, dashed,minimum width=1cm,fill=pink] (7) at (1, -1) {$S$};
        \node[draw, circle, dashed,minimum width=1cm] (8) at (0, -2) {$3+$};
        \node[draw, circle, dashed,minimum width=1cm] (9) at (2, -2) {$4-$};
        \draw[-latex] (1) -- (3);
        \draw[-latex][red] (1) -- (2) node[midway,above] {\textcolor{blue}{$e'$}};
        \draw[-latex][red] (3) -- (4);
        \draw[-latex] (3) -- (5);
        \draw[-latex] (4) -- (6);
        \draw[-latex] (4) -- (7) node[midway,above] {\textcolor{blue}{$e$}};
        \draw[-latex][red] (7) -- (8);
        \draw[-latex] (7) -- (9);
        \draw[-implies,double equal sign distance] (2.65,0) -- (4.15,0) node[midway,above]{$t \triangleleft_e (e, e')$};
        \node[draw, circle, minimum width=1cm,fill=pink] (10) at (6.5, 1.8) {$S$};
        \node[draw, circle, minimum width=1cm] (11) at (4.7, -0.5) {$1$};
        \node[draw, circle, minimum width=1cm,fill=pink] (12) at (7.5, .7) {$S-$};
        \node[draw, circle, dashed, minimum width=1cm,fill=cyan] (13) at (5.5, .7) {$P+$};
        \node[draw, circle, minimum width=1cm] (14) at (8.3, -.5) {$5-$};
        \node[draw, circle, minimum width=1cm] (15) at (7, -.5) {$2+$};
        \node[draw, circle, dashed, minimum width=1cm,fill=pink] (16) at (6, -.5) {$S$};
        \node[draw, circle, dashed, minimum width=1cm] (17) at (5.3, -1.8) {$3+$};
        \node[draw, circle, dashed,minimum width=1cm] (18) at (6.7, -1.8) {$4-$};
        \draw[-latex][red] (10) -- (13);
        \draw[-latex] (10) -- (12);
        \draw[-latex] (13) -- (16);
        \draw[-latex] (13) -- (11);
        \draw[-latex][red] (12) -- (15);
        \draw[-latex] (12) -- (14);
        \draw[-latex][red] (16) -- (17);
        \draw[-latex] (16) -- (18);
    \end{tikzpicture}
    \captionof{figure}{An example OP-PR edge operation on BDT $t_0$.}\label{fig:edge-op-BDT}
\end{minipage}

The tree updates in our MCMC admit both local and global edge operations. In the local edge operation, an edge can only be moved to a neighboring edge, i.e. if $e = \langle e_1, e_2\rangle$, $e'$ is selected from $e$'s neighboring edges $E_l(e|t)$ such that
\begin{align*}
    E_l(e|t)=\{\langle e'_1,e'_2 \rangle \!\in\! E(t)\mid e'_2 \!=\! \cev e_1 \text{ or }
    e'_1 \!=\! \vec e_2 \text{ or } e'_2\!=\!\vec e_1 \}.
\end{align*}
These ``small'' changes have a higher acceptance rate.
The global edge operation moves an edge $e$ to any $e' \in E(t)\backslash e$. For $t\in\T_{[n]}$, we typically perform 1 global edge operation for every $n$ local edge operations. We present the MCMC algorithm for BDT with the QJ-B observation model in Algorithm \ref{alg:MCMC-BDT}, omitting the standard $q,p$ and $\phi$ updates. A simple internal node type update is included. The algorithm for QJ-U observation model is similar but without the $\phi$-update. 

\begin{algorithm}
    \caption{The MCMC algorithm for the BDT with QJ-B observation model at step $k$.}\label{alg:MCMC-BDT}
    \begin{algorithmic}
        \Require $y,t^{(k-1)}=t, q^{(k-1)}=q, p^{(k-1)}=p,\phi^{(k-1)}=\phi$ with $t=(F(t),E(t),L(t)),\ t\in \T_{[n]}$.
        \Ensure \\
        {\centering\vspace*{-0.1in}
         $
        \begin{aligned}
            &t^{(k)}\sim \pi(t|y,q,p,\phi), \\
            &q^{(k)}\sim\pi(q|y,t^{(k)},p,\phi),\\
            &p^{(k)}\sim\pi(p|y,t^{(k)},q^{(k)},\phi),\\
            &\phi^{(k)}\sim\pi(\phi|y,t^{(k)},q^{(k)},p^{(k)})\\[0.05in]
        \end{aligned}
        $
        \par}
        \Function{type}{$i|t$}
            \If {$L_i(t)=\emptyset$}
            \Return $P$
            \Else \Return $S$
            \EndIf
        \EndFunction
        
        \hrulefill\emph{Update for $t$ (internal node type)}\hrulefill
        \State $t'\gets t^{(k)}\gets t$
        \State Sample $i\sim \U(\A)$
        \If {\Call{type}{$i|t$}=$P$}
            \State Sample $z\sim \U\{0,1\}$
        \begin{align*}
        L_i(t')&\gets zf_c(i|t)[(1,2)]+(1-z)f_c(i|t)[(2,1)]\\    
            \eta_1&\gets\frac{2\times Q(y|v(t'),p,\phi)\pi_{\T_{[n]}}(t'|q)}{Q(y|v(t),p,\phi)\pi_{\T_{[n]}}(t|q)}\\[-0.3in]
            \end{align*}
        \ElsIf{\Call{type}{$i|t$}=$S$}
        
        \State  \vspace*{-0.25in}
            \begin{align*}
            L_i(t')&\gets \emptyset\\
            \eta_1&\gets\frac{Q(y|v(t'),p,\phi)\pi_{\T_{[n]}}(t'|q)}{2Q(y|v(t),p,\phi)\pi_{\T_{[n]}}(t|q)}\\[-0.3in]
            \end{align*}
        \EndIf
        \If {$\U(0,1)\le\eta_1$}
            \State $t\gets t^{(k)} \gets t'$
        \EndIf
        
        \hrulefill\emph{Update for $t$ (global edge operation)}\hrulefill
        \State Sample $e\sim \U(E_{-0}(t))$, $e' \sim \U(E(t)\backslash e)$
        \State \vspace*{-0.25in}
            \begin{align*}
            t'&\gets t\triangleleft_e (e,e')\\
            \eta_2 &\gets \frac{Q(y|v(t'),p,\phi)\pi_{\T_{[n]}}(t'|q)}{Q(y|v(t^{(k)}),p,\phi)\pi_{\T_{[n]}}(t|q)}\\[-0.3in]
            \end{align*}
        \If {$\U(0,1)\le \eta_2$}
            \State $t\gets t^{(k)} \gets t'$
        \EndIf

        \hrulefill\emph{Update for $t$ (local edge operation)}\hrulefill
        \State Sample $e\sim \U(E_{-0}(t)), e' \sim \U(E_l(e|t))$
        \State \vspace*{-0.25in}
            \begin{align*}
            t'&\gets t\triangleleft_e (e,e')\\
            \eta_3 &\gets \frac{Q(y|v(t'),p,\phi)\pi_{\T_{[n]}}(t'|q)|E_l(e|t)|}{Q(y|v(t),p,\phi)\pi_{\T_{[n]}}(t|q)|E_l(e|t')|}\\[-0.3in]
            \end{align*}
        \If {$\U(0,1)\le \eta_3$}
            \State $t\gets t^{(k)} \gets t'$
        \EndIf

        \hrulefill\emph{Updates for $q,p$ and $\phi$ omitted}\hrulefill
    \end{algorithmic}
\end{algorithm}

        

        

        
        

\subsection{MCMC sampler in the MDT representation}\label{sec:MCMC-app-MDT}

We can target the VSP-posterior directly. Since MDT's are one-to-one with VSP's, we can parameterise using MDT's and define (in Defn.~\ref{defn:edge-op-MDT}) a sub-tree prune and regraft operator for MDT's.

\begin{definition}[Subtree Prune and Regraft on a MDT]\label{defn:edge-op-MDT}
    For $m = (F(m),E(m),L(m)),\ m\in\M_{[n]}$ a MDT with leaf node labels $\F$ and internal nodes labels $\A$, an edge operation $m' = m\triangleleft_e (e,i)$ creates a new MDT with nodes $\F',\A'$, moving edge $e=\langle e_1,e_2\rangle,\ e\in E_{-0}(m)$ onto node $i\in (\F\cup\A)\backslash \{e_1,e_2\}$.
    
    We need at most $2n$ node labels below. Assume $\F_{-0}\cup\A\subset [2n]$ and let $pop(\F,\A)=\min([2n]\setminus (\F\cup\A))$ be a function we call when we need a new node label. There are three types of edge operation. 
    \begin{enumerate}
        \item $i \in \A$: we connect $e$ to node $i$.\\ 
        Here $F(m')=F(m)$ and 
        \[E(m')=E(m)\backslash \{e\} \cup \langle i,e_2 \rangle.\] Set $L(m')=L(m)$ and make the following changes as needed. If $L_{e_1}(m)\neq \emptyset$ then set $L_{e_1}(m')=L_{e_1}(m)\backslash \{e_1\}$. If $L_{i}(m)\ne \emptyset$ then suppose $L_{i}(m)=(j_1,\dots,j_k)$. Take $L_i(m')\sim \U\{(e_1,j_1,\dots,j_k),\dots, (j_1,\dots,j_k, e_1)\}$ (insert the subtree below $\langle e_1,e_2\rangle$ uniformly in the stack under $i$). 
        \item $i \in \F$: we connect $e$ into edge $\langle \cev i,i \rangle$ with $\cev i=f_p(i|m)$ and add an additional internal node $j = pop(\F,\A)$.\\
        Here $F(m')=F(m)$ and 
        \[E(m')=E(m)\backslash \{e,\langle \cev i,i \rangle\} \cup \{\langle \cev i,j\rangle,\langle j,i\rangle,\langle j,e_2 \rangle\}.\]
        Set $L(m')=L(m)$ and make the following changes as needed.
        If $L_{e_1}(m)\neq \emptyset$ then set $L_{e_1}(m')=L_{e_1}(m)\setminus \{e_1\}$.
        If $L_{\cev i}(m)\neq \emptyset$ (parent is $S$), suppose $L_{\cev i}(m)=(j_1,\dots,i,\dots,j_k)$. Set $L_{\cev{j}}(m')=(j_1,\dots,j,\dots,j_k)$ and $L_j(m')=\emptyset$ (new child is $P$). Finally, if $L_{\cev i}(m)=\emptyset$ (parent is $P$), take $L_j(m)\sim \U\{(i,e_2),(e_2,i)\}$ (new child is $S$). 
        \item $i = 0$: connect $e$ into the edge above the root, $r=f_c(0|m),\ r\in \A$ and add an additional internal node $j=pop(\F,\A)$ which will root $m'$.\\ 
        Here $F(m')=F(m)$ and 
        \[E(m')=E(m)\backslash e \cup \{\langle 0,j \rangle,\langle j,r \rangle,\langle j,e_2\rangle\}.\]
        Set $L(m')=L(m)$ and make the following changes as needed. If $L_{e_1}(m)\neq \emptyset$ then set $L_{e_1}(m')=L_{e_1}(m)\backslash \{e_1\}$.
        If $L_r(m)\neq \emptyset$ (child is $S$), we define $L_j(m')=\emptyset$ (new node is $P$). Otherwise, if $L_{r}(m)=\emptyset$ (child is $P$), we take $L_j(m')\sim\U\{(r,e_2),(e_2,r)\}$ (new node is $S$).
    \end{enumerate}
\end{definition}

Figure \ref{fig:edge-op-MDT} illustrates an example edge operation on a MDT. Moving an edge $e=\langle e_1,e_2\rangle$ may increase or decrease the number of edges and internal nodes. For example, if in case (1) $f_c(e_1|m)=\{e_2,\vec e_2\}$, moving $e$ replaces $\langle \cev e_1,e_1\rangle,\langle e_1,\vec e_2\rangle$ with $\langle \cev e_1,\vec e_2\rangle$ and $e_1$ is removed. If $e$ is attached in an existing internal node $i\in \A$ then the number of nodes and edges each go down by one. 

If we take $e\sim \U(E_{-0}(m))$ and $i\sim \U[(\F\cup \A)\setminus\{e_1,e_2\}]$ and set $m' = m\triangleleft_e (e,i)$ as given in Defn.~\ref{defn:edge-op-MDT} then the proposal probability $\rho(m'|m)$ depends on $e$ and $i$.
A simple generic expression is 
\begin{equation}\label{eq:MDT-proposal-prob}
    \rho(m'|m)=\frac{1}{|E(m)|}\times \frac{1}{|\F\cup \A|-2}\times \rho_{m,m'}
\end{equation}
 where $\rho_{m,m'}$ is given as follows: (Case 1) $\rho_{m,m'}=1/(c_i+1)$ if $i$ is internal and has $c_i$ child nodes and type $S$ ($e_1$ must be placed in the stack below $i$) and $\rho_{m,m'}=1$ if $i$ is internal and type $P$; (Case 2) $\rho_{m,m'}=1/2$ if $i$ is a leaf and $\cev i$ is type $P$ (as $i$ and $e_2$ must be stacked) and $\rho_{m,m'}=1$ if $i$ is leaf and $\cev i$ is type $S$; (Case 3) $\rho_{m,m'}=1/2$ if $i=0$ and $r=f_c(0|m)$ is type $P$ (as $r$ and $e_2$ must be stacked) and $\rho_{m,m'}=1$ if $i=0$ and $r$ is type $S$. 
 
 Not every operation is admissible: if $f_c(e_1|m)=\{e_2,\vec e_2\}$ and $\vec e_2$ is not a leaf, then $\vec e_2$ and $\cev e_1$ must have the same type. An edge  $\langle \cev e_1,\vec e_2\rangle$ would then connect two internal nodes of the same type and so $m'\not\in \M_{[n]}$. In Eqn.~\ref{eq:MDT-proposal-prob}, $\rho(m'|m)$ has a simple form because we do not ``keep trying till we get $m'\in\M_{[n]}$''. We know $m'\not\in \M_{[n]}$ is a possible outcome for $m'$, but we don't try to write down $\rho(m'|m)$ in this case as these proposals will be rejected without the need to evaluate $\rho(m'|m)$. 
 
\begin{minipage}{\linewidth}
    \centering
    \begin{tikzpicture}[thick,scale=.7, every node/.style={scale=0.6}]
        \node[draw, circle,fill=pink, minimum width=1cm,label={\small $i_2$}] (12) at (3.5, 1.8) {$S$};
        \node[draw, circle, minimum width=1cm,label={\small 1,$i_1$}] (13) at (2, 0) {$1$};
        \node[draw, circle, minimum width=1cm,fill=cyan,label={\small 2}] (14) at (3, 0) {$P$};
        \node[draw, circle, minimum width=1cm,label={\small 3}] (15) at (4, 0) {$5$};
        \node[draw, circle, minimum width=1cm,label={\small 4}] (16) at (5, 0) {$6$};
        \node[draw, circle,dashed, minimum width=1cm] (17) at (2, -2) {$2$};
        \node[draw, circle, minimum width=1cm] (18) at (3, -2) {$3$};
        \node[draw, circle, minimum width=1cm] (19) at (4, -2) {$4$};
        \draw[-latex] (12) -- (13);
        \draw[-latex] (12) -- (14);;
        \draw[-latex] (12) -- (15);
        \draw[-latex] (12) -- (16);
        \draw[-latex] (14) -- (17) node[midway,xshift=-.1cm,above] {\textcolor{blue}{$e$}};
        \draw[-latex] (14) -- (18);
        \draw[-latex] (14) -- (19);
        \draw[-implies,double equal sign distance] (5.5,0) -- (6.5,0) node[midway,above] {$m\triangleleft_e (e,i_2)$};;
        \draw[-implies,double equal sign distance] (5.5,1) -- (6.5,3) node[midway,below] {$m\triangleleft_e (e,i_1)$};;
        \draw[-implies,double equal sign distance] (5.5,-1) -- (6.5,-3) node[midway,above] {$m\triangleleft_e (e,0)$};;
        \node[draw, circle,fill=pink, minimum width=1cm] (20) at (9.2, 5.8) {$S$};
        \node[draw, circle, minimum width=1cm] (21) at (7.2, 2.7) {$1$};
        \node[draw, circle, minimum width=1cm,fill=cyan, dashed,label={\small 1}] (22) at (7.7, 4.2) {$P$};
        \node[draw, circle, minimum width=1cm,fill=cyan,label={\small 2}] (28) at (8.7, 4.2) {$P$};
        \node[draw, circle, minimum width=1cm,label={\small 3}] (23) at (9.7, 4.2) {$5$};
        \node[draw, circle, minimum width=1cm,label={\small 4}] (24) at (10.7, 4.2) {$6$};
        \node[draw, circle,dashed, minimum width=1cm] (25) at (8.2, 2.7) {$2$};
        \node[draw, circle, minimum width=1cm] (26) at (9.2, 2.7) {$3$};
        \node[draw, circle, minimum width=1cm] (27) at (10.2, 2.7) {$4$};
        \draw[-latex] (20) -- (22);
        \draw[-latex] (20) -- (28);
        \draw[-latex] (20) -- (23);
        \draw[-latex] (20) -- (24);
        \draw[-latex] (22) -- (21);
        \draw[-latex] (22) -- (25);
        \draw[-latex] (28) -- (26);
        \draw[-latex] (28) -- (27);
        \node[draw, circle,fill=pink, minimum width=1cm] (29) at (9.2, 1.5) {$S$};
        \node[draw, circle, minimum width=1cm,label={\small 1}] (30) at (7.2, 0) {$1$};
        \node[draw, circle, minimum width=1cm,fill=cyan,label={\small 2}] (31) at (8.2, 0) {$P$};
        \node[draw, circle, minimum width=1cm,label={\small 4}] (32) at (10.2, 0) {$5$};
        \node[draw, circle, minimum width=1cm,label={\small 5}] (33) at (11.2, 0) {$6$};
        \node[draw, circle, dashed,minimum width=1cm,label={[xshift=.2cm]\small 3}] (34) at (9.2, 0) {$2$};
        \node[draw, circle, minimum width=1cm] (35) at (7.7, -1.5) {$3$};
        \node[draw, circle, minimum width=1cm] (36) at (8.7, -1.5) {$4$};
        \draw[-latex] (29) -- (30);
        \draw[-latex] (29) -- (31);
        \draw[-latex] (29) -- (32);
        \draw[-latex] (29) -- (33);
        \draw[-latex] (29) -- (34);
        \draw[-latex] (31) -- (35);
        \draw[-latex] (31) -- (36);
        \node[draw, circle,fill=cyan, minimum width=1cm] (37) at (9.5, -2.5) {$P$};
        \node[draw, circle,fill=pink, minimum width=1cm] (38) at (8.5, -3.5) {$S$};
        \node[draw, circle,dashed,minimum width=1cm] (39) at (10.5, -3.5) {$2$};
        \node[draw, circle, minimum width=1cm,label={\small 1}] (40) at (7.2, -4.7) {$1$};
        \node[draw, circle, minimum width=1cm,fill=cyan,label={\small 2}] (41) at (8.3, -4.7) {$P$};
        \node[draw, circle, minimum width=1cm,label={\small 3}] (42) at (9.4, -4.7) {$5$};
        \node[draw, circle, minimum width=1cm,label={\small 4}] (43) at (10.5, -4.7) {$6$};
        \node[draw, circle, minimum width=1cm] (44) at (7.8, -6) {$3$};
        \node[draw, circle, minimum width=1cm] (45) at (8.8, -6) {$4$};
        \draw[-latex] (37) -- (38);
        \draw[-latex] (37) -- (39);
        \draw[-latex] (38) -- (40);
        \draw[-latex] (38) -- (41);
        \draw[-latex] (38) -- (42);
        \draw[-latex] (38) -- (43);
        \draw[-latex] (41) -- (44);
        \draw[-latex] (41) -- (45);
    \end{tikzpicture}
    
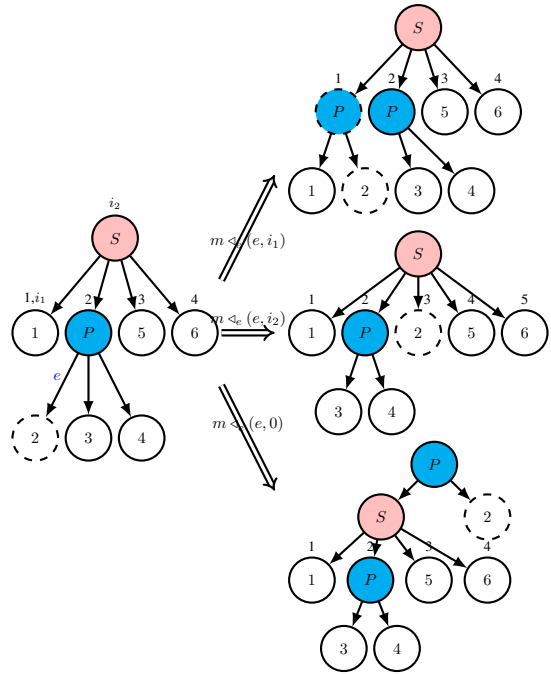
\captionof{figure}{Some possible operations on the MDT $m_1$ from Fig.~4. 
    The edge $e$ connected to leaf for actor 2 is reconnected to leaf node $i_1$ (where it must give a new $P$ node as its neighbor, the parent of $i_1$, is $S$), to ancestral node $i_2$ (and is randomly allocated position 3 among the nodes stacked below the $S$-node $i_2$), and to node $0$ (where it is added above the root as a $P$-node, as its neighbor the ex-root node is $S$). 
    }\label{fig:edge-op-MDT}
\end{minipage}

Some operations are inadmissible, so we need to check our proposal defines an irreducible Markov chain on its own, or add other operations.

\begin{restatable}[Posterior Marginals]{prop}{MDTirreducible}
\label{prop:MDT-irreducible}
Consider the MDT Markov chain $M_k,\ k\ge 0$ with $M_0\in\M_{[n]}$ formed by repeated random updates defined as follows: let $M_t=m$; let $e\sim \U(E_{-0}(m))$ and $i\sim \U[(\F\cup \A)\setminus\{e_1,e_2\}]$; Let $m' = m\triangleleft_e (e,i)$ be given by Defn.~\ref{defn:edge-op-MDT}; if $m'\in\M_{[n]}$ set $M_{k+1}=m'$ and otherwise $M_{k+1}=m$. This proposal-chain is irreducible.
\end{restatable}

\begin{proof}[Proposition~\ref{prop:MDT-irreducible}]
    Consider the two building-block MDT's $m_a,m_b$ shown in the top row of Fig.~\ref{fig:MDT-block}. These have a single internal node with $n$ leaves. Any MDT $m\in\M_{[n]}$ has a root node which must be of type $P$ or $S$. We show that every MDT with a root of type $P$ (or $S$) intercommunicates with $m_a$ (respectively $m_b$) and that $m_a$ intercommunicates with $m_b$ and hence $\M_{[n]}$ is a closed communicating class. 

    \begin{minipage}{\linewidth}
    \centering
    \begin{tikzpicture}[thick,scale=.7, every node/.style={scale=0.6}]
        \node[draw, circle, minimum width=1cm, fill=cyan] (1) at (-2.2,2) {$P$};
        \node[draw, circle, minimum width=1cm] (2) at (-2.9, .5) {};
        \node[draw, circle, minimum width=1cm] (3) at (-3.9, .5) {};
        \node[draw, circle, minimum width=1cm] (4) at (-1.5, .5) {};
        \node[draw, circle, minimum width=1cm] (5) at (-0.5, .5) {};
        \node[] at (-2.2,-.2) {$(a)$};
        \node[] at (-2.2,.5) {$\dots$};
        \draw[-latex] (1) -- (2);
        \draw[-latex] (1) -- (3);
        \draw[-latex] (1) -- (4);
        \draw[-latex] (1) -- (5);
        \node[draw, circle, fill=pink, minimum width=1cm] (1) at (2.6,2) {$S$};
        \node[draw, circle, minimum width=1cm,label={[xshift=.1cm]\small $n-1$}] (2) at (3.3, .5) {};
        \node[draw, circle, minimum width=1cm,label={\small $n$}] (3) at (4.3, .5) {};
        \node[draw, circle, minimum width=1cm,label={\small 2}] (4) at (1.9, .5) {};
        \node[draw, circle, minimum width=1cm,label={\small 1}] (5) at (0.9, .5) {};
        \node[] at (2.6,-.2) {$(b)$};
        \node[] at (2.6,.5) {$\dots$};
        \draw[-latex] (1) -- (2);
        \draw[-latex] (1) -- (3);
        \draw[-latex] (1) -- (4);
        \draw[-latex] (1) -- (5);
        \node[draw, circle, minimum width=1cm,fill=cyan] (1) at (-2.2,-1) {$P$};
        \node[draw, circle, minimum width=1cm,fill=pink] (3) at (-3.2, -2) {$S$};
        \node[draw, circle, minimum width=1cm,fill=pink] (4) at (-1.2, -2) {$S$};
        \node[draw, circle, minimum width=1cm,label={\small 1}] (5) at (-4.1, -3.2) {};
        \node[draw, circle, minimum width=1cm,label={\small 2}] (6) at (-2.7, -3.2) {};
        \node[draw, circle, minimum width=1cm,label={\small 1}] (7) at (-1.7, -3.2) {};
        \node[draw, circle, minimum width=1cm,label={\small 2}] (8) at (-.3, -3.2) {};
        \node[] at (-2.2,-3.9) {$(c)$};
        \node[] at (-2.2,.-2) {$\dots$};
        \node[] at (-3.4,.-3.2) {$\dots$};
        \node[] at (-1,.-3.2) {$\dots$};
        \draw[-latex] (1) -- (3);
        \draw[-latex] (1) -- (4);
        \draw[-latex] (3) -- (5);
        \draw[-latex] (3) -- (6);
        \draw[-latex] (4) -- (7);
        \draw[-latex] (4) -- (8);
        \node[draw, circle, minimum width=1cm,fill=pink] (1) at (2.6,-1) {$S$};
        \node[draw, circle, minimum width=1cm,fill=cyan,label={\small 2}] (3) at (3.6, -2) {$P$};
        \node[draw, circle, minimum width=1cm,fill=cyan,label={\small 1}] (4) at (1.6, -2) {$P$};
        \node[draw, circle, minimum width=1cm] (5) at (4.5, -3.2) {};
        \node[draw, circle, minimum width=1cm] (6) at (3.1, -3.2) {};
        \node[draw, circle, minimum width=1cm] (7) at (2.1, -3.2) {};
        \node[draw, circle, minimum width=1cm] (8) at (.7, -3.2) {};
        \node[] at (2.6,-3.9) {$(d)$};
        \node[] at (2.6,.-2) {$\dots$};
        \node[] at (3.8,.-3.2) {$\dots$};
        \node[] at (1.4,.-3.2) {$\dots$};
        \draw[-latex] (1) -- (3);
        \draw[-latex] (1) -- (4);
        \draw[-latex] (3) -- (5);
        \draw[-latex] (3) -- (6);
        \draw[-latex] (4) -- (7);
        \draw[-latex] (4) -- (8);
    \end{tikzpicture}
    
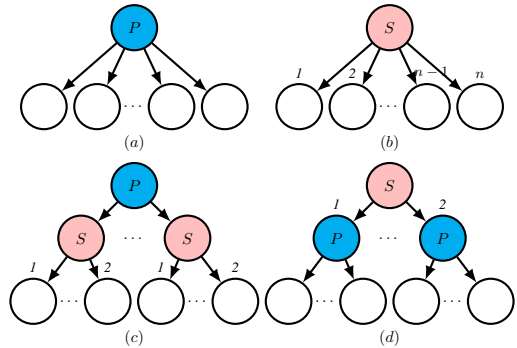
\captionof{figure}{Four building-block MDT's.
    }\label{fig:MDT-block}
\end{minipage}

    
    We first show $m_a\to m_b$. We use the 0 node but there are many paths. Let $r_a$ be the label of the root node in $m_a$ and $r_b$ in $m_b$. Suppose $L_{r_b}(m_b)=(i_1,...,i_n)$ gives the stacking data for the children of the $S$ node $r_b$. Label nodes of $m_a$ so $f_c(r_a|m_a)=\{i_1,...,i_n\}$ and $F_{i_k}(m_a)=F_{i_k}(m_b),\ k=1,...,n$. Now take $e=\langle r_a,i_1\rangle$ in $m_a$ and $i=0$ and set $m=m_a \triangleleft_e (e,i)$. This creates a new node $j$ of type $S$ above the root.
    Let the stacking data of this new node be $L_j(m)=(i_1,r)$. Now apply $m\gets m \triangleleft_e (\langle r,i_k\rangle,j)$ for each $k=2,...,n-1$,
    adding $i_k$ into position $k$ in the list $L_j(m)$. When we do the last node $k=n-1$, node $r$ is removed and $j$ connects directly to $i_n$ with $i_n$ in the correct position in $L_j(m)$. This gives $m=m_b$. All these operations are admissible and have non-zero probability. The same scheme can be reversed, so
    we can take a MDT of type $m_b$ and reorder the entries in $L_{r_b}(m_b)$ by going to $m_a$ and back, placing the leaves in any desired order in $L_j(m)$ as we pass back.

    Now take a general $m^*\in \M_{[n]}$. Its root $r^*$ matches $m_a$ or $m_b$ by type. The root of $m^*$ partitions the leaves into $K$ sets $\{s_1,...,s_K\}$ where $K$ is the number of child nodes of $r^*$ and $s_k=(s_{k,1},...,s_{k,c_k}),\ k=1,...,K$. 
    
    If the root type of $m^*$ is $S$ then these partitions are ordered. In this case we permute the leaves of $m_b$ so that $L_{r_b}(m_b)=(s_{1,1},...,s_{K,c_K})$. Let $m=m_b$ with root $r$. If $i\in s_{k'}$ is a child of $r^*$ which is a leaf then $s_{k'}=\{i\}$ and we are done. All the other partitions $s_k$ correspond to child nodes $i_k$ of $r^*$ which are $P$ nodes. We pull the edges $\langle r,i\rangle,\ i\in s_k$ of $m$ down one at a time to create a $P$ node with child nodes $s_k$ matching the leaf-descendants of $i_k$ in $m^*$. This gives a new $m$ matching $m^*$ down to all nodes of depth less than or equal to two. The passage from $m_b$ to the new $m=m_d$ is illustrated bottom right in Fig.~\ref{fig:MDT-block}.
    
    If the root type of $m^*$ is $P$ then the partitions are $\{s_1,...,s_K\}$ are unordered. The same process is repeated for $m=m_a$, pulling down the edges $\langle r,i\rangle,\ i\in s_k$ one at a time to build an $S$-node with leaves $s_k$ matching the leaf-descendants of $i_k$ and their order in $m^*$.
    
    The process can now be repeated, as the problem of changing an MDT $m$ so that it matches $m^*$ to depth three when it already matches $m^*$ to depth two is the problem of changing the MDT's in $m$ rooted by $i_1,...,i_K$ to match the corresponding subtrees of $m^*$ to depth two. This task is the same as the original task and we have shown we can match to depth two. Since we can always increase the depth of the match and the depth is finite, we can change $m_a$ or $m_b$ to match $m^*$.
    

    
    It is straightforward to check that these processes can be reversed and so the MDT proposal Markov chain formed by repeated edge operation defined in Defn. \ref{defn:edge-op-MDT} is irreducible. 
\end{proof}

Our MCMC algorithm for MDT with the QJ-B observation model is given in Algorithm \ref{alg:MCMC-MDT}, omitting the standard $q,p$ and $\phi$ updates. The algorithm for QJ-U model omits the $\phi$-update. 

\begin{algorithm}[H]
    \caption{The MCMC algorithm for the MDT with QJ-B observation model at step $k$.}\label{alg:MCMC-MDT}
    \begin{algorithmic}
        \Require $y,m^{(k-1)}\!\!=\!m, q^{(k-1)}\!\!=\!q, p^{(k-1)}\!\!=\!p,\phi^{(k-1)}\!\!=\!\phi$ with $m\!=\!(F(m),E(m),L(m)),\ m \!\in\! \M_{[n]}$\\
        \Ensure \\
        {\centering\vspace*{-0.1in}
         $
        \begin{aligned}
            &m^{(k)}\sim \pi(m|y,q,p,\phi), \\
            &q^{(k)}\sim\pi(q|y,m^{(k)},p,\phi),\\
            &p^{(k)}\sim\pi(p|y,m^{(k)},q^{(k)},\phi),\\
            &\phi^{(k)}\sim\pi(\phi|y,m^{(k)},q^{(k)},p^{(k)})\\[0.05in]
        \end{aligned}
        $
        \par}

        \hrulefill\emph{Update for $m$}\hrulefill
        \State $m' \gets m^{(k-1)} \gets m$
        \State Sample $e\sim \U(E_{-0}(m))$ and $i \sim \U[(\F \cup \A)\backslash\{e_1,e_2\}]$
        \State $m'\gets m\triangleleft_e (e,i)$
        \If{$m'\in\M_{[n]}$}
            $$\eta_1 \gets \frac{Q(y|v(m'),p,\phi)\pi_{\M_{[n]}}(m'|q)\rho(m|m')}{Q(y|v(m),p,\phi)\pi_{\M_{[n]}}(m|q) \rho(m'|m)}$$
            \If {$\U(0,1)\le \eta_1$}
            \State $m\gets m^{(k)} \gets m'$
            \EndIf
        \EndIf
        
        \hrulefill\emph{Updates for $q,p$ and $\phi$ omitted}\hrulefill
    \end{algorithmic}
\end{algorithm}

The queue-jumping probability $p>0$ (almost surely) so the Hastings ratio $\eta>0$ in Algorithm~\ref{alg:MCMC-MDT} is not zero for all $m,m'\in \M_{[n]}$ connected by an update. Since the proposal chain $M_k,\ k\ge 0$ in Proposition~\ref{prop:MDT-irreducible} is irreducible, it follows that our MDT-MCMC is irreducible.

\pagebreak 

\section{Data background and Additional Results}\label{sec:data-and-results-app}

\subsection{The `Royal Acta' Data}

The ``Royal Acta'' data is a database made for ``The Charters of William II and Henry I'' project by the late Professor Richard Sharpe and Dr Nicholas Karn \citep{sharpe14}. It collects dated witness lists from legal documents in England and Wales in the eleventh and twelfth century. Each witness list is dated though the dating is sometimes uncertain (a few years is typical). Lower and upper bounds on the date of a list are part of the data. Each individual is associated with a profession (title) such as Queen, Archbishop, etc. We assign witnesses with no title as ``other''. \Fig~\ref{fig:ex-list} gives an example of such witness list. The data records different number of lists with various lengths over time - summarised in Figure~\ref{fig:datasum}. 

\begin{figure}[!htb]
  \centering
  \includegraphics[width=.4\linewidth]{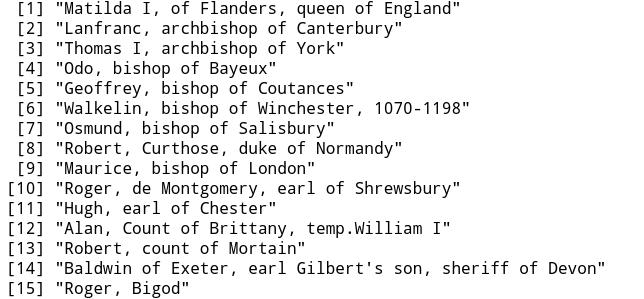}
  \caption{An example witness list from 1080, extracted from the ``Royal Acta'' data. The witnesses names are entered by a clerk in order from top to bottom.}\label{fig:ex-list}
\end{figure}

\begin{figure}[!htb]
  \centering
  \includegraphics[width=0.8\linewidth]{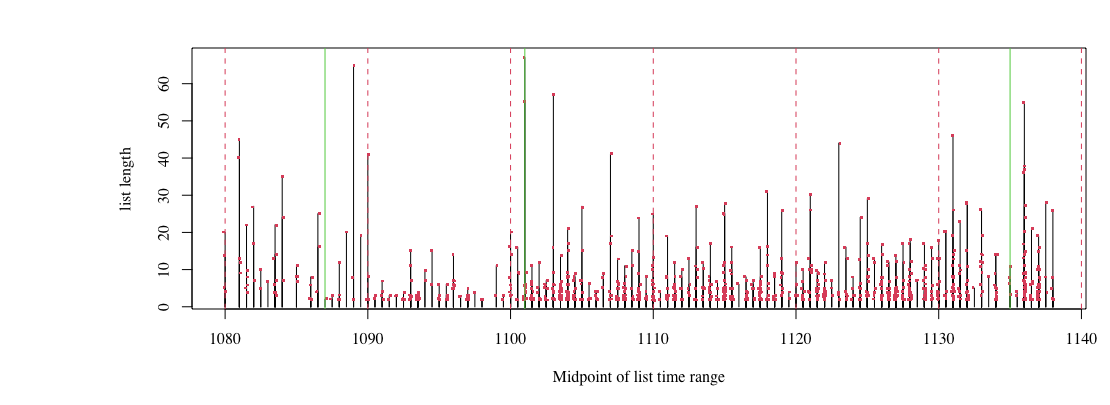}
  \caption{The midpoint of list time range v.s. list range. Each red dot is a list of length $y$ created in a time range midpointed by $x$. The bars represents the length of the longest list at time $x$. }\label{fig:datasum}
\end{figure}

In Section 5, 
we limit the number of lists per actor (LPA) participate in to be at least 5 for ease of presentation. However, it is possible to fit our model on much larger datasets. We chose time periods with a large number of lists with relatively long lengths - 1080-1084 and 1136-1138, and extract the lists with 1LPA. Table \ref{tab:data} summarises the data in the different experiments. In Section~\ref{sec:1lpa}, we carry out Bayesian inference on the 1LPA datasets. In Sections~\ref{sec:5lpa-u} and \ref{sec:5apl-b}, we present MCMC traceplots and effective sample sizes for MCMC samples of key parameters in the analysis on 5LPA data, from the VSP/QJ-U and VSP/QJ-B models respectively. 

\begin{table}[!htb]
    \centering
    \caption{Data content for time periods of interest including the number of actors ($n$), number of lists ($N$) and the length of their longest list ($\max(y)$). Data analysed with both VSP/QJ-U and VSP/QJ-B are marked in blue. The 1134-1138 bishop-only data is 34-38(b). }\label{tab:data}
    \resizebox{.5\columnwidth}{!}{%
    \begin{tabular}{c|cccc|cc}
      \toprule 
      {} & \multicolumn{4}{c}{\bfseries 5LPA} & \multicolumn{2}{c}{\bfseries 1LPA}\\
      \midrule
      \bfseries  & \textcolor{blue}{80-84} & \textcolor{blue}{26-30} & 34-38 & \textcolor{blue}{34-38(b)}  & 80-84 & 34-38\\
      \midrule 
      $n$ & \textcolor{blue}{17} & \textcolor{blue}{13} & 49 & \textcolor{blue}{14} & 181 & 216 \\
      $N$ & \textcolor{blue}{20} & \textcolor{blue}{30} & 82 & \textcolor{blue}{37} & 27 & 95\\
      $\max(y)$ & \textcolor{blue}{17} & \textcolor{blue}{8} & 35 & \textcolor{blue}{14} & 45 & 55\\
      \bottomrule 
    \end{tabular}
    }
\end{table}

\pagebreak

\subsubsection{Inference Results on List Data with 1LPA (QJ-U Observation Model)}\label{sec:1lpa}

Using the full-data lists (allowing $LPA=1$), we arrive at much larger datasets with 181 actors (1080-1084) and 216 actors (1134-1138) respectively, as is summarised in table \ref{tab:data}. Though QJ-B observation model has higher flexibility, it is rather computationally demanding when we move to large datasets. In this section, we fit the VSP/QJ-U model on both data lists instead. 

We perform 50,000 MCMC iterations on 1080-1084 (1LPA) data and 48,000 iterations on 1134-1138 (1LPA) data. For details of the MCMC algorithm, see Algorithm ~\ref{alg:MCMC-BDT}. Every 10 steps is recorded from the MCMC. The effective sample sizes and traceplots for the key parameters $p$ and $P(S)=q$ from the MCMC samples are shown in Table~\ref{tab:ess-1} and Figure~\ref{fig:trace-1}. The MCMC on the 1080-1084 (1LPA) data displays fair mixing, however, the MCMC for 1134-1138 (1LPA) is yet to be fully mixed. We are aware the effective sample sizes are relatively small, here we only present the current results as a demonstration. 

\begin{table}[!htb]
    \centering
    \caption{The effective sample sizes for $P(S)$ and error probability $p$ on four datasets with 1LPA.}\label{tab:ess-1}
    \begin{tabular}{*3c}
      \toprule 
      {} & \multicolumn{2}{c}{\bfseries ESS}\\
      \midrule
      \bfseries Parameter & 1080-1084 & 1134-1138 \\
      \midrule 
      $P(S)$ & 41 & 25\\
      $p$ & 32 & 47 \\
      \bottomrule 
    \end{tabular}
\end{table}

\begin{figure}[H]
\centering
\begin{subfigure}{.5\textwidth}
  \centering
  \includegraphics[width=\linewidth]{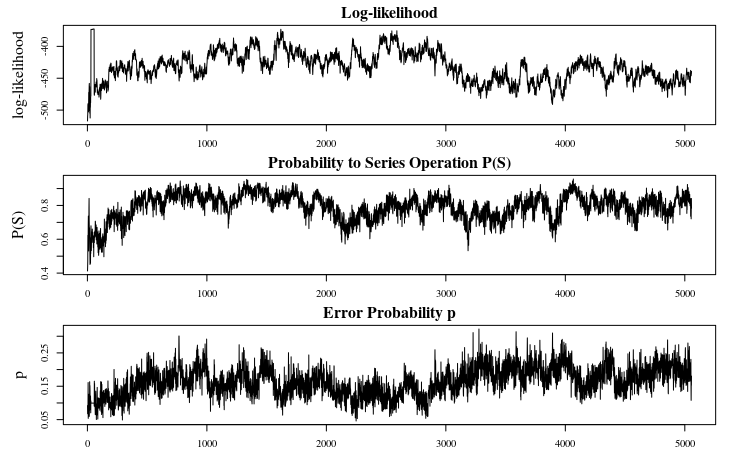}
  \caption{1080-1084 with 1 LPA}
\end{subfigure}%
\begin{subfigure}{.5\textwidth}
  \centering
  \includegraphics[width=\linewidth]{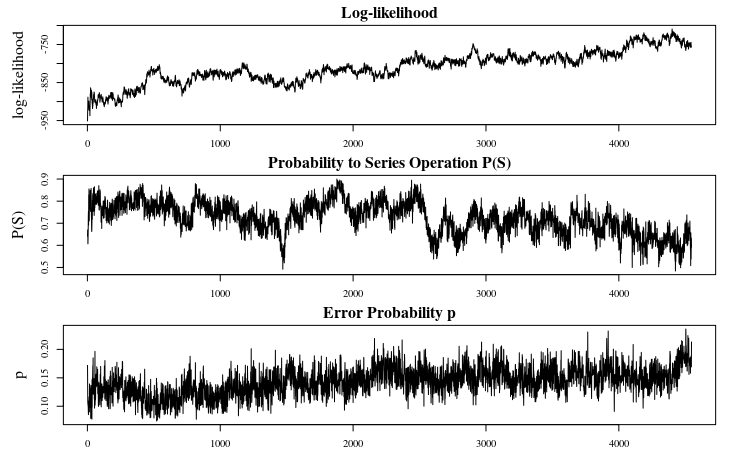}
  \caption{1134-1138 with 1 LPA}
\end{subfigure}
\caption{Traceplots for log-likelihood, $P(S)$ and error probability $p$ for the two data sets of interest here - 1080-1084 (a) and 1134-1138 (b) with 1 LPA data. }
\label{fig:trace-1}
\end{figure}

We present the consensus orders $V^{con}(\epsilon)$ in Figure~\ref{fig:con80841} for 1080-1084 (1LPA) and Figure~\ref{fig:con34381} for 1134-1138 (1LPA). We choose a threshold of $\epsilon=0.6$ in order to represent readable consensus orders graphically. Considering the large number of actors in both time periods, we also extract the non-'other' actors and reconstruct the consensus orders in Figure~\ref{fig:con80841_no} for 1080-1084 (1LPA) and Figure~\ref{fig:con34381_no} for 1134-1138 (1LPA). 

A clear order relation for king $\succ$ queen $\succ$ archbishop $\succ$ bishop is observed in both time periods. The actors roughly appear in the ``group'' of their professions.

\begin{figure}[h]
\centering
\includegraphics[width=.95\linewidth]{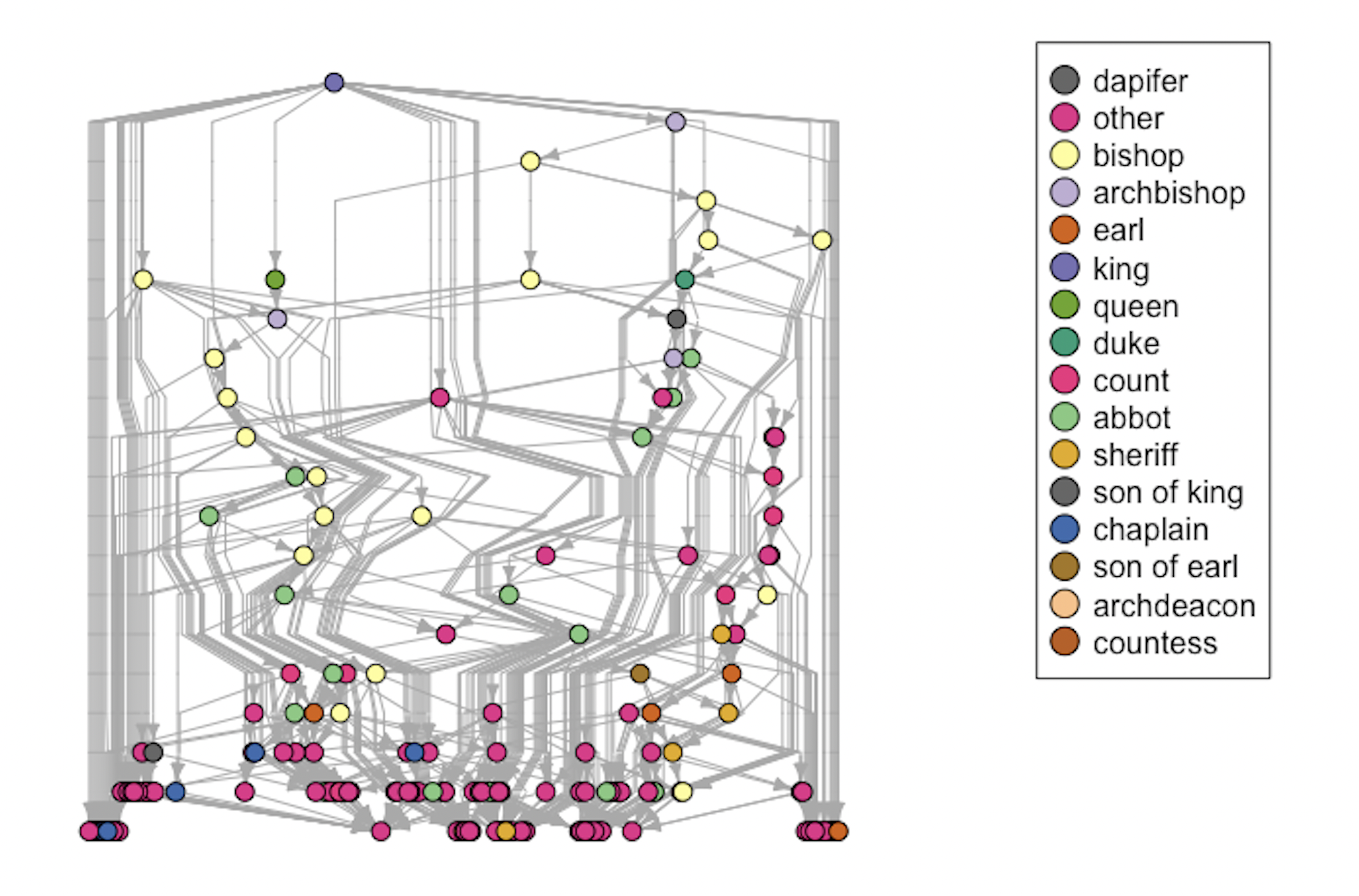}
\caption{The consensus order for 1080-1084 (1LPA) data in a VSP/QJ-U analysis. }
\label{fig:con80841}
\end{figure}

\begin{figure}[h]
  \centering
  \includegraphics[width=.95\linewidth]{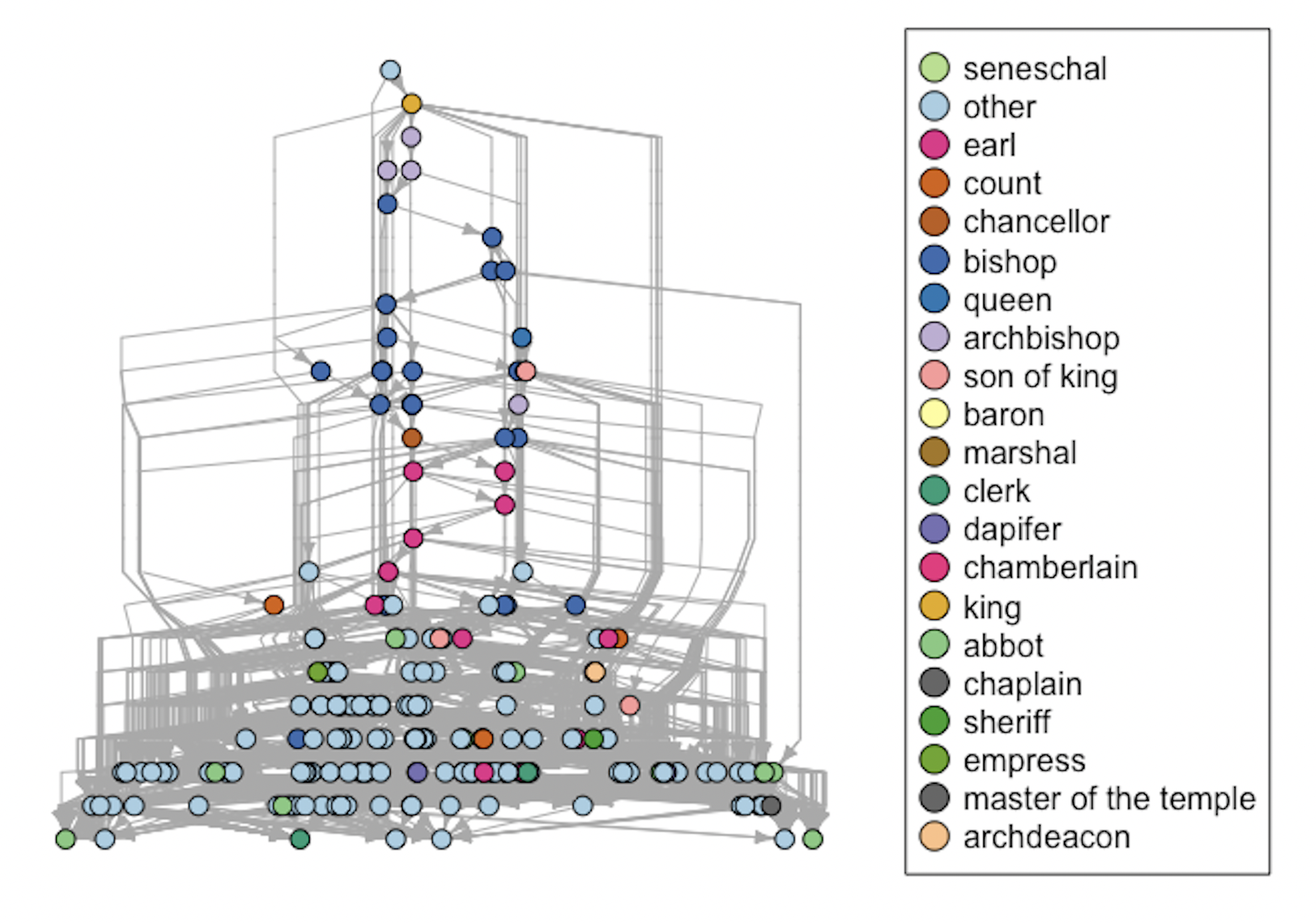}
  \caption{The consensus order for 1134-1138 (1LPA) data in a VSP/QJ-U analysis. }\label{fig:con34381}
\end{figure}

\begin{figure}[H]
  \centering
  \includegraphics[width=.95\linewidth]{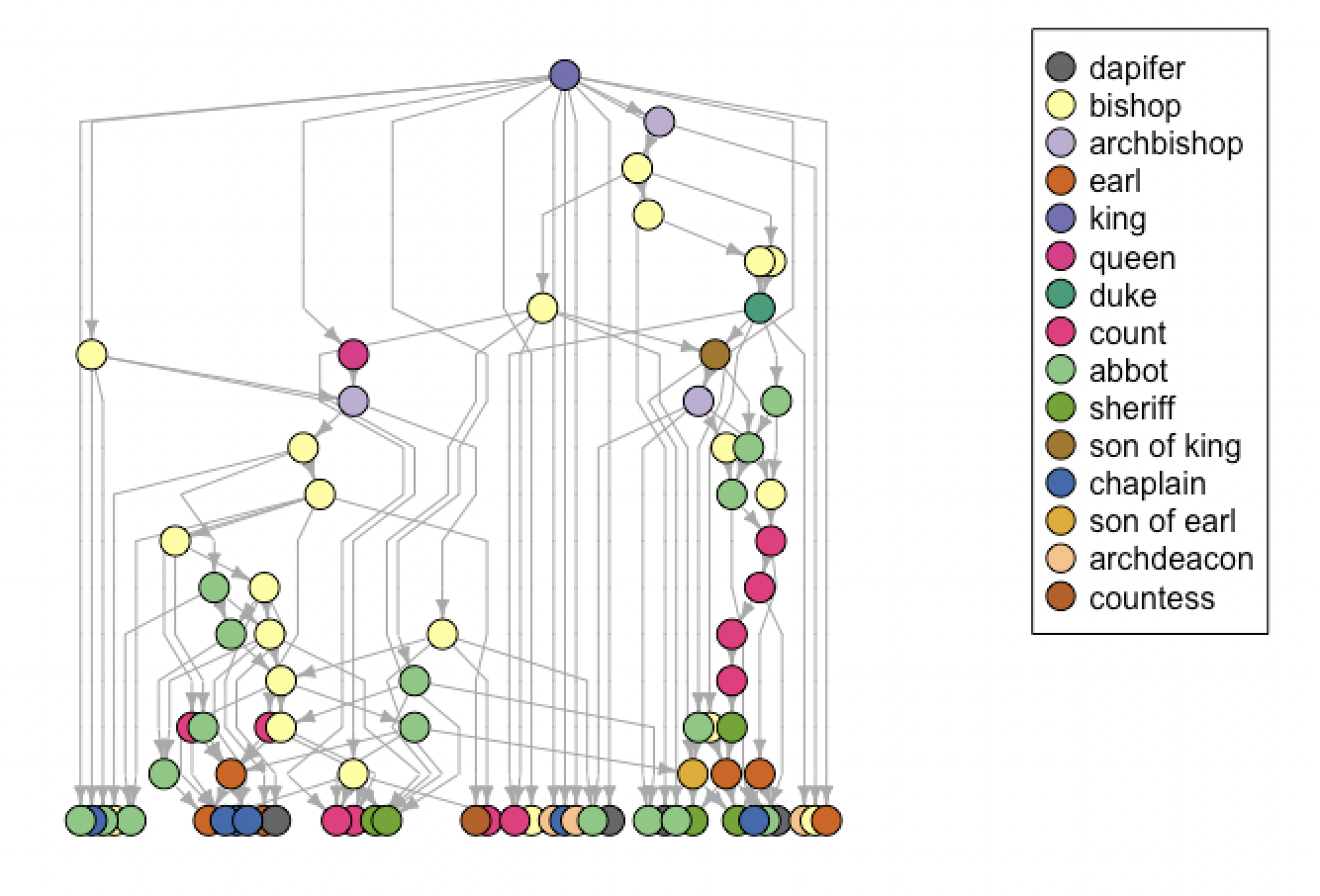}
\caption{The consensus order for 1080-1084 (1LPA) data without `other' actors in a VSP/QJ-U analysis.}
\label{fig:con80841_no}
\end{figure}

\begin{figure}[H]
  \centering
  \includegraphics[width=.95\linewidth]{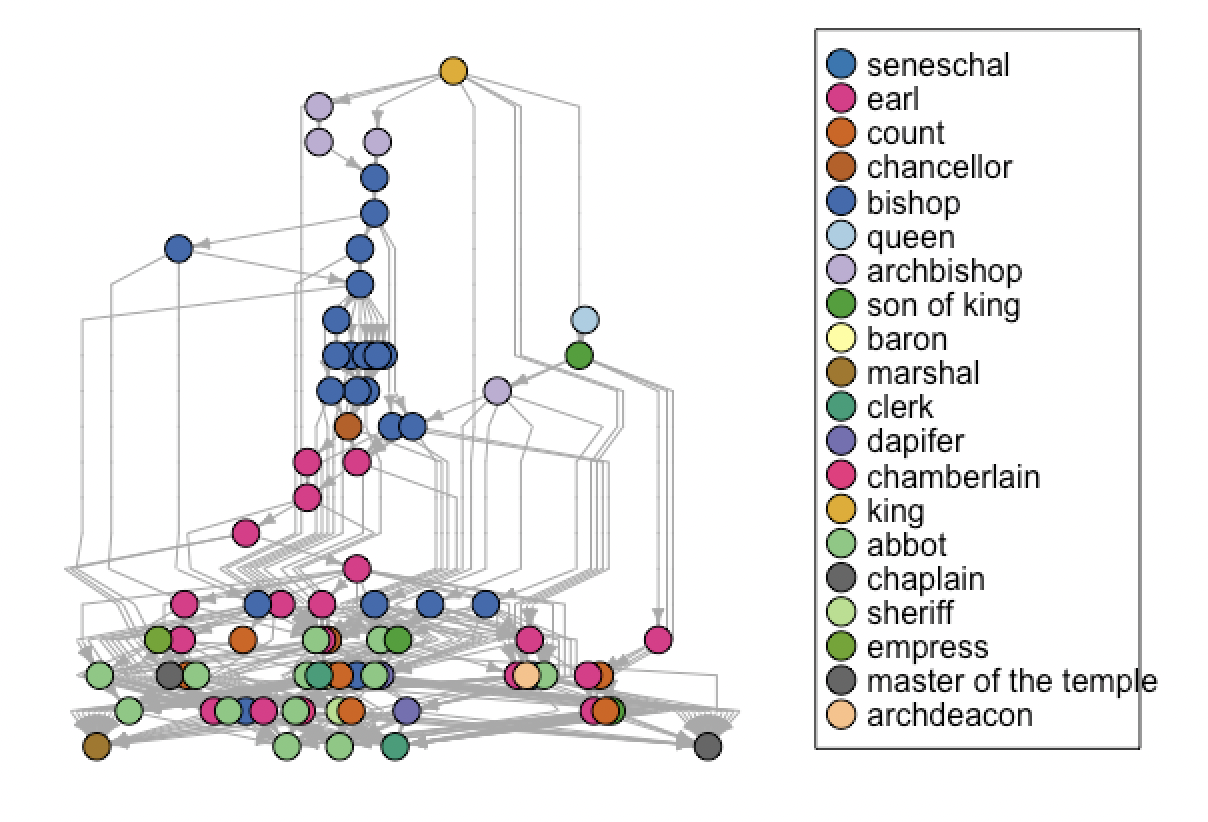}
\caption{The consensus order for 1134-1138 (1LPA) data without `other' actors in a VSP/QJ-U analysis.}
\label{fig:con34381_no}
\end{figure}

Table \ref{tab:exp-rank-1} presents the average rankings of different professions for 1080-1084 (1LPA) and 1134-1138 (1LPA). The average rankings support our observations above. Interestingly, abbots tend to be ranked higher during 1080-1084 than 1134-1138, and the archdeacon is ranked higher in 1134-1138 than 1080-1084.

\begin{table}[!htb]
    \centering
    \small
    \caption{The professions and their average rankings 
    for 1080-1084 (1LPA) and 1134-1138 (1LPA). NA means the profession of interest does not appear in this time period.}\label{tab:exp-rank-1}
    \begin{tabular}{*3c}
      \toprule 
      {} & \multicolumn{2}{c}{\bfseries Average Rank}\\
      \midrule
      \bfseries Profession & 1080-1084 & 1134-1138\\
      \midrule 
      King & 1.21 (0.007) & 3.73 (0.02) \\
      Queen & 4.81 (0.03) & 4.97 (0.02)\\
      Archbishop & 9.70 (0.05) & 8.89 (0.04) \\
      Empress & NA & 16.0 (0.07) \\
      Duke & 15.4 (0.08) & NA  \\
      Bishop & 18.7 (0.10) & 20.8 (0.10) \\
      Son of King & 18.8 (0.10) & 24.0 (0.11)  \\
      Seneschal & NA & 28.0 (0.13)\\
      Abbot & 32.8 (0.18) & 88.0 (0.41) \\
      Countess & 39.0 (0.22) & NA \\
      Count & 43.1 (0.24) & 33 (0.15) \\
      Son of Earl & 43.5 (0.24) & NA \\
      Earl & 44.3 (0.24) & 44.3 (0.20)  \\
      Dapifer & 44.5 (0.25) & 81.3 (0.38) \\
      Archdeacon & 48.7 (0.27) & 35.3 (0.16) \\
      Chancellor & NA & 43.6 (0.20)  \\
      Other & 50.1 (0.28) & 79.2 (0.37) \\
      Chaplain & 50.3 (0.28) & 44.7 (0.21) \\
      Baron & NA & 78.4 (0.36)\\
      Sheriff & 60.5 (0.33) & 95.7 (0.44) \\
      Chamberlain & NA & 101 (0.47) \\
      Clerk & NA & 114 (0.53)\\
      Master of the temple & NA & 137 (0.63) \\
      Marshal & NA & 150 (0.70)\\
      \bottomrule 
    \end{tabular}
\end{table}

Posterior distributions for the key parameters in Figure \ref{fig:pos-1lpa} show that witness lists in 1080-1084 tend to respect a stronger social hierarchy than in 1134-1138 with larger $P(S)$. The error probabilities $p$ are relatively smaller for witness lists in 1134-1138. This agrees with the results for 5LPA presented in \Fig~5, 
Section~5.2. 
The prior and posterior VSP depth distributions are shown in \Fig~\ref{fig:depth-1lpa}. Despite the roughly uniform prior distribution over the VSP depth, the posterior depths appear to concentrate around 75 for 1080-1084 and 90 for 1134-1138.

\begin{figure}[!htb]
  \centering
  \includegraphics[width=0.66\linewidth]{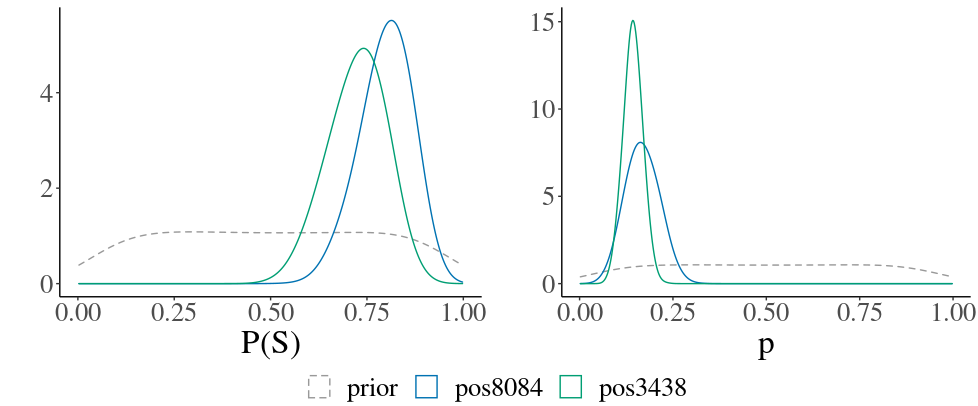}
  \caption{Prior (grey line) and posterior distributions for
$q=P(S)$ (left) and error probability $p$ (right) for the time periods 1080-1084 (1LPA) (blue) and 1134-1138 (1LPA) (green) in a VSP/QJ-U analysis.}\label{fig:pos-1lpa}
\end{figure}

\begin{figure}[!htb]
  \centering
  \includegraphics[width=0.66\linewidth]{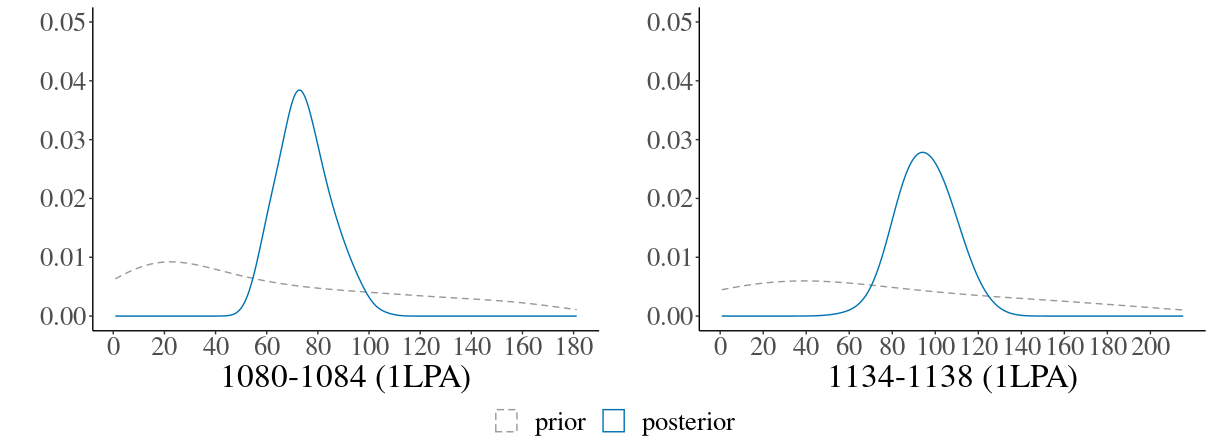}
  \caption{The prior (grey) and posterior (blue) VSP depth distribution for 1080-1084 (1LPA) (left) and 1134-1138 (1LPA) (right) in a VSP/QJ-U analysis.}\label{fig:depth-1lpa}
\end{figure}

\pagebreak 

\subsubsection{Inference Results on List Data with 5LPA (QJ-U Observation Model)}\label{sec:5lpa-u}

\Fig~6 
and \Fig~7 
(top-row) show the consensus orders $V^{con}$ for 1134-1138 (5LPA), 1080-1084 (5LPA), 1126-1130 (5LPA) and 1134-1138 (bishop) (5LPA) under the VSP/QJ-U model. The MCMC converge well. Here we estimate and report effective sample sizes (ESS, Table~\ref{tab:ess}) and inspect MCMC traces (\Fig~\ref{fig:test}). Both the high ESSs and the traceplots indicate good convergence to the posterior distribution. 

\begin{table}[!htb]
    \centering
    \caption{The effective sample sizes for $P(S)$ and error probability $p$ on the four datasets with 5LPA and QJ-U.}\label{tab:ess}
    \begin{tabular}{*5c}
      \toprule 
      {} & \multicolumn{4}{c}{\bfseries ESS}\\
      \midrule
      \bfseries Parameter & 1080-1084 & 1126-1130 & 1134-1138 & 1134-1138(b) \\
      \midrule 
      $P(S)$ & 1676 & 1477 & 95 & 648\\
      $p$ & 1297 & 1426 & 262 & 586\\
      \bottomrule 
    \end{tabular}
\end{table}

\begin{figure}[!htb]
\centering
\begin{subfigure}{.35\textwidth}
  \centering
  \includegraphics[width=\linewidth]{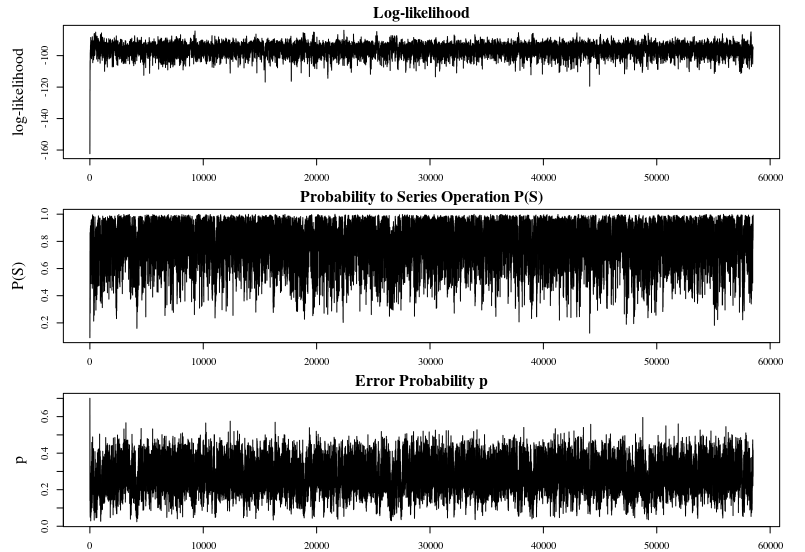}
  \caption{1080-1084 with 5 LPA}
\end{subfigure}%
\begin{subfigure}{.35\textwidth}
  \centering
  \includegraphics[width=\linewidth]{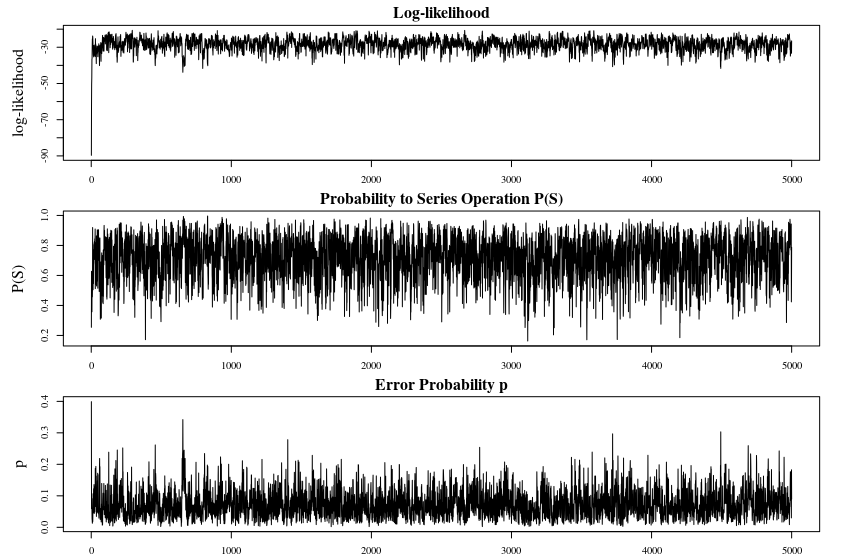}
  \caption{1126-1130 with 5 LPA}
\end{subfigure}
\begin{subfigure}{.35\textwidth}
  \centering
  \includegraphics[width=\linewidth]{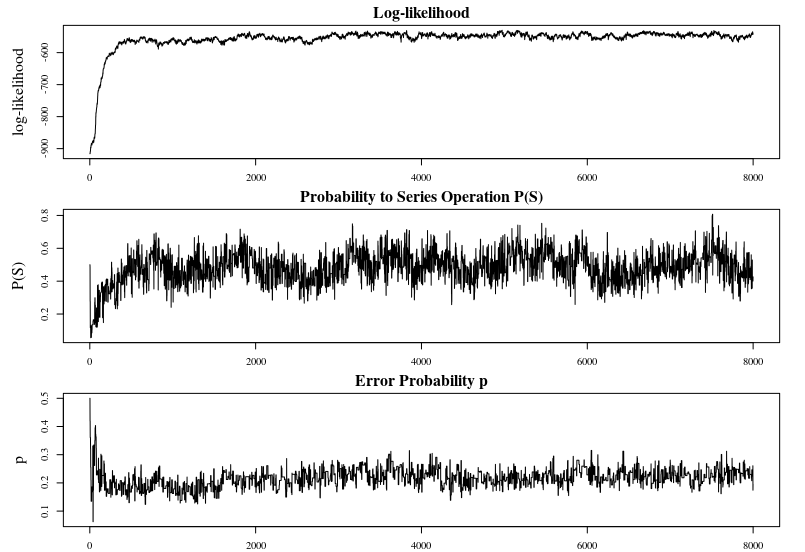}
  \caption{1134-1138 with 5 LPA}
\end{subfigure}
\begin{subfigure}{.35\textwidth}
  \centering
  \includegraphics[width=\linewidth]{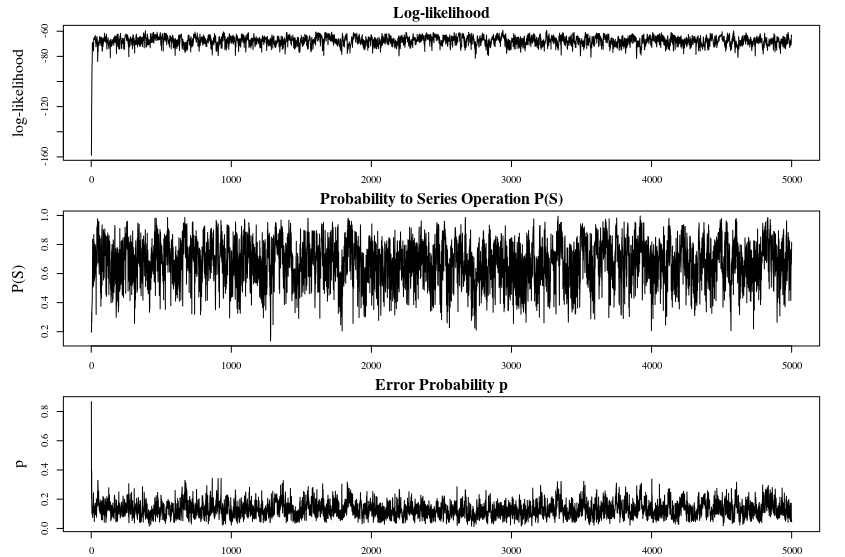}
  \caption{1134-1138(b) with 5 LPA}
\end{subfigure}
\caption{Traceplots for log-likelihood, $P(S)$ and error probability $p$ for the four list data of interest - 1080-1084 (a) and 1126-1130 (b), 1134-1138 (c) and 1134-1138 (bishops) (d) with 5 LPA data and a VSP/QJ-U analysis. }
\label{fig:test}
\end{figure}

The posterior distributions for both $p$ and $q=P(S)$ are shown in \Fig~8. 
We also present the posterior depth-distributions for the datasets in Figure \ref{fig:depth2}. It appears that 1080-1084 (5LPA) admits the most rigid social hierarchy, while 1134-1138 (5LPA) has less hierarchy with respect to $n$. The average rankings per profession are reported in Table \ref{tab:exp-rank-34385}. Similar to the consensus orders (\Fig~6 
and \Fig~7), 
king $\succ$ queen $\succ$ archbishop $\succ$ bishop. The three time periods show similar hierarchical structure, although the power gap between count and earl is relatively narrower in 1126-1130.

\begin{figure}[!tbp]
  \centering
  \includegraphics[width=0.5\linewidth]{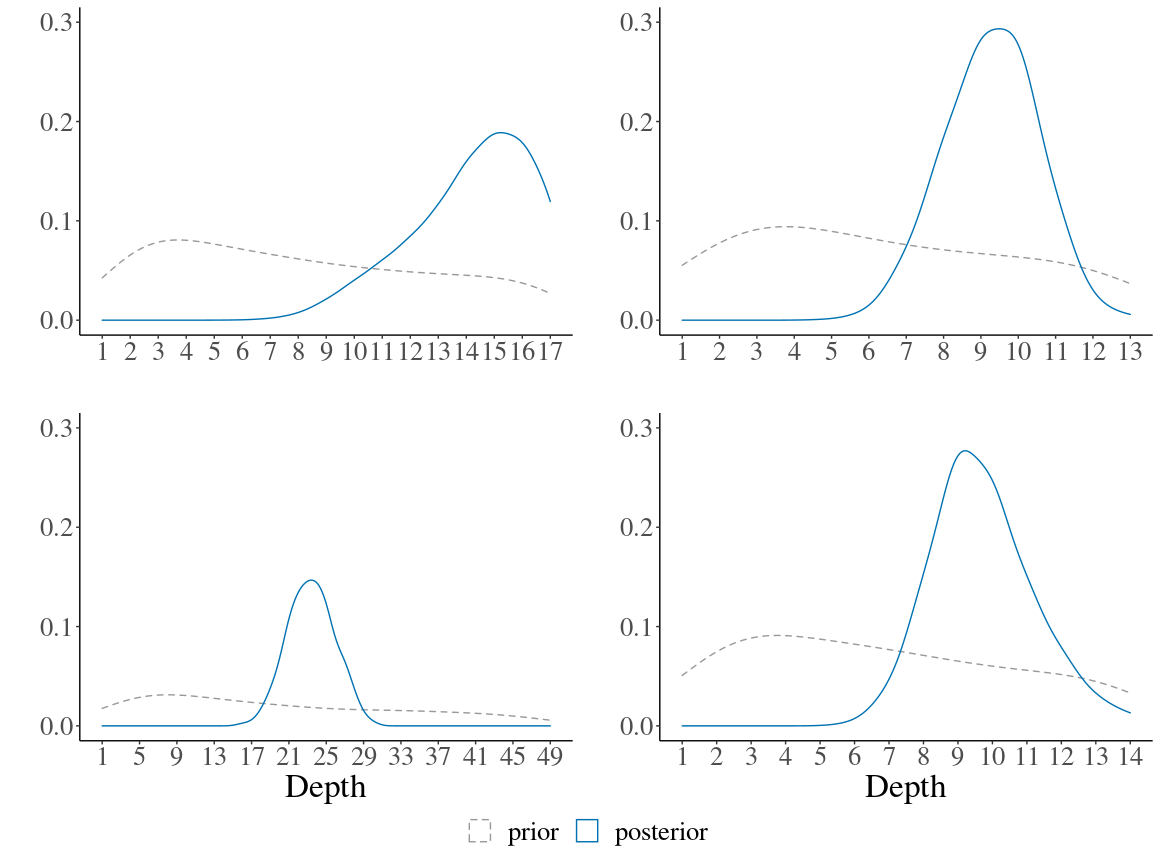}
  \caption{The prior (grey) and posterior (blue) VSP depth distribution for 1180-1184 (top-left), 1126-1130 (top-right), 1134-1138 (bottom-left) and 1134-1138(b) (bottom-right) with 5LPA and QJ-U. }\label{fig:depth2}
\end{figure}

\begin{table}[!htb]
    \centering
    \small
    \caption{The professions and their average rankings 
    for all three time periods with 5LPA and QJ-U. NA means the profession of interest does not appear in this time period.}\label{tab:exp-rank-34385}
    \begin{tabular}{*4c}
      \toprule 
      {} & \multicolumn{3}{c}{\bfseries Average Rank}\\
      \midrule
      \bfseries Profession & 1080-1084 & 1126-1130 & 1134-1138\\
      \midrule 
      King & 1.02 (0.06) & NA & 1.01 (0.02)\\
      Queen & 2.15 (0.13) & NA & 2.01 (0.04)\\
      Duke & 2.79 (0.16) & NA & NA \\
      Son of King & 4.63 (0.27) & NA  & 3.11 (0.06)\\
      Archbishop & 4.45 (0.26) & 1 (0.08) & 4.55 (0.09)\\
      Bishop & 8.25 (0.49) & 4.02 (0.31) & 11.10 (0.23)\\
      Chancellor & NA & NA & 21.40 (0.44) \\
      Count & 10.90 (0.64) & 5.92 (0.45) & 24.00 (0.49)\\
      Earl & 12.20 (0.72) & 5.98 (0.46) & 28.10 (0.57)\\
      Other & 15.30 (0.90) & 8.80 (0.68) & 33.10 (0.68)\\
      \bottomrule 
    \end{tabular} 
\end{table}

\pagebreak

As discussed, we perform reconstruction accuracy tests on each dataset to assess the reliability of our estimations. This is done by taking representative parameters (the last sample state of the parameters sampled from the corresponding posterior), and generating synthetic data with the same list-memberships and lengths as the real data. We carry out or standard analysis on these synthetic datasets, fitting the same model used to simulate the data, and construct the corresponding consensus orders $V^{con} (\epsilon)$ with $\epsilon \in [0,1]$. The results are summarised using receiver operator characteristic (ROC) curves. The ROC curve shows the relation between the proportion of inferred false-positive order relations (x-axis) and true-positive relations (y-axis) for different $\epsilon$. The existence of a $\epsilon$ that gives high true-positive and low false-positive reconstructed fraction means reconstruction accuracy is high. 

\Fig~\ref{fig:roc2} shows ROC curves for such a reconstruction test on the 1080-1084 (5LPA), 1126-1130 (5LPA) and 1134-1138 (5LPA) data in a VSP/QJ-U model. The proportion of inferred false-positive (x-axis) and true-positive (y-axis) relations increases with decreasing $\epsilon$ from (0, 0) at $\epsilon = 1$ (the consensus order is empty) to (1, 1) at $\epsilon = 0$ (complete graph). For all time periods, we observe $\epsilon$ that gives high true-positive and low false-positive reconstructed fraction, indicating our model's high reliability to reconstruct relations.

\begin{figure}[!htb]
  \centering
  \includegraphics[width=.4\linewidth]{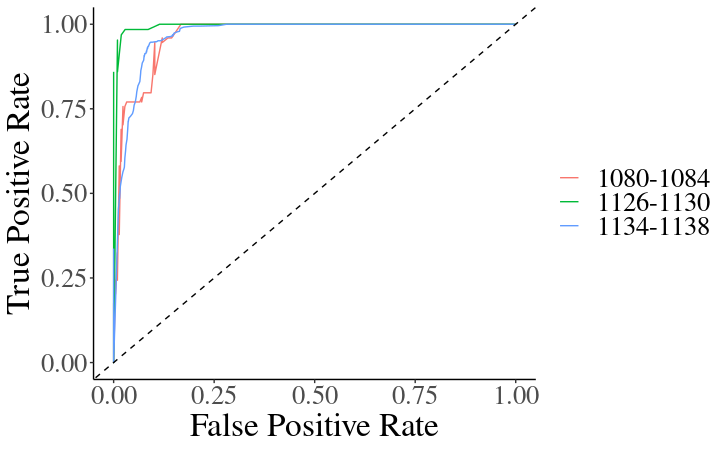}
  \caption{Receiver operating characteristic (ROC) curves for synthetic data using 1080-1084, 26-30 and 34-38 list membership structures with 5LPA and QJ-U.}\label{fig:roc2}
\end{figure}

\subsubsection{Inference Results on List Data with 5LPA (QJ-B Observation Model)}\label{sec:5apl-b}

In this section, we fit the VSP/QJ-B data on the datasets 1080-1084 (5LPA), 1126-1130 (5LPA) and 1134-1138 (bishop) (5LPA). See algorithm \ref{alg:MCMC-BDT} for the MCMC details. Traceplots for the log-likelihood, $P(S)$, error probability $p$ and bi-directional top/bottom insertion probability $\phi$ are all presented in Figure~\ref{fig:trace}. They all display reasonable convergence. In table~\ref{tab:ess-b} we estimate effective sample sizes (ESS) for key parameters. Mixing for the key parameters are fair during time period 1080-1084 and 1134-1138 (bishop), and the agreement (to some extent) to the analyses in Section~\ref{sec:5lpa-u} supports our conclusion that the samples are representative. 

\begin{table}[!htp]
    \centering
    \caption{The effective sample sizes for $P(S)$ and error probability $p$ on the three datasets with 5LPA fitting VSP/QJ-B.}\label{tab:ess-b}
    \begin{tabular}{*4c}
      \toprule 
      {} & \multicolumn{3}{c}{\bfseries ESS}\\
      \midrule
      \bfseries Parameter & 1080-1084 & 1126-1130 & 1134-1138(b) \\
      \midrule 
      $P(S)$ & 47 & 1875 & 121\\
      $p$ & 61 & 3401 & 197 \\
      $\phi$ & 69 & 3428 & 728\\
      \bottomrule 
    \end{tabular}
\end{table}

\begin{figure}[!htb]
\centering
\begin{subfigure}{.33\textwidth}
  \centering
  \includegraphics[width=\linewidth]{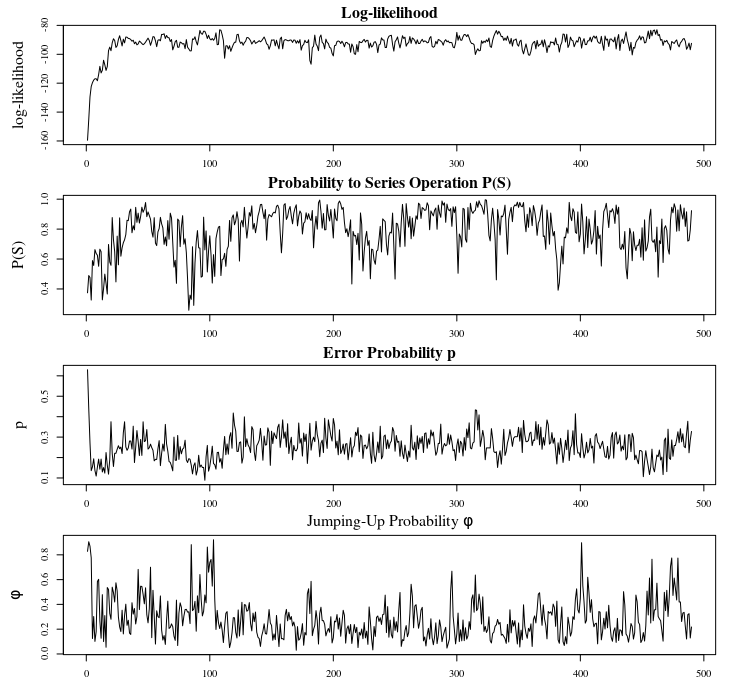}
  \caption{1080-1084 with 5 LPA and QJ-B}
\end{subfigure}%
\begin{subfigure}{.33\textwidth}
  \centering
  \includegraphics[width=\linewidth]{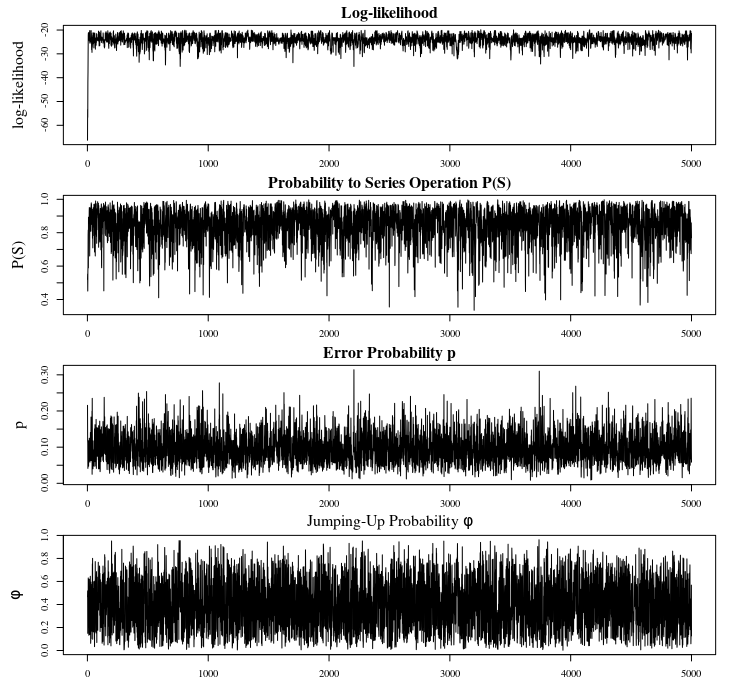}
  \caption{1126-1130 with 5 LPA}
\end{subfigure}
\begin{subfigure}{.33\textwidth}
  \centering
  \includegraphics[width=\linewidth]{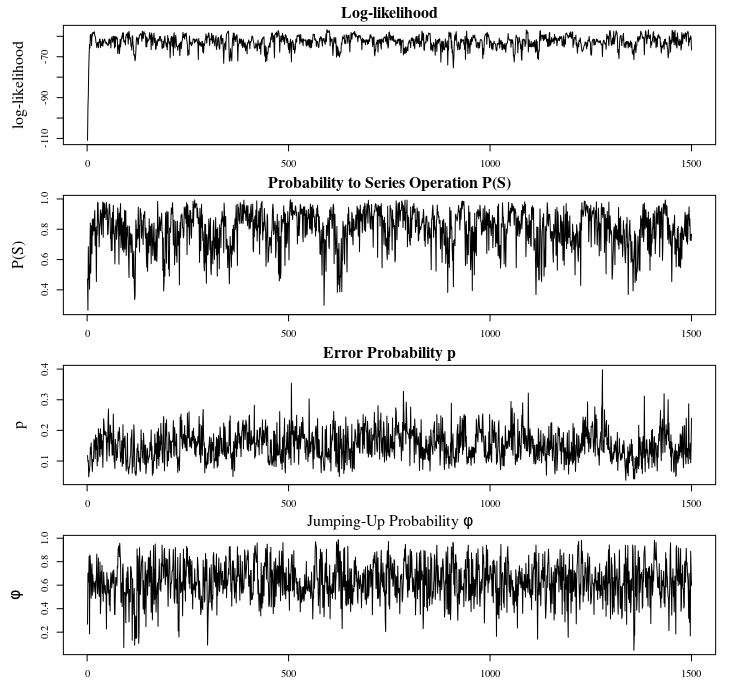}
  \caption{1134-1138(b) with 5 LPA}
\end{subfigure}
\caption{Traceplots for the log-likelihood, $P(S)$ and error probability $p$ for the three list data sets of interest - 1080-1084 (a) and 1126-1130 (b) and 1134-1138 (bishops) (c) with 5LPA data and a VSP/QJ-B analysis.}
\label{fig:trace}
\end{figure}

Consensus orders $V^{con}(\epsilon)$ with $\epsilon=0.5$ are shown in \Fig~7 
(bottom-row). We report the average rankings per profession for 1080-1084 (5LPA) and 1126-1130 (5LPA) in Table~\ref{tab:exp-rank-b}. The posterior distributions for the key parameters $p$, $q=P(S)$ and $\phi$ are shown in \Fig~8. 
Here we display the posterior depth distribution for the three time periods in \Fig~\ref{fig:depth-b}. All periods favour higher VSP depths. By comparing the consensus orders, the bi-directional queue-jumping model seems to fit a more rigid social hierarchy than the queue-jumping-up model, especially during periods 1126-1130 and 1134-1138. This is also illustrated by higher posterior means on $q=P(S)$ for both the 1126-1130 (5LPA) and 1134-1138 (bishop) (5LPA) data. It is surprising that earl $\succ$ count in 1126-1130 under the QJ-B model, although the opposite is observed under QJ-U. Both QJ-U and QJ-B models conclude similar posterior distribution on $p$, the error probability in the data-lists. By inspecting the posterior distributions on $\phi$, it appears that QJ-D is slightly preferred for 1080-1084 (5LPA) while QJ-U/QJ-B is preferred for 1134-1138 (bishop) (5LPA). This is justified by the Bayes Factors in section 5. 

\begin{figure}[!htb]
  \centering
  \includegraphics[width=.5\linewidth]{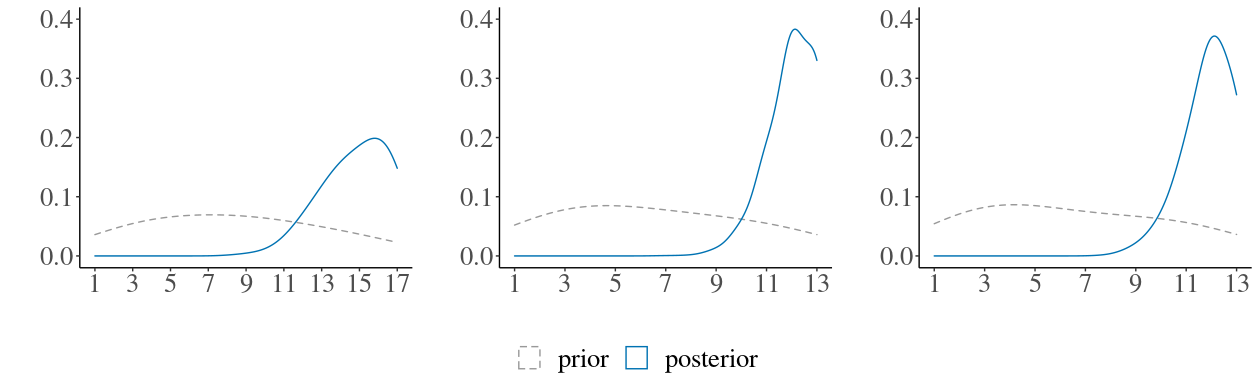}
  \caption{The prior (grey) and posterior (blue) VSP depth distribution for 1180-1184 (left), 1126-1130 (middle) and 1134-1138(b) (right) with 5LPA data in a VSP/QJ-B analysis. }\label{fig:depth-b}
\end{figure}

\begin{table}[!htb]
    \centering
    \small
    \caption{The professions and their average rankings 
    for all three time periods with 5LPA data and QJ-B. NA means the profession of interest does not appear in this time period.}\label{tab:exp-rank-b}
    \begin{tabular}{*3c}
      \toprule 
      {} & \multicolumn{2}{c}{\bfseries Average Rank}\\
      \midrule
      \bfseries Profession & 1080-1084 & 1126-1130\\
      \midrule 
      King & 1.03 (0.06) & NA \\
      Queen & 1.95 (0.11) & NA\\
      Duke & 4.29 (0.25) & NA  \\
      Son of King & 6.18 (0.36) & NA \\
      Archbishop & 3.88 (0.23) & 1 (0.08) \\
      Bishop & 8.38 (0.49) & 3.99 (0.31) \\
      Earl & 12.40 (0.73) & 6.93 (0.53) \\
      Count & 13.00(0.77) & 8.94 (0.69) \\
      Other & 15.90 (0.94) & 10.40 (0.80)\\
      \bottomrule 
    \end{tabular}
\end{table}

Figure \ref{fig:roc-b} displays ROC curves from a reconstruction accuracy test using VSP/QJ-B to simulate and fit synthetic data matching the 1126-1130 and 1134-1138 5LPA data, as described in Section~5. 
Again, we see the proportion of inferred false-positive and true-positive relations increasing while decreasing $\epsilon$ from $(0,0)$ at $\epsilon=1$ to $(1,1)$ at $\epsilon=0$. The $\epsilon$'s that give high true-positive and low false-positive reconstruction fraction can be easily identified in \Fig~\ref{fig:roc-b}. This indicates our model's high accuracy in reconstruction order relations. 

\begin{figure}[!htb]
  \centering
  \includegraphics[width=.4\linewidth,height=115pt]{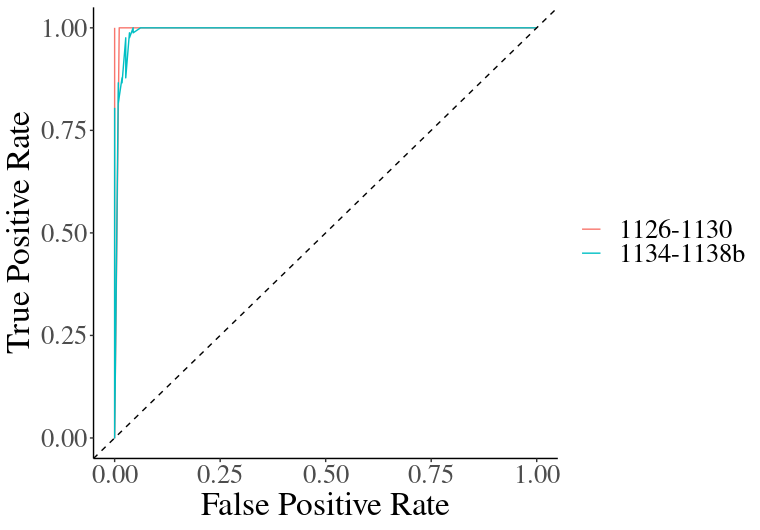}
  \caption{Receiver operating characteristic (ROC) curves for synthetic data using 1126-1130 and 1134-1138 (bishop) list membership structures with 5LPA and QJ-B.}\label{fig:roc-b}
\end{figure}

\pagebreak

\subsection{The Formula 1 Race Data}

The Formula 1 race data (2017 - 2022) \cite{F1Data} records information about every formula 1 race in the past five seasons. The data gives the top 20 drivers in each Grand Prix race in each season. One typical list, for the British Grand Prix (Silverstone Circuit) in 2021, is as follows

\vspace*{-0.6in}

\begin{quote}
    \small 1 – HAM, 2 – LEC, 3 – BOT, 4 – NOR, 5 – RIC, 6 – SAI, 7 – ALO, 8 – STR, 9 – OCO, 10 – TSU, 11 – GAS, 12 – RUS, 13 – GIO, 14 – LAT, 15 – RAI, 16 – PER, 17 – MAZ, 18 – MSC, R – VET, R – VER. 
\end{quote}

\vspace*{-0.6in}

Each abbreviation is a unique code for a driver (see table \ref{tab:drivers}), e.g. `HAM' stands for Lewis Hamilton, who was the winner of this race. The drivers are ordered based on their finishing position. The label `R' indicates special circumstances, e.g. collision, accident, retirement, etc. 

We are interested in the order relations between these drivers and construct a VSP map of their performance in a specific season. This is an intersting test of the method as a heuristic model (in the sense that Plackett-Luce and Mallows are in general heuristic). There is no constraint other than car speed and skill to stop one driver overcoming another so it is not clear that the order relations we recover correspond to any element of reality. One feature that is characteristic of a PO-style analysis (such as ours with VSPs) is that the race resembles a queue in which drivers exchange places subject to skill and car-speed. 
In a race, a driver can fall down the order with a certain probability due to unexpected circumstances (poor tyre management, problems in the pits, small collisions, time penalties etc). However, there is no obvious mechanism promoting a driver up the race order. We therefore believe the QJ-D observation model is natural.

In this analysis, we take a snapshot of 2021, assuming relative car-quality and skill are roughly constant over a year. The Formula 1 (F1) 2021 data consists of 22 lists corresponding to the 22 Grand Prix races. Each list is has at most 20 elements. We disregard the `R' positions, so the lists are of unequal length. There are a total of 21 drivers participating in season 2021. We assign each of them a unique Driver ID, listed in table \ref{tab:drivers}. 

\begin{table}[!htb]
    \centering
    \small
    \caption{The list of drivers in Formula 1 season 2021. Each driver is assigned a unique `Code' and `Driver ID' in our analysis. We also include further information of the drivers, including their date of birth (`DOB') and `Nationality'. }\label{tab:drivers}
    \begin{tabular}{*5c}
      \toprule 
      \bfseries Driver ID & \bfseries Code & \bfseries Name & \bfseries DOB & \bfseries Nationality\\
      \midrule 
      1 & HAM & Lewis Hamilton & 07/01/85 & British \\
      2 & ALO & Fernando Alonso & 29/07/81 & Spanish\\
      3 & RAI	& Kimi Raikonnen & 17/10/79	& Finnish \\
      4 & KUB	& Robert Kubica & 07/12/84 & Polish \\
      5 & VET	& Sebastian	Vettel &	03/07/87 &	German \\
      6 & GAS	& Pierre	Gasly	& 07/02/96 &	French \\
      7 & PER &	Sergio	Perez	& 26/01/90 &	Mexican \\
      8 & RIC &	Daniel	Ricciardo &	01/07/89 &	Australian \\
      9 & BOT	& Valtteri	Bottas &	28/08/89 &	Finnish\\
      10 & VER &	Max	Verstappen &	30/09/97 &	Dutch\\
      11 & SAI & Carlos	Sainz &	01/09/94 &	Spanish\\
      12 & OCO &	Esteban	Ocon &	17/9/96 &	French\\
      13 & STR &	Lance	Stroll &	29/10/98 &	Canadian\\
      14 & GIO &	Antonio	Giovinazzi &	14/12/93 &	Italian\\
      15 & LEC &	Charles	Leclerc	& 16/10/97 &	Monegasque\\
      16 & NOR &	Lando	Norris	& 13/11/99 &	British\\
      17 & RUS &	George	Russell &	15/02/98 &	British\\
      18 & LAT &	Nicholas	Latifi & 29/06/95 &	Canadian\\
      19 & TSU &	Yuki	Tsunoda &	11/05/00 &	Japanese\\
      20 & MAZ &	Nikita	Mazepin &	02/03/99 &	Russian\\
      21 & MSC &	Mick	Schumacher &	22/03/99 &	German\\
      \bottomrule 
    \end{tabular}
\end{table}

We analyse the data-lists from season 2021 between the 21 actors using the VSP/QJ-D model. The consensus order for the drivers in this season is shown in \Fig~\ref{fig:f1cpo}. Both Lewis Hamilton and Max Verstappen are ranked at top of the consensus VSP for the 2021 season, with high posterior probability (more than 0.9). 

The posterior distributions for individual parameters and the depth are shown in \Fig~\ref{fig:f1pos}. The effective sample sizes are 567 for $q=P(S)$ and 130 for $p$. The posterior for $P(S)$ concentrates at around 0.5, showing a relatively relaxed ranking relation. The posterior distribution for $p$ concentrates at a lower value at 0.15. This suggests the VSP model relatively accurately represents the strength of each driver-car pairing. The VSP depths are relatively low for this data. We are not observing a ranking as deep as the social hierarchy for witnesses in ``Royal Acta''. 

\begin{minipage}{\linewidth}
    \centering
    \begin{tikzpicture}[thick,scale=1, every node/.style={scale=0.8}]
        \node[draw, circle, minimum width=.7cm] (1) at (-1, 6) {1};
        \node[draw, circle, minimum width=.7cm] (10) at (1, 6) {10};
        \node[draw, circle, minimum width=.7cm] (9) at (0, 4.5) {9};
        \node[draw, circle, minimum width=.7cm] (3) at (-3, 1.5) {3};
        \node[draw, circle, minimum width=.7cm] (6) at (-3, 3) {6};
        \node[draw, circle, minimum width=.7cm] (7) at (-4, 1.5) {7};
        \node[draw, circle, minimum width=.7cm] (5) at (-4, 0) {5};
        \node[draw, circle, minimum width=.7cm] (11) at (-1, 3) {11};
        \node[draw, circle, minimum width=.7cm] (15) at (0, 3) {15};
        \node[draw, circle, minimum width=.7cm] (2) at (0.5, 1.5) {2};
        \node[draw, circle, minimum width=.7cm] (19) at (0.5, 0) {19};
        \node[draw, circle, minimum width=.7cm] (14) at (-.5, 0) {14};
        \node[draw, circle, minimum width=.7cm] (8) at (2, 0) {8};
        \node[draw, circle, minimum width=.7cm] (12) at (3, 0) {12};
        \node[draw, circle, minimum width=.7cm] (16) at (4, 3) {16};
        \node[draw, circle, minimum width=.7cm] (13) at (4, 1.5) {13};
        \node[draw, circle, minimum width=.7cm] (17) at (4, 0) {17};
        \node[draw, circle, minimum width=.7cm] (18) at (4, -1.5) {18};
        \node[draw, circle, minimum width=.7cm] (4) at (0, -1.5) {4};
        \node[draw, circle, minimum width=.7cm] (20) at (0, -3) {20};
        \node[draw, circle, minimum width=.7cm] (21) at (3, -3) {21};
        \begin{pgfonlayer}{bg}
            \draw[-latex,red] (1) -- (9);
            \draw[-latex,red] (10) -- (9);
            \draw[-latex,red] (9) -- (16);
            \draw[-latex] (9) -- (15);
            \draw[-latex] (9) -- (11);
            \draw[-latex] (9) -- (6);
            \draw[-latex] (9) -- (15);
            \draw[-latex] (9) -- (8);
            \draw[-latex] (9) -- (12);
            \draw[-latex] (9) -- (7);
            \draw[-latex] (6) -- (3);
            \draw[-latex] (6) -- (14);
            \draw[-latex] (14) -- (4);
            \draw[-latex,red] (11) -- (3);
            \draw[-latex] (11) -- (14);
            \draw[-latex] (15) -- (14);
            \draw[-latex] (15) -- (2);
            \draw[-latex,red] (16) -- (13);
            \draw[-latex] (13) -- (17);
            \draw[-latex] (17) -- (18);
            \draw[-latex] (3) -- (4);
            \draw[-latex] (2) -- (19);
            \draw[-latex] (19) -- (4);
            \draw[-latex] (8) -- (4);
            \draw[-latex] (5) -- (4);
            \draw[-latex] (7) -- (5);
            \draw[-latex] (4) -- (20);
            \draw[-latex] (4) -- (21);
            \draw[-latex] (12) -- (20);
            \draw[-latex] (12) -- (21);
        \end{pgfonlayer}
    \end{tikzpicture}
    \captionof{figure}{VSP/QJ-D model. Consensus order for Formula 1 (season 2021) data. Significant/strong order relations are indicated by black/red edges respectively.}
    \label{fig:f1cpo}
\end{minipage}

\begin{figure}[!htb]
  \centering
  \includegraphics[width=0.8\linewidth,height=120pt]{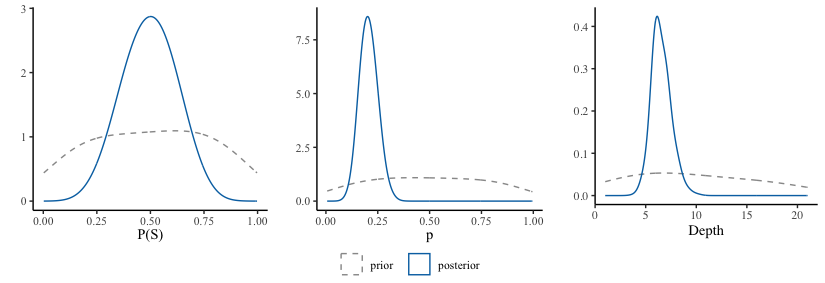}
  \caption{The prior (grey) and posterior (blue) distributions for $P(S)$ (left), $p$ (middle) and depth (right) for the Formula 1 (season 2021) data.}\label{fig:f1pos}
\end{figure}

\pagebreak 

\section{Model Comparison}

\subsection{Model Comparison with Plackett-Luce and Mallows}\label{sec:PL}

The Plackett-Luce model, the Mallows model, and their mixture-models are two categories of model widely used for ranking and partial ranking. In this section, we compare the VSP/QJ-U and VSP/QJ-B models with the two PL-models\footnote{We use the MCMC sampler available in the R-package \texttt{PLmix} \cite{mollica2017}. This uses a data augmentation scheme due to \cite{caron2012efficient}.} and the two Mallows models\footnote{We use the MCMC sampler available in the R-package \texttt{BayesMallows} \cite{bayesmallows}. } using the WAIC. This estimates the expected log pointwise predictive density (ELPD, \cite{vehtari2017practical}). It is a principled criterion for model comparison which is relatively easily estimated. 

The Plackett-Luce model defines a distribution over ranked lists $y_i\in\P_{[n]},\ i\in [N]$ with actor attributes $\lambda=(\lambda_1,\dots,\lambda_n)\in \mathbb{R}^n$. Taking into account the list membership sets $o_i,\ i\in [N]$, the likelihood is
\begin{equation}\label{Plackett-Luce}
    PL(y|\lambda) = \prod_{i=1}^N\prod_{j=1}^{n_i} \frac{e^{\lambda_{y_{i,o_j}}}}{\sum_{k = j}^{n_i} e^{\lambda_{y_{i,o_k}}}}.
\end{equation}

The Plackett-Luce mixture assumes the lists are sampled from a heterogeneous population composed of $D$ sub-populations. Each mixture component has a Plackett-Luce distribution over lists with actor attributes $\lambda^{(d)}\in\mathbb{R}^n,\ d\in [D]$. A finite mixture of Plackett-Luce models was proposed as a robust model for ranked data with incomplete lists in \citet{mollica2017,mollica2020plmix}. Let $\Lambda = (\lambda^{(d)})_{d\in [D]}\in\mathbb{R}^{n\times D}$ give the matrix of actor attributes and $\omega = (\omega_1,\dots,\omega_D)$ give the weights of mixture components with $\sum_{d=1}^D \omega_d = 1$. The $D$-component mixture Plackett-Luce model likelihood is
\begin{equation}
    PL_{mix}(y|\Lambda,\omega) = \prod_{i=1}^N\sum_{d=1}^D \omega_d PL(y_i|\lambda^{(d)}). 
\end{equation}
Non-informative priors suggested by \cite{mollica2020plmix} are assigned with $e^{\lambda_j^{(d)}}\sim \text{Gamma}(1,0.001)$ for $j\in [n]$ and $d\in[D]$ and $\omega_1,\dots,\omega_D\sim Dir(1,\dots,1)$. 

The Mallows model \cite{mallows1957non} is typically controlled by a \textit{location parameter (consensus ranking)} $\rho\in\mathcal{P}_n$ and a \textit{scaling parameter} $\alpha\in (0,\infty)$. Letting $d(\cdot,\cdot):\mathcal{P}_n\times\mathcal{P}_n\to \mathbb{R}_+$ be a \textit{discrepancy function} between two permutations, the Mallows model is
\begin{equation}
    P_d(y|\rho,\alpha) = \prod_{i=1}^N \frac{1}{Z_n(\alpha)}e^{-\frac{\alpha}{n} d(\rho,y_i)},
\end{equation}
where $Z_n(\alpha):=\sum_{y\in \mathcal{P}_n}e^{-\frac{\alpha}{n} d(\rho,y)}$ is the normalising constant. A typical distance choice is the Kendall's tau distance. Let $\sigma(l,a)=\{k\in [n]:l_k=a\}$. The Kendall's tau distance counts the number of pairwise disagreements between two permutations, 
$d(y,l) = \sum_{i<j} \mathbbm{1}_{\sigma(l,y_i)>\sigma(l,y_j)}$, and this gives a tractable normalising constant $Z_n(\alpha)$. We use the Mallows $\phi$ model in our model comparison. A truncated exponential prior is specified for $\alpha$ and a uniform prior $\pi(\rho)$ on $\mathcal{P}_n$ is taken for $\rho$, as is suggested in \cite{bayesmallows} which implements the MCMC proposed in \cite{vitelli2018probabilistic}. The \texttt{BayesMallows} R-package deals with partial ranking by applying data augmentation techniques before fitting the full Mallows model. 

Similar to the Plackett-Luce Mixture, the finite Mallows mixture allows for heterogeneity. Let $\{\rho_d,\alpha_d\}_{d=1,\dots,D}$ be the set of parameters for cluster $d$ and let $z_1,\dots,z_N\in \{1,\dots,D\}$ be the cluster labels that assign each list to one cluster. The $D$-component mixture Mallows likelihood is
\begin{equation}
    P(y|\{\rho_d,\alpha_d\}_{d=1,\dots,D},\{z_i\}_{i=1,\dots,N})=\prod_{i=1}^N \frac{1}{Z_n(\alpha_{z_i})}e^{-\frac{\alpha_{z_i}}{n}d(y_i,\rho_{z_i})}.
\end{equation}
Independent truncated exponential priors  and independent uniform priors are specified for $\alpha$ and $\rho$ respectively. Following \cite{bayesmallows}, $z_1,\dots,z_N$ follow a uniform multinomial distribution and are assumed conditionally independent given the cluster parameters.

The ELPD measures the posterior predictive accuracy of a model. It is a natural choice for goodness-of-fit and model comparison. 
We use the WAIC to estimate the ELPD for a generic model (``A'' say). The estimator resembles the AIC and BIC,
\begin{equation}
    \widehat{elpd}_{waic}(A|y) = \sum_{i=1}^N \log p_A(y_i|y) - p_{waic}, 
\end{equation}
where 
\begin{equation}\label{eq:predp}
    p_A(y_i|y) = \int p_A(y_i|\theta)p_A(\theta|y)d\theta
\end{equation}
with $\theta$ representing all parameter in model $A$. The predictive probability in Eqn. \ref{eq:predp} is estimated using MCMC samples. For a MCMC sample (after burn-in) of length $K$ targeting $p_A(\theta|y)$, 
\[\widehat{p_A}(y_i|y) = \frac{1}{k}\sum_{k\in [K]}p_A(y_i|\theta^{(k)}).\] 
The term $p_{waic}$ is the effective number of parameters. If $V_{k=1}^K a_k = \frac{1}{K-1}\sum_{k=1}^K (a_k - \Bar{a})^2$, then $p_{waic}$ is estimated using $\hat{p}_{waic}=\sum_{i=1}^N V_{k=1}^K(\log(p(y_i|\theta^{(k)})))$. The \texttt{waic} function from R package \texttt{loo} \citep{vehtari2017practical} is used for $elpd_{waic}$ estimation. 

The \texttt{PLmix} package in R provides a range of model selection criterion to select the optimal number of mixture components $D$. We use the Deviation Information Criterion to select the optimal model on a given data. Similar model selection procedures are implemented for the Mallows model. 

\subsubsection{Model comparison on the `Royal Acta' Data}

Table \ref{tab:waic} summarises the estimated $elpd_{waic}$ for the six models, on three signature dataset - 1080-1084, 1126-1130 and 1134-1138(b) (5PLA). The VSP/QJ models outperforms the PL, PL-mixture, Mallows and Mallows moxture models significantly in all time periods.  The VSP/QJ-B model is relatively favourable compared to VSP/QJ-U. We note that we made no careful choice of priors on the PL models and the Mallows models. Non-informative priors are adapted in both cases so it is possible the performance of these models could be improved. However, they have a long way to go to catch up.

%
%



\begin{table}[!htb]
    \centering
    \small
    \caption{The estimated $elpd_{waic}$ ($se$) under six different models - VSP/QJ-U, VSP/QJ-B, Plackett-Luce (PL) and 2-mixture Plackett-Luce (PL-Mix2) model.}\label{tab:waic}
    \begin{tabular}{*5c}
      \toprule 
      {} & \multicolumn{3}{c}{\bfseries $elpd_{waic}$ ($se$)}\\
      \midrule
      \bfseries Model & \bfseries 1080-1084 & \bfseries 1126-1130 & \bfseries 1134-1138(b)\\
      \midrule 
      VSP/QJ-B & -103.5 (26.0) & -28.6 (9.6)& -72.2 (21.9) \\
      VSP/QJ-U & -197.2 (77.8) & -37.8 (10.8) &  -86.3 (27.6)\\
      PL & -316.5 (38.5) & -270.4 (25.8) & -336.2 (35.6)\\
      PL-Mix2 & -291.1 (37.2) & -267.6 (24.7)& -318.6 (36.3)\\
      Mallows & -601.9 (6.8) & -624.5 (3.0) & -770.2 (7.6) \\
      Mallows-Mix & -613.9 (4.1) (D=4) & -604.7 (1.9) (D=6)  & -820.7 (4.8) (D=4) \\
      \bottomrule 
    \end{tabular}
\end{table}

We estimate consensus orders for both the PL and PL-Mixture models. This is done by first sampling from the posterior distribution of ranking(s). We turn the rankings into partial order representations. For a PL-mixture, we calculate the intersection order that records the order relation appearing in all rankings. The consensus order is then constructed from this `posterior distribution of partial orders'. The estimated consensus orders for the PL and PL-Mixture (D=2) models are shown in Figure~\ref{fig:con-order-pl}.

\begin{minipage}{\linewidth}
    \centering
    \begin{tikzpicture}[thick,scale=.3, every node/.style={scale=0.3}]
        \node[draw, circle, minimum width=.7cm,fill=Magenta] (1) at (-3, 13) {};
        \node[draw, circle, minimum width=.7cm,fill=Magenta] (2) at (-3, 12) {};
        \node[draw, circle, minimum width=.7cm,fill=Gray] (3) at (-3, 11) {};
        \node[draw, circle, minimum width=.7cm,fill=Magenta] (4) at (-3, 10) {};
        \node[draw, circle, minimum width=.7cm,fill=Magenta] (5) at (-3, 9) {};
        \node[draw, circle, minimum width=.7cm,fill=Magenta] (6) at (-3, 8) {};
        \node[draw, circle, minimum width=.7cm,fill=Gray] (7) at (-3, 7) {};
        \node[draw, circle, minimum width=.7cm,fill=SkyBlue] (8) at (-3, 6) {};
        \node[draw, circle, minimum width=.7cm,fill=Gray] (9) at (-3, 5) {};
        \node[draw, circle, minimum width=.7cm,fill=Gray] (10) at (-3,4) {};
        \node[draw, circle, minimum width=.7cm,fill=Maroon] (11) at (-3, 3) {};
        \node[draw, circle, minimum width=.7cm,fill=Gray] (12) at (-3, 2) {};
        \node[draw, circle, minimum width=.7cm,fill=Apricot] (13) at (-3, 1) {};
        \begin{pgfonlayer}{bg}
            \draw[-latex] (1) -- (2);
            \draw[-latex,red] (2) -- (3);
            \draw[-latex] (3) -- (4);
            \draw[-latex] (4) -- (5);
            \draw[-latex] (5) -- (6);
            \draw[-latex] (6) -- (7);
            \draw[-latex] (7) -- (8);
            \draw[-latex] (8) -- (9);
            \draw[-latex] (9) -- (10);
            \draw[-latex] (10) -- (11);
            \draw[-latex] (11) -- (12);
            \draw[-latex] (12) -- (13);
        \end{pgfonlayer}
        \node[draw, circle, minimum width=.7cm,fill=Magenta] (14) at (5, 13) {};
        \node[draw, circle, minimum width=.7cm,fill=Magenta] (15) at (6, 13) {};
        \node[draw, circle, minimum width=.7cm,fill=Magenta] (16) at (4.5, 7) {};
        \node[draw, circle, minimum width=.7cm,fill=Magenta] (17) at (6.5, 7) {};
        \node[draw, circle, minimum width=.7cm,fill=Gray] (18) at (8.5, 7) {};
        \node[draw, circle, minimum width=.7cm,fill=Magenta] (19) at (3, 1) {};
        \node[draw, circle, minimum width=.7cm,fill=Gray] (20) at (4, 1) {};
        \node[draw, circle, minimum width=.7cm,fill=Gray] (21) at (5, 1) {};
        \node[draw, circle, minimum width=.7cm,fill=Maroon] (22) at (6, 1) {};
        \node[draw, circle, minimum width=.7cm,fill=SkyBlue] (23) at (7, 1) {};
        \node[draw, circle, minimum width=.7cm,fill=Apricot] (24) at (8, 1) {};
        \node[draw, circle, minimum width=.7cm,fill=Gray] (25) at (9, 1) {};
        \node[draw, circle, minimum width=.7cm,fill=Gray] (26) at (10, 1) {};
        \begin{pgfonlayer}{bg}
            \draw[-latex] (14) -- (16);
            \draw[-latex] (14) -- (17);
            \draw[-latex] (15) -- (16);
            \draw[-latex] (15) -- (17);
            \draw[-latex] (14) -- (19);
            \draw[-latex] (15) -- (19);
            \draw[-latex] (16) -- (20);
            \draw[-latex] (17) -- (20);
            \draw[-latex] (14) -- (21);
            \draw[-latex] (15) -- (21);
            \draw[-latex] (14) -- (22);
            \draw[-latex] (15) -- (22);
            \draw[-latex] (15) -- (23);
            \draw[-latex] (14) -- (24);
            \draw[-latex] (15) -- (24);
            \draw[-latex] (18) -- (24);
            \draw[-latex] (14) -- (25);
            \draw[-latex] (15) -- (25);
            \draw[-latex] (18) -- (25);
            \draw[-latex] (18) -- (26);
            \draw[-latex] (15) -- (26);
        \end{pgfonlayer}
        \matrix [draw,below left] at (current bounding box.north east) {
          \node [circle,fill=Maroon,label=right:Archbishop] {}; \\
          \node [circle,fill=Magenta,label=right:Bishop] {}; \\
          \node [circle,fill=SkyBlue,label=right:Earl] {}; \\
          \node [circle,fill=Apricot,label=right:Count] {}; \\
          \node [circle,fill=Gray,label=right:Other] {}; \\
        };     
    \end{tikzpicture}
    \captionof{figure}{The estimated consensus orders from the Plackett-Luce (left) and PL-Mixture (D=2) (right) models on the 1126-1130 data. Red edges indicate order relations that posterior probabilities are higher than 0.9. }\label{fig:con-order-pl}
\end{minipage}

Both the PL and PL-Mixture (D=2) model are not designed to reconstruct partial orders in the way we use it here. It was of interest to see if they did capture the same or similar relations to those we find with VSP models. This is not the case. Although we don't know the true partial order, we do expect a fairly deep social hierarchy in the 12th century. Neither model reflects such a feature. 


\subsubsection{Model comparison on the Formula 1 Race Data}

We compare the VSP/QJ-D model with the Placket-Luce and Mallows model, and their mixtures on the Formula 1 dataset. The comparison result using $elpd_{waic}$ is shown in table \ref{tab:waic-f1}. The VSP/QJ-D model outperforms both the Plackett-Luce, the Mallows and their mixtures significantly. 

\begin{table}[h]
    \centering
    \small
    \caption{The estimated $elpd_{waic}$ ($se$) under five different models for the Formula 1 Racing Data - VSP/QJ-D, Plackett-Luce (PL) and 2-mixture Plackett-Luce (PL-Mix2), Mallows and 3-Mixture Mallows (Mallows-Mix3) model.}\label{tab:waic-f1}
    \begin{tabular}{*2c}
      \toprule 
      \bfseries Model & {\bfseries $elpd_{waic}$ ($se$)}\\
      \midrule 
      VSP/QJ-D & -597.1 (25.2) \\
      PL & -847.4 (18.6) \\
      PL-Mix2 & -821.6 (17.4) \\
      Mallows & -973.7 (3.4) \\
      Mallows-Mix3 & -963.5 (3.9) \\
      \bottomrule 
    \end{tabular}
\end{table}

\subsection{Model comparison VSP v. Bucket order}\label{sec:bucket-BF}

Bayes factors $B_{01}$ for bucket orders (see Section~1 
over VSPs can be estimated using the Savage-Dickey Ratio. 
Results are summarized in Table \ref{tab:bf} for both models QJ-U and QJ-B and both 1LPA and 5LPA datasets. Numbers above one support bucket orders. Numbers below one support VSPs. For 1PLA dataset, we observe strong support for VSPs. For 5LPA data there is a very slight preference for bucket orders ``barely worth mentioning'' over QJ-B. Presumably the extra model complexity of QJ-B is costing something here. For QJ-U and the period 1180-84 there is no strong preference - the consensus order in \Fig~7 
is ``nearly'' a bucket order. However, for QJ-U, 1126-30 and 1134-38 and 1134-38(b) the consensus orders are more complex and VSP's are strongly preferred over Bucket orders.

\begin{table}[!htb]
    \centering
    \small
    \caption{The Bayes factors $B_{01}$ for `bucket' order over VSP on all datasets 5LPA (Left) and 1LPA (Right).}\label{tab:bf}
    \begin{tabular}{*3c}
      \toprule 
      {} & \multicolumn{2}{c}{\bfseries Bayes Factor $B_{01}$}\\
      \midrule
      \bfseries Dataset & \bfseries VSP/QJ-U & \bfseries VSP/QJ-B\\
      \midrule 
      1080-1084 & 1.73 & 2.83\\
      1126-1130 & 0.18 & 2.83\\
      1134-1138 & 0.00 & NA\\
      1134-1138(b) & 0.33 & 2.59\\
      \bottomrule 
    \end{tabular}
    \qquad
    \begin{tabular}{c c}
     \toprule 
     {} &{\bfseries Bayes Factor $B_{01}$} \\ \midrule
      \bfseries Dataset & \bfseries VSP/QJ-U\\ \midrule
       & \\
      1080-1084 & 0.00 \\ 
       & \\
       1134-1138&0.00 \\
      \bottomrule 
    \end{tabular}
\end{table}

\subsection{Model comparison with the latent Partial Order Model}

\cite{nicholls122011partial} proposes a latent partial order model, which can be applied to fit general partial orders to rank-order list-data. Though their method is not scalable to datasets of more than around 20 actors, we are interested in comparing the performance between their partial order (PO) model and the VSP class of models proposed in this paper. We choose the same observation model, QJ-U, to make the test. We choose a relatively small dataset, 1126-1130 with 5LPA, for this comparison, so the full PO model is tractable. We chose priors $\rho\sim Beta(1,\frac{1}{6})$ as suggested in \cite{nicholls122011partial} and a non-informative prior for the error probability $p=\frac{e^r}{1+e^r}$ where $r\sim \mathcal{N}(0,1.5)$ in order to get a reasonably flat depth distribution for the PO-prior. 

The consensus order from the PO/QJ-U model is shown in \Fig~\ref{fig:26-30pocon} (left). We also copy the result from the VSP/QJ-U model here for comparison. The two models indicates similar social hierarchy. However, the PO/QJ-U model presents a less strict hierarchy among bishops. 

\begin{minipage}{\linewidth}
    \centering
    \begin{tikzpicture}[thick,scale=.3, every node/.style={scale=0.3}]
        \node[draw, circle, minimum width=.7cm,fill=Maroon] (1) at (0, 5) {};
        \node[draw, circle, minimum width=.7cm,fill=Magenta] (2) at (0, 4) {};
        \node[draw, circle, minimum width=.7cm,fill=Magenta] (3) at (0, 3) {};
        \node[draw, circle, minimum width=.7cm,fill=Magenta] (4) at (-1.2, 2) {};
        \node[draw, circle, minimum width=.7cm,fill=Magenta] (5) at (1.2, 2) {};
        \node[draw, circle, minimum width=.7cm,fill=Magenta] (6) at (0, 0) {};
        \node[draw, circle, minimum width=.7cm,fill=SkyBlue] (7) at (-2.25, 1.1) {};
        \node[draw, circle, minimum width=.7cm,fill=Gray] (8) at (-1.5, 0) {};
        \node[draw, circle, minimum width=.7cm,fill=Gray] (9) at (-3, 0) {};
        \node[draw, circle, minimum width=.7cm,fill=Apricot] (10) at (3, 0) {};
        \node[draw, circle, minimum width=.7cm,fill=Gray] (11) at (0, -1) {};
        \node[draw, circle, minimum width=.7cm,fill=Gray] (12) at (0, -2) {};
        \node[draw, circle, minimum width=.7cm,fill=Gray] (13) at (0, -3) {};
        \begin{pgfonlayer}{bg}
            \draw[-latex,red] (1) -- (2);
            \draw[-latex,red] (2) -- (3);
            \draw[-latex,red] (3) -- (5);
            \draw[-latex] (3) -- (4);
            \draw[-latex,red] (4) -- (6);
            \draw[-latex,red] (5) -- (6);
            \draw[-latex,red] (4) -- (7);
            \draw[-latex] (5) -- (7);
            \draw[-latex,red] (4) -- (10);
            \draw[-latex,red] (5) -- (10);
            \draw[-latex] (7) -- (8);
            \draw[-latex] (7) -- (9);
            \draw[-latex] (9) -- (11);
            \draw[-latex] (8) -- (11);
            \draw[-latex] (10) -- (11);
            \draw[-latex] (6) -- (11);
            \draw[-latex] (11) -- (12);
            \draw[-latex] (12) -- (13);
        \end{pgfonlayer}

        \node[draw, circle, minimum width=.7cm,fill=Maroon] (81) at (10, 5) {};
        \node[draw, circle, minimum width=.7cm,fill=Magenta] (82) at (10, 4.1) {};
        \node[draw, circle, minimum width=.7cm,fill=Magenta] (83) at (10, 3.2) {};
        \node[draw, circle, minimum width=.7cm,fill=Magenta] (84) at (10, 2.3) {};
        \node[draw, circle, minimum width=.7cm,fill=Magenta] (85) at (10, 1.4) {};
        \node[draw, circle, minimum width=.7cm,fill=SkyBlue] (86) at (7, 0.2) {};
        \node[draw, circle, minimum width=.7cm,fill=Gray] (87) at (7, -.7) {};
        \node[draw, circle, minimum width=.7cm,fill=Gray] (88) at (7, -1.6) {};
        \node[draw, circle, minimum width=.7cm,fill=Gray] (89) at (8, -2.5) {};
        \node[draw, circle, minimum width=.7cm,fill=Apricot] (90) at (9, -1.1) {};
        \node[draw, circle, minimum width=.7cm,fill=Magenta] (91) at (11, -.7) {};
        \node[draw, circle, minimum width=.7cm,fill=Gray] (92) at (13, -.7) {};
        \node[draw, circle, minimum width=.7cm,fill=Gray] (93) at (11, -3) {};
        \begin{pgfonlayer}{bg}
            \draw[-latex,red] (81) -- (82);
            \draw[-latex,red] (82) -- (83);
            \draw[-latex,red] (83) -- (84);
            \draw[-latex] (84) -- (85);
            \draw[-latex,red] (85) -- (86);
            \draw[-latex] (86) -- (87);
            \draw[-latex] (87) -- (88);
            \draw[-latex] (88) -- (89);
            \draw[-latex] (90) -- (89);
            \draw[-latex,red] (85) -- (90);
            \draw[-latex] (90) -- (89);
            \draw[-latex,red] (85) -- (91);
            \draw[-latex,red] (92) -- (93);
            \draw[-latex,red] (85) -- (92);
            \draw[-latex] (91) -- (93);
            \draw[-latex] (89) -- (93);
        \end{pgfonlayer}
        
    \end{tikzpicture}
    \qquad
    \begin{tikzpicture}[scale=0.6,every node/.style={scale=0.4}]
        \matrix [draw, left]{
          \node [circle,fill=Maroon,label=right:Archbishop] {}; \\
          \node [circle,fill=Magenta,label=right:Bishop] {}; \\
          \node [circle,fill=SkyBlue,label=right:Earl] {}; \\
          \node [circle,fill=Apricot,label=right:Count] {}; \\
          \node [circle,fill=Gray,label=right:Other] {}; \\
        };
    \end{tikzpicture}
    
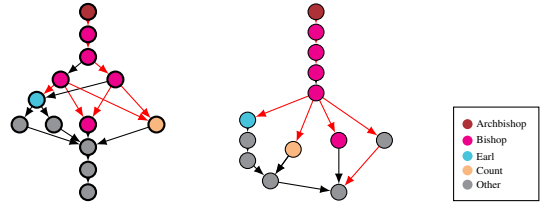
\captionof{figure}{PO/QJ-U model(left) and VSP/QJ-U model (right; same as \Fig~7). 
    Consensus order for 1126-1130 5LPA data. Significant/strong order relations are indicated by black/red edges respectively.}
    \label{fig:26-30pocon}
\end{minipage}

The consensus order from the PO/QJ-U model is actually a VSP. \Fig~\ref{fig:depth-po} shows the prior and posterior depth distributions for both the PO/QJ-U and VSP/QJ-U models. Although the prior distributions over depth are all relatively flat for the two models, the PO/QJ-U model favour partial orders with relatively lower depth. 

\begin{figure}[!htb]
  \centering
  \includegraphics[width=.6\linewidth]{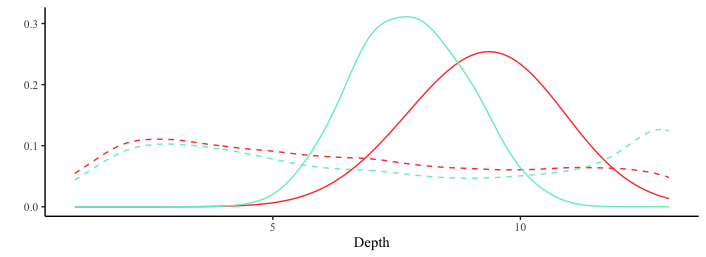}
  \caption{The prior (dashed) and posterior (solid) distribution over depth for the PO/QJ-U (green) and VSP/QJ-U model (red). }\label{fig:depth-po}
\end{figure}

The posterior probability to get a VSP given the PO/QJ-U model is $p_{PO/QJ-U}(h\in \mathcal{V}_{[n]}|\mathbf{y}) = 0.31$ so there is a reasonable chance in the more general model that the unknown true social hierarchy is a VSP. The model comparison performed in Table \ref{tab:waic-po} indicates similar $elpd_{waic}$ for both models. Considering the uncertainty in our estimation, we conclude both models fitting the data equally well. 

\begin{table}[!htb]
    \centering
    \small
    \caption{The estimated $elpd_{waic}$ ($se$) for the VSP/QJ-U and PO/QJ-U models.}\label{tab:waic-po}
    \begin{tabular}{*2c}
      \toprule 
      \bfseries Model & {\bfseries $elpd_{waic}$ ($se$)}\\
      \midrule 
      VSP/QJ-U &  -37.8 (10.8)\\
      PO/QJ-U & -36.7 (10.1) \\
      \bottomrule 
    \end{tabular}
\end{table}

We compare the average ranking for different professions in table \ref{tab:exp-rank-po} and observe the same ranking order in professions although ranking scales are slight different.

\begin{table}[!htb]
    \centering
    \caption{The professions and their average rankings under the PO/QJ-U and VSP/QJ-U models for time period 1126-1130. }\label{tab:exp-rank-po}
    \begin{tabular}{*3c}
      \toprule 
      {} & \multicolumn{2}{c}{\bfseries Average Rank}\\
      \midrule
      \bfseries Profession & PO/QJ-U & VSP/QJ-U\\
      \midrule 
      Archbishop & 1 (0.08) & 1 (0.08) \\
      Bishop & 3.76 (0.29) & 3.99 (0.31) \\
      Earl & 5.75 (0.44) & 6.93 (0.53) \\
      Count & 6.04 (0.46) & 8.94 (0.69) \\
      Other & 9.28 (0.71) & 10.40 (0.80)\\
      \bottomrule 
    \end{tabular}
\end{table}

We summarise the posterior distributions over POs/VSPs using the consensus adjacency matrix $m$, such that $$m_{i,j}=p(i \succ j|\mathbf{y}), i,j \in [n].$$ The consensus orders are inferred from the consensus adjacency matrix by setting a certain threshold. This paper chooses a threshold of 0.5. \Fig~\ref{fig:conpocomp} plots the entries of the two consensus adjacency matrices against each other. The points roughly scatter along the reference line $y=x$, and show a positive monotone trend. Based on \Fig~\ref{fig:conpocomp}, the two consensus adjacency matrices roughly agree with each other, highlighting the fact that although the VSP is a more restricted model, it works as well as a flexible and scalable partial order model in social hierarchy scenarios. 

\begin{figure}[!htb]
  \centering
  \includegraphics[width=.35\linewidth]{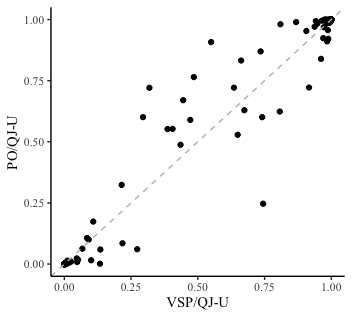}
  \caption{The comparison plot between the consensus adjacency matrices from the VSP/QJ-U (x-axis) and PO/QJ-U (y-axis) models. The gray dashed line is the $y=x$ reference line. }\label{fig:conpocomp}
\end{figure}

\pagebreak

\section{Scaling analysis}\label{sec:runtime}

Counting the number of linear extensions of a general partial order is known to be \#P-complete (\cite{brightwell1991counting}). \textit{LEcount} by \cite{kangas2016counting} seems to be the most computationally efficient counting tool available. \textit{LEcount} chooses between two algorithms, one counts by recursion in $O(2^n n)$ operations and the other by variable elimination in $O(n^{t+4})$ where $t$ is the treewidth of the cover graph. The linear-extension counting algorithm we use exploits the tree representation (1, 2) 
so it only works for VSPs, but it is more reliable and faster than \textit{LEcount} especially for the complicated and large VSPs at the right end of \Fig~\ref{fig:runtime1}.

The likelihood evaluation involves substantial computation of the number of linear extensions, and is an essential part of our MCMC analysis. We compare the computational cost to the likelihood evaluations under either the VSP tree representation or \textit{LEcount}. This is done by simulating $N=20$ full length lists on VSPs of increasing size $n=3,6,...,39$ from our VSP prior. For each group of $N$ lists we evaluate the likelihood for the VSP used in simulation. We repeat this 50 times for each VSP size $n$ for each method to derive an estimated distribution over run-times. The log-scaled maximum run-time (in seconds) for each sample size is shown in \Fig~\ref{fig:runtime1}. The log-scaled maximum run-time appears to be linear for the tree representation and  exponential for \textit{LEcount}. The optimised \textit{LEcount} approach outperforms the tree representation LE evaluation when we have VSPs less than 25 actors. However, VSP-based counting significantly outperforms \textit{LEcount} when we move to much larger datasets (completely as expected, all that matters is that we are comparing a simple implementation of a fast VSP algorithm with a well optimised implementation of a PO algorithm and the simple VSP implementation still beats the optimised PO implementation at large enough VSP sizes because the VSP algorithm only works for a subset of POs, so there is no criticism of LEcount here). 

\begin{figure}[H]
  \centering
  \includegraphics[width=.4\linewidth]{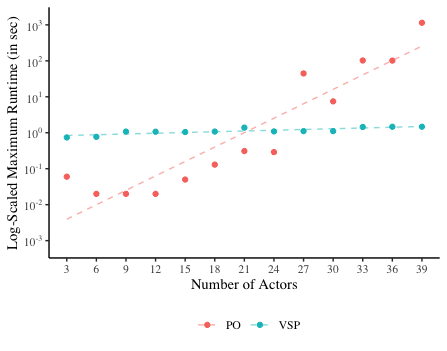}
  \caption{Run-time analysis between the count approach from tree representation and \textit{LEcount} (\cite{kangas2016counting}) on VSPs. The plot compares likelihood (QJ-U) evaluation exploiting the VSP structure (in green) and for a general PO (in red). The log-scaled maximum run-time (in seconds) from the tree representation (green) and the \textit{LEcount} is shown in y-axis, and the number of actors in VSP is shown in the x-axis. 
  }\label{fig:runtime1}
\end{figure}

The scaling analysis demonstrates the high scalability of the VSP counting method via the tree representation. This enables our model to work on datasets with more than 200 actors, see Section~\ref{sec:1lpa}.

\section{Detecting VSP's}

\cite{valdes1979recognition} proposes an efficient way to recognise VSP's by detecting the so-called \textit{forbidden sub-graph} (\Fig~\ref{forbidden_subgraph}).

\begin{minipage}{\linewidth}
    \centering
    \begin{tikzpicture}[thick,scale=.7, every node/.style={scale=0.8}]
        \node[draw, circle, minimum width=.7cm,fill=red] (1) at (0, 1.5) {$1$};
        \node[draw, circle, minimum width=.7cm,fill=red] (2) at (1.5, 1.5) {$2$};
        \node[draw, circle, minimum width=.7cm,fill=red] (3) at (0, 0) {$3$};
        \node[draw, circle, minimum width=.7cm,fill=red] (4) at (1.5, 0) {$4$};
        \draw[-latex] (1) -- (3);
        \draw[-latex] (1) -- (4);
        \draw[-latex] (2) -- (4);
    \end{tikzpicture}
    
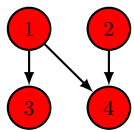
\captionof{figure}{The `forbidden sub-graph' to the VSP class of partial orders.}\label{forbidden_subgraph}
\end{minipage}

A partial order $h\in\mathcal{H}_{[n]}$ is a VSP if it does not contain a set of vertices $o=\{j_1,\dots,j_4\}\subset [n]$ with sub-graph $h=h[o]$ that is isomorphic to the `forbidden sub-graph' $F=([4],\{\langle 1,3\rangle,\langle 1,4\rangle,\langle 2,4\rangle\})$. If two graphs are isomorphic, $F$ and $h'$ in our case, they must be identical after vertex relabelling. This means edges absent in $F$ must also be absent in $h'$. This makes it straightforward to test if a partial order is a VSP.



\section{PRIOR DISTRIBUTION ON DEPTH}

Our VSP-prior gives good control over partial order depth. We can choose the prior distribution over $q$ so that the marginal distribution $\pi_{\V_{[n]}}(v)$ has a reasonably flat distribution over the depth $D(v)$ of the VSP-partial order $v$. This ensures the prior is non-informative with respect to partial-order depth, a property of a social hierarchy on actors which is of particular interest. After some experimentation we found that taking $\eta\sim \mathcal{N}(1,1.5)$ and setting $q=\frac{1}{1+e^{-\eta}}$ gave a reasonably non-informative depth distribution. \Fig~\ref{fig:depth-prior} shows an example prior depth distribution for partial orders with 50 actors under this prior. 

\begin{figure}[!htb]
  \centering
  \includegraphics[width=.6\linewidth]{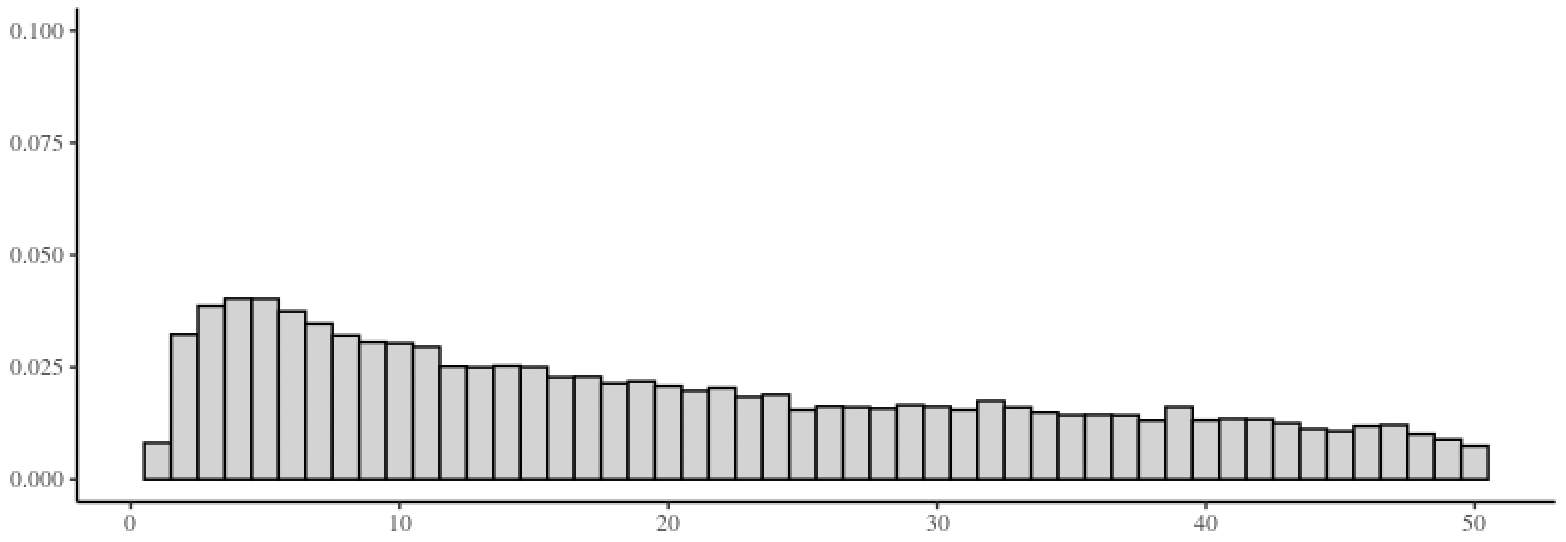}
  \caption{The prior distribution over depth for partial orders with 50 actors, when $q=\frac{1}{1+e^{-\eta}},\eta\sim\mathcal{N}(1,1.5)$. }\label{fig:depth-prior}
\end{figure}


